\documentclass[%
 aps,
 superscriptaddress,
 prx,
 amsmath,amssymb,
 longbibliography,
 reprint,
]{revtex4-2}
\usepackage{dcolumn}% Align table columns on decimal point
\usepackage{bm}% bold math
\usepackage[version=3]{mhchem} 
\usepackage{mathrsfs}
\usepackage{float}
\usepackage{url}
\usepackage{epstopdf}
\usepackage{graphicx}
\usepackage{booktabs}
\usepackage{latexsym}
\usepackage[caption=false]{subfig}
\usepackage{xcolor}
\usepackage[inline]{enumitem}
\usepackage{braket}
\usepackage[unicode=true,pdfusetitle,
 bookmarks=true,bookmarksnumbered=false,bookmarksopen=false,
 breaklinks=false,pdfborder={0 0 0},backref=false,colorlinks=true,
 allcolors=blue]
 {hyperref}
\usepackage[capitalise]{cleveref}
\usepackage{mathtools}
\usepackage{complexity}
\usepackage{amsthm}
\usepackage[inline]{enumitem}
\usepackage{relsize}
\usepackage[ruled, linesnumbered]{algorithm2e}

\newcommand{\add}[1]{{\color{black}#1}} 
\newtheorem{theorem}{Theorem}
\newtheorem{lemma}[theorem]{Lemma}
\newtheorem{definition}[theorem]{Definition}

\DeclarePairedDelimiter{\norm}{\lVert}{\rVert}

\newcommand{\rvline}{\hspace*{-\arraycolsep}\vline\hspace*{-\arraycolsep}}

\usepackage{tikz}
\usepackage{ifthen}
\usepackage{xstring}
\usetikzlibrary{matrix, positioning}
\usepackage{xcolor}
\usetikzlibrary{patterns}

\definecolor{mpscolor}{RGB}{190,169,213}
\definecolor{gatecolor}{RGB}{245,178,151}

\newenvironment{diagram}
{
\begin{tikzpicture}[baseline = (X.base),every node/.style={scale=1.},scale=0.7]
}
{
\end{tikzpicture}
}

\newenvironment{diagram_05}
{
\begin{tikzpicture}[baseline = (X.base),every node/.style={scale=0.7},scale=0.4]
}
{
\end{tikzpicture}
}

\newenvironment{diagram_08}
{
\begin{tikzpicture}[baseline = (X.base),every node/.style={scale=0.7},scale=0.6]
}
{
\end{tikzpicture}
}

\usetikzlibrary{shapes.geometric}
\usetikzlibrary{decorations.pathreplacing,calligraphy}

\newcommand{\MpsRect}[4]{
\draw[fill=white] (#1-0.5, #2+0.5) rectangle (#1+0.5, #2-0.5);
\draw (#1, #2) node {#3};
\IfSubStr{#4}{u}{ \draw (#1, #2+0.5) -- (#1, #2+1);  }{  }
\IfSubStr{#4}{d}{ \draw (#1, #2-0.5) -- (#1, #2-1); }{  }
\IfSubStr{#4}{l}{ \draw (#1-0.5, #2) -- (#1-1, #2); }{  }
\IfSubStr{#4}{r}{ \draw (#1+0.5, #2) -- (#1+1, #2); }{  }
}

\newcommand{\MpsRectBig}[4]{
\draw[fill=white] (#1-0.7, #2+0.7) rectangle (#1+0.7, #2-0.7);
\draw (#1, #2) node {#3};
\IfSubStr{#4}{u}{ \draw (#1, #2+0.7) -- (#1, #2+1.4);  }{  }
\IfSubStr{#4}{d}{ \draw (#1, #2-0.7) -- (#1, #2-1.4); }{  }
\IfSubStr{#4}{l}{ \draw (#1-0.7, #2) -- (#1-1.2, #2); }{  }
\IfSubStr{#4}{r}{ \draw (#1+0.7, #2) -- (#1+1.2, #2); }{  }
}

\newcommand{\MpsRectRed}[4]{
\draw[fill=red!30] (#1-0.5, #2+0.5) rectangle (#1+0.5, #2-0.5);
\draw (#1, #2) node {#3};
\IfSubStr{#4}{u}{ \draw (#1, #2+0.5) -- (#1, #2+1);  }{  }
\IfSubStr{#4}{d}{ \draw (#1, #2-0.5) -- (#1, #2-1); }{  }
\IfSubStr{#4}{l}{ \draw (#1-0.5, #2) -- (#1-1, #2); }{  }
\IfSubStr{#4}{r}{ \draw (#1+0.5, #2) -- (#1+1, #2); }{  }
}

\newcommand{\MpsRectGreen}[4]{
\draw[fill=white] (#1-0.5, #2+0.5) rectangle (#1+0.5, #2-0.5);
\draw (#1, #2) node {#3};
\IfSubStr{#4}{u}{ \draw (#1, #2+0.5) -- (#1, #2+1);  }{  }
\IfSubStr{#4}{d}{ \draw (#1, #2-0.5) -- (#1, #2-1); }{  }
\IfSubStr{#4}{l}{ \draw (#1-0.5, #2) -- (#1-1, #2); }{  }
\IfSubStr{#4}{r}{ \draw (#1+0.5, #2) -- (#1+1, #2); }{  }
}

\newcommand{\GivensGate}[4]{
\draw[fill=gatecolor] (#1, #2) rectangle (#1+2.5, #2+1);
\draw (#1, #2) node {#3};
\IfSubStr{#4}{u}{ \draw (#1, #2+0.5) -- (#1, #2+1);  }{  }
\IfSubStr{#4}{d}{ \draw (#1, #2-0.5) -- (#1, #2-1); }{  }
\IfSubStr{#4}{l}{ \draw (#1-0.5, #2) -- (#1-1, #2); }{  }
\IfSubStr{#4}{r}{ \draw (#1+0.5, #2) -- (#1+1, #2); }{  }
}

\newcommand{\MpsCircle}[4]{
\draw[fill=white] (#1, #2) circle (.5);
\draw (#1, #2) node {#3};
\IfSubStr{#4}{u}{ \draw (#1, #2+0.5) -- (#1, #2+1);  }{  }
\IfSubStr{#4}{d}{ \draw (#1, #2-0.5) -- (#1, #2-1); }{  }
\IfSubStr{#4}{l}{ \draw (#1-0.5, #2) -- (#1-1, #2); }{  }
\IfSubStr{#4}{r}{ \draw (#1+0.5, #2) -- (#1+1, #2); }{  }
}

\newcommand{\MpsWalkerCirclePattern}[4]{
\draw[fill=mpscolor] (#1, #2) circle (.5);
\draw (#1, #2) node {#3};

    \tikzset{mypattern/.style={pattern=north east lines, pattern color=white}}
    
    \begin{scope}
        \tikzset{clip}
        \draw[mypattern] (#1, #2) circle (.5);
    \end{scope}
    
\IfSubStr{#4}{u}{ \draw (#1, #2+0.5) -- (#1, #2+1);  }{  }
\IfSubStr{#4}{d}{ \draw (#1, #2-0.5) -- (#1, #2-1); }{  }
\IfSubStr{#4}{l}{ \draw (#1-0.5, #2) -- (#1-1, #2); }{  }
\IfSubStr{#4}{r}{ \draw (#1+0.5, #2) -- (#1+1, #2); }{  }
}

\newcommand{\MpsWalkerCircle}[4]{
\draw[fill=mpscolor] (#1, #2) circle (.5);
\draw (#1, #2) node {#3};
\IfSubStr{#4}{u}{ \draw (#1, #2+0.5) -- (#1, #2+1);  }{  }
\IfSubStr{#4}{d}{ \draw (#1, #2-0.5) -- (#1, #2-1); }{  }
\IfSubStr{#4}{l}{ \draw (#1-0.5, #2) -- (#1-1, #2); }{  }
\IfSubStr{#4}{r}{ \draw (#1+0.5, #2) -- (#1+1, #2); }{  }
}

\newcommand{\MpsWalkerCircleOrange}[4]{
\draw[fill=gatecolor] (#1, #2) circle (.5);
\draw (#1, #2) node {#3};
\IfSubStr{#4}{u}{ \draw (#1, #2+0.5) -- (#1, #2+1);  }{  }
\IfSubStr{#4}{d}{ \draw (#1, #2-0.5) -- (#1, #2-1); }{  }
\IfSubStr{#4}{l}{ \draw (#1-0.5, #2) -- (#1-1, #2); }{  }
\IfSubStr{#4}{r}{ \draw (#1+0.5, #2) -- (#1+1, #2); }{  }
}

\newcommand{\MpsCircleGR}[4]{
\draw[fill=white] (#1, #2) circle (.8);
\draw (#1, #2) node {#3};
\IfSubStr{#4}{u}{ \draw (#1, #2+0.5) -- (#1, #2+1);  }{  }
\IfSubStr{#4}{d}{ \draw (#1, #2-0.5) -- (#1, #2-1); }{  }
\IfSubStr{#4}{l}{ \draw (#1-0.5, #2) -- (#1-1, #2); }{  }
\IfSubStr{#4}{r}{ \draw (#1+0.5, #2) -- (#1+1, #2); }{  }
}

\newcommand{\MpsCircleRed}[4]{
\draw[fill=red!30] (#1, #2) circle (.5);
\draw (#1, #2) node {#3};
\IfSubStr{#4}{u}{ \draw (#1, #2+0.5) -- (#1, #2+1);  }{  }
\IfSubStr{#4}{d}{ \draw (#1, #2-0.5) -- (#1, #2-1); }{  }
\IfSubStr{#4}{l}{ \draw (#1-0.5, #2) -- (#1-1, #2); }{  }
\IfSubStr{#4}{r}{ \draw (#1+0.5, #2) -- (#1+1, #2); }{  }
}

\newcommand{\MpsCircleBlue}[4]{
\draw[fill=blue!30] (#1, #2) circle (.5);
\draw (#1, #2) node {#3};
\IfSubStr{#4}{u}{ \draw (#1, #2+0.5) -- (#1, #2+1);  }{  }
\IfSubStr{#4}{d}{ \draw (#1, #2-0.5) -- (#1, #2-1); }{  }
\IfSubStr{#4}{l}{ \draw (#1-0.5, #2) -- (#1-1, #2); }{  }
\IfSubStr{#4}{r}{ \draw (#1+0.5, #2) -- (#1+1, #2); }{  }
}

\newcommand{\MpsRectQP}[4]{
\draw (#1-0.5, #2+0.5) rectangle (#1+0.5, #2-0.5);
\fill [red!30]    (#1-0.49, #2) rectangle (#1+0.49, #2-0.49);
\draw (#1, #2) node {#3};
\IfSubStr{#4}{u}{ \draw (#1, #2+0.5) -- (#1, #2+1);  }{  }
\IfSubStr{#4}{d}{ \draw (#1, #2-0.5) -- (#1, #2-1); }{  }
\IfSubStr{#4}{l}{ \draw (#1-0.5, #2) -- (#1-1, #2); }{  }
\IfSubStr{#4}{r}{ \draw (#1+0.5, #2) -- (#1+1, #2); }{  }
}

\newcommand{\MpsRectPQ}[4]{
\draw (#1-0.5, #2+0.5) rectangle (#1+0.5, #2-0.5);
\fill [red!30]    (#1-0.49, #2) rectangle (#1+0.49, #2+0.49);
\draw (#1, #2) node {#3};
\IfSubStr{#4}{u}{ \draw (#1, #2+0.5) -- (#1, #2+1);  }{  }
\IfSubStr{#4}{d}{ \draw (#1, #2-0.5) -- (#1, #2-1); }{  }
\IfSubStr{#4}{l}{ \draw (#1-0.5, #2) -- (#1-1, #2); }{  }
\IfSubStr{#4}{r}{ \draw (#1+0.5, #2) -- (#1+1, #2); }{  }
}

\newcommand{\MpsRectrho}[4]{
\draw[fill=red!30] (#1-0.5, #2+0.5) rectangle (#1+0.5, #2-0.5);
\draw (#1, #2) node {#3};
\IfSubStr{#4}{u}{ \draw (#1, #2+0.5) -- (#1, #2+1);  }{  }
\IfSubStr{#4}{d}{ \draw (#1, #2-0.5) -- (#1, #2-1); }{  }
\IfSubStr{#4}{l}{ \draw (#1-0.5, #2) -- (#1-1, #2); }{  }
\IfSubStr{#4}{r}{ \draw (#1+0.5, #2) -- (#1+1, #2); }{  }
}

\tikzstyle{fcicirc}=[ellipse, draw, fill=white, minimum width=15.0em, minimum height=2.0em]

\tikzstyle{fcirect}=[rectangle, draw, rounded corners=.55em, fill=blue!20, minimum width=15.0em, minimum height=2.0em]

\tikzstyle{state}=[circle, draw, fill=red!30, thick, minimum width=2.0em]
\tikzstyle{basis}=[circle, draw, fill=green!30, thick, minimum width=2.5em]
\tikzstyle{operator}=[rectangle, draw, fill=red!40, thick, minimum width=2.5em, minimum height=2.5em]
\tikzstyle{other}=[rectangle, draw, fill=yellow!50, thick, minimum width=2.5em, minimum height=2.5em]
\begin{document}

\author{Tong Jiang}
\affiliation{Department of Chemistry and Chemical Biology, Harvard University, Cambridge, MA, 02138, USA}

\author{Bryan O'Gorman}
\affiliation{IBM Quantum, IBM Research Cambridge, Cambridge, MA 02142, USA}
\author{Ankit Mahajan}
\affiliation{Department of Chemistry, Columbia University, New York, New York 10027, USA}
\author{Joonho Lee}
\email{joonholee@g.harvard.edu}
\affiliation{Department of Chemistry and Chemical Biology, Harvard University, Cambridge, MA, 02138, USA}

\title{Unbiasing Fermionic Auxiliary-Field Quantum Monte Carlo with \\Matrix Product State Trial Wavefunctions}

\date{\today}% It is always \today, today,
             %  but any date may be explicitly specified

\begin{abstract}
In this work, we report, for the first time, an implementation of fermionic auxiliary-field quantum Monte Carlo (AFQMC) using matrix product state (MPS) trial wavefunctions, dubbed MPS-AFQMC. 
Calculating overlaps between an MPS trial and arbitrary Slater determinants up to a multiplicative error, a crucial subroutine in MPS-AFQMC, is proven to be \#P-hard.
Nonetheless, we tested several promising heuristics in successfully improving fermionic phaseless AFQMC energies.
We also proposed a way to evaluate local energy and force bias evaluations free of matrix product operators. This allows for larger basis set calculations without significant overhead. We showcase the utility of our approach on one- and two-dimensional hydrogen lattices, even when the MPS trial itself struggles to obtain high accuracy. Our work offers a new set of tools that can solve currently challenging electronic structure problems with future improvements.
\end{abstract}
\maketitle

\section{Introduction}
One of the long-standing goals in quantum chemistry and condensed matter physics is to elucidate the $\textit{ab initio}$ electronic structure of interacting fermions~\cite{helgaker_molecular_2000,friesner_ab_2005}.
Although the dimension of the Hilbert space grows exponentially with system size, $N$, most studies rely on numerical techniques that exploit properties of chemical problems, such as weak correlation~\cite{jones2015density,koch2015chemist}, low entanglement and locality~\cite{lee_evaluating_2023}, \textit{etc}.
While methods such as density functional theory could often describe weakly correlated systems well, the computational challenge remains for strongly correlated systems with non-negligible dynamic correlation, including transition metal oxides, iron-sulfur clusters, \textit{etc.}~\cite{dagotto_complexity_2005,li_electronic_2019,sayfutyarova_constructing_2019,goings_reliably_2022,cui2022systematic}.
In the challenging regimes, the density matrix renormalization group (DMRG)~\cite{verstraete_density_2023,baiardi_density_2020,schollwock_density-matrix_2011,chan_density_2011} and quantum Monte Carlo (QMC)~\cite{austin_quantum_2012,becca_quantum_2017} methods emerge as particularly powerful many-body approaches with often reduced scaling for target accuracy. This paper explores the marriage of these two methods in fermionic simulations and offers insights from complexity theory and numerical results.

DMRG, a prominent tensor network algorithm, is extensively utilized for studying low-dimensional lattice systems in condensed matter physics~\cite{white1992density,*white1993density} and operates as an approximate full configuration interaction solver for the active space in quantum chemistry~\cite{olivares2015ab}. 
It leverages matrix product states (MPS) as the underlying wavefunction ansatz and refines the wavefunction through variational optimization. 
Representing arbitrary wavefunctions typically requires MPS's bond dimension $D$ to grow exponentially with system size (i.e., $D\sim\exp({N})$.) However, the innate locality in many physical systems limits the entanglement in the ground state, which MPS can capture with a constant $D$ in one-dimension~\cite{PhysRevLett.100.070502,*PhysRevLett.90.227902,*RevModPhys.82.277,*orus2014practical}.
The accuracy of DMRG systematically improves with higher values of $D$, and its cost scales as $\mathcal{O}(ND^3)$. 
In terms of {\it ab initio} fermionic simulations, DMRG solvers have been limited to a small active space and have been used to handle strong correlation. In contrast, the remaining correlation outside active space is left out~\cite{hachmann2007radical,*marti2008density,*ghosh2008orbital,*kurashige2009high,*sharma_spin-adapted_2012,*ma2013assessment}.
Currently, DMRG solvers have become the {\it de facto} solvers for active space with emerging strong correlation.
Addressing the dynamic correlation treatment on top of DMRG requires post-DMRG methods, which tend to add layers of complexity and elevate computational demands~\cite{yanai2010multireference,*kurashige2011second,*wouters_thouless_2013,*sharma2014low,*sharma2014communication,*cheng_post-density_2022,*larsson2022matrix}, as they require higher reduced density matrices (RDMs) beyond two-body RDMs.

Projector quantum Monte Carlo (QMC) stochastically performs imaginary-time evolution to sample the ground state and efficiently evaluates a statistical estimate for the ground state energy. Among different types of QMC, auxiliary-field quantum Monte Carlo (AFQMC)~\cite{zhang_constrained_1995,Zhang2003Apr}, has emerged as a particularly accurate and efficient approach for solving {\it ab initio} problems.
AFQMC has gained extensive use in quantum chemistry~\cite{motta_ab_2018,lee_twenty_2022,Lee2020Jul,lee2021phaseless,lee2021constrained,mahajan_selected_2022,mahajan2023response}. In recent years, it is also seen as a powerful tool for quantum computation of chemical systems in the noisy intermediate-scale quantum era~\cite{huggins_unbiasing_2022,wan_matchgate_2023,amsler_quantum-enhanced_2023,kiser_classical_2023,Blunt2024Oct,Kiser2024Aug,Jiang2024Jul}.
The constraint-free, exact AFQMC, known as free-projection AFQMC, inherently faces an exponential-scaling challenge in sample complexity due to the notorious fermionic phase/sign problem. To circumvent this, a trial wavefunction is introduced to constrain the phase/sign of the QMC walkers, referred to as the phaseless approximation~\cite{Zhang2003Apr}.

Phaseless AFQMC (ph-AFQMC) effectively controls the phase problem and is polynomial scaling in sample complexity at the expense of introducing biases.
Advancements in AFQMC research often center around the quest for improved trial wavefunctions~\cite{huggins_unbiasing_2022,mahajan_taming_2021,vitali_calculating_2019,chang_auxiliary-field-based_2016,kanno_quantum_2023,Mahajan2024Oct}. The accuracy and scalability of ph-AFQMC hinge on the selection of an appropriate trial wavefunction. Traditionally, the unrestricted Hartree-Fock (UHF) wavefunction has been the most common choice of trial due to the good tradeoff between accuracy and cost.  ph-AFQMC with UHF trial has demonstrated success in solving complex systems such as Hubbard model~\cite{zheng_stripe_2017} and \textit{ab initio} systems~\cite{lee_twenty_2022}.
While single determinant trial wavefunctions are favored for their scalability, their accuracy is inherently limited, as shown in various applications~\cite{al-saidi_bond_2007, simons_collaboration_on_the_many-electron_problem_direct_2020,simons_collaboration_on_the_many-electron_problem_towards_2017,mahajan_taming_2021}.
Recent research on selected configuration interaction trial wavefunctions has investigated the use of a large number of Slater determinants, leading to significant enhancements in accuracy~\cite{mahajan_selected_2022}, even when the trial wavefunction does not capture any correlation outside active space~\cite{lee_twenty_2022}. 
\add{Although alternative trial wavefunctions may offer 
benefits in specific cases, their ability to systematically approach 
true ground state accuracy remains an open question
~\cite{PhysRevB.94.235119,vitali_calculating_2019,chang_auxiliary-field-based_2016,simons_collaboration_on_the_many-electron_problem_towards_2017}.
}   

It is an enticing prospect that, for systems too large to be accurately handled by DMRG alone due to the requirement of large bond dimension and the large correlation space left out from DMRG, ph-AFQMC can account for the remaining electron correlations (i.e., dynamic correlation).
Such a goal can be achieved by exploiting MPS as a trial wavefunction for ph-AFQMC.
While the combination of tensor network states and projector QMC has demonstrated success in spin model systems~\cite{wouters_projector_2014, qin_combination_2020}, 
such combination has not been shown in {\it ab initio} electronic structure problems.
While some flavors of projector QMC can relatively straightforwardly combine the two methods, we prove by using complexity theoretic tools that such a combination for ph-AFQMC is \#P-hard in this work.

One may imagine performing the imaginary time evolution directly with MPS via the time-dependent variational principle (TDVP)~\cite{paeckel_time-evolution_2019,haegeman_time-dependent_2011,haegeman_unifying_2016,ren2022time}, as opposed to relying on QMC to perform the imaginary time evolution. This approach is usually used for model systems and is impractical when one needs to study a large basis set of \textit{ab initio} systems where the construction of matrix product operator quickly becomes intractable~\cite{chan_matrix_2016,ren_general_2020}.
We also note that MPS was also used as a trial wavefunction in quantum-classical hybrid AFQMC~\cite{amsler_quantum-enhanced_2023} by expanding the MPS as a linear combination of multiple Slater determinants. This post-processing step for MPS trials is not scalable due to the exponential growth in the number of Slater determinants, which is not feasible beyond 20 orbitals.
The objective of this paper is to explore scalable heuristic approaches to employ the MPS trial obtained via DMRG in ph-AFQMC, which aims to leverage the strengths of both methods. The resulting method, dubbed MPS-AFQMC, will benefit from the strength of DMRG in accurately capturing static correlations within active spaces and the power of AFQMC in efficiently capturing electron correlation across the entire set of orbitals. A successful implementation of such a method will allow for direct, quantitative comparisons with experiments towards the basis set limit. 
\add{For higher-dimensional systems where DMRG may struggle to describe the ground state efficiently, one could extend our approach to the PEPS-AFQMC method, where a projected entangled-pair state (PEPS) is the trial wavefunction. This idea was recently explored in spin systems~\cite{qin_combination_2020}.} 

The structure of this paper is as follows: 
Section~\ref{sec:theory} presents theoretical and computational details developed and used in this work.
Section~\ref{sec:afqmc} provides a brief overview of AFQMC. Section~\ref{sec:MPS} then summarizes the MPS/MPO representations of fermionic many-body states and operators. 
Section~\ref{sec:mps_afqmc} provides a detailed introduction to the MPS-AFQMC algorithm, including various strategies for calculating the overlap between MPS trial and walker wavefunctions and the computation of force bias and local energy.
Section~\ref{sec:virt} describes the MPO-free technique in MPS-AFQMC to capture dynamic correlation beyond the active space by employing the MPS solution in the active space as the trial. 
Section~\ref{sec:BQP} demonstrates the \#P-hardness of computing the overlap between the MPS trial and the walker's SD.
Section~\ref{sec:results} presents numerical results and discussions, 
and Section~\ref{sec:conclusion} concludes the paper with suggestions for future directions.

\begin{figure*}
    \centering
    \includegraphics[width=\textwidth]{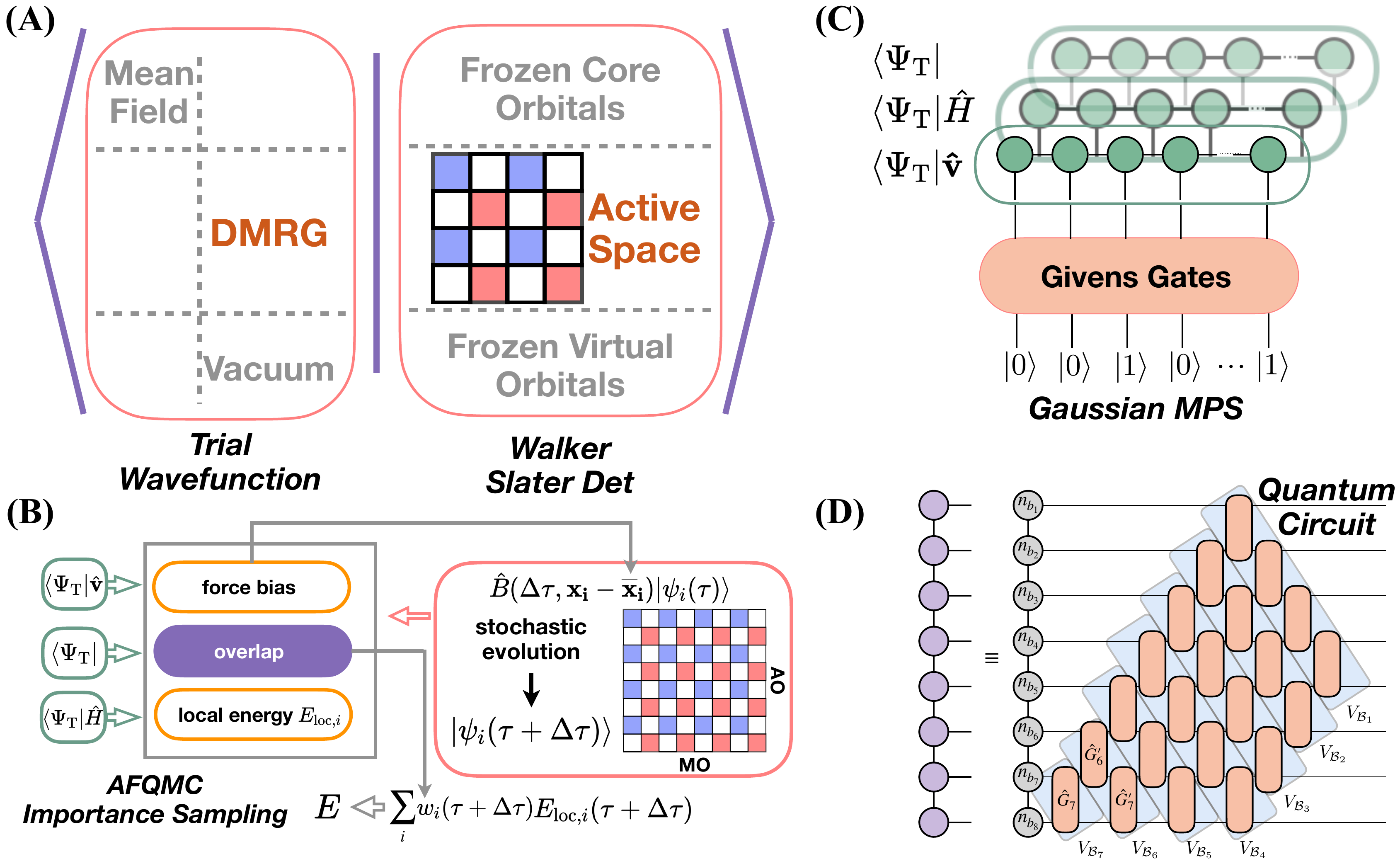}
    \caption{
    \add{\textbf{Overview of MPS-AFQMC with the walker in the Gaussian MPS form.}
    (A) Evaluating the overlap between the walker and trial
    whose active space component is an MPS wavefunction.
    (B) The workflow of a general AFQMC algorithm
    where the computation of force bias, overlap and local energy 
    are performed, which is described in (C) when a MPS trial is used. 
    In (C), without loss of generality, we assume both walker and MPS trial are 
    in the active space. 
    (D) describes the process of construct the Gaussian MPS for 
    the walker Slater determinant.} 
    }
    \label{fig:wflow}
\end{figure*} 

\section{Theory}\label{sec:theory}
\subsection{Brief review of ph-AFQMC}\label{sec:afqmc}
The ph-AFQMC method aims to perform imaginary time evolution:
\begin{equation}
    \left|\Psi_0\right\rangle \propto \lim _{\tau \rightarrow \infty} \exp (-\tau \hat{H})\left|\Phi_0\right\rangle=\lim _{n\rightarrow \infty}(\exp 
 (-\Delta \tau \hat{H}))^n|\Phi_0\rangle\label{eq:afqmc}
\end{equation}
where $\left|\Psi_0\right\rangle$ is the exact ground state wavefunction, and $\left|\Phi_0\right\rangle$ is an initial state that satisfying $\langle\Phi_0|\Psi_0\rangle\neq 0$. The repeated application of short-time propagators implements the limit of $\tau=n\Delta\tau\rightarrow\infty$. With spin-orbital notation, we can express the \textit{ab initio} electronic Hamiltonian in second quantization as
\begin{equation}
\hat{H}=\sum_{p,q=1}^N h_{p q} \hat{a}_p^{\dagger} \hat{a}_q+\frac{1}{2} \sum_{p,q,r,s=1}^N g_{psqr} \hat{a}_p^{\dagger} \hat{a}_q^{\dagger} \hat{a}_r \hat{a}_s\label{eq:qcham}
\end{equation}
where the two-electron repulsion integral (ERI) can be factorized with the Cholesky decomposition $g_{psqr}=(ps|qr)=\sum_{\gamma=1}^{N_\gamma} L_{ps}^\gamma L_{qr}^{\gamma}$. 
Using the Hubbard-Stratonovich transformation~\cite{hubbard1959calculation,*Stratonovich}, we can write the short-time propagator as
\begin{equation}
    e^{-\Delta \tau \hat{H}} = \int \mathrm{d} \mathbf{x}\: p(\mathbf{x}) \hat{B}(\Delta \tau, \mathbf{x})+\mathcal{O}\left(\Delta \tau^2\right)
\end{equation}
where $p(\mathbf x)$ is a Gaussian distribution and $\hat{B}(\mathbf{x})$ is a one-body propagator that is coupled to the auxiliary field $\mathbf x$.
According to Thouless' theorem~\cite{thouless1960stability,*thouless1961vibrational}, applying $\hat{B}(\mathbf{x})$ to a single determinant results in another (non-orthogonal) single determinant. 
Utilizing this property, in ph-AFQMC, we write the global wavefunction at an imaginary time $\tau$ as a weighted sum over walkers (i.e., single determinants),
\begin{equation}
|\Psi(\tau)\rangle=\sum_i^{N_\textrm{w}} w_i(\tau) \frac{\left|\psi_i(\tau)\right\rangle}{\left\langle\Psi_\textrm{T}|\psi_i(\tau)\right\rangle}.
\end{equation}
where $w_i(\tau)$ is the weight for the $i$-th walker at time $\tau$, $|\psi_i(\tau)\rangle$ is the $i$-th walker wavefunction and $|\Psi_\textrm{T}\rangle$ is the trial wavefunction introduced for importance sampling.

In ph-AFQMC, the walker wavefunction, $\left|\psi_i(\tau)\right\rangle$, is updated upon applying the effective one-body propagator, and the weight is updated following the phaseless approximation:
\begin{equation}
\left|\psi_i(\tau+\Delta \tau)\right\rangle  =\hat{B}\left(\Delta \tau, \mathbf{x}_i-\overline{\mathbf{x}}_i\right)\left|\psi_i(\tau)\right\rangle\label{eq: propagate}
\end{equation}
\begin{equation}
\label{eq:phaseless}
w_i(\tau+\Delta \tau)  =I_\text{ph}\left(\mathbf{x}_i, \overline{\mathbf{x}}_i, \tau, \Delta \tau\right) \times w_i(\tau),
\end{equation}
where $\overline{\mathbf{x}}_i$ is the force bias that dynamically shifts Gaussian probability distribution, and the phaseless importance function, $I_\text{ph}$, is given as
\begin{equation}
    I_\text{ph}\left(\mathbf{x}_i, \overline{\mathbf{x}}_i, \tau, \Delta \tau\right)=\left|I\left(\mathbf{x}_i, \overline{\mathbf{x}}_i, \tau, \Delta \tau\right)\right| \times \max \left(0, \cos \left(\theta_i(\tau)\right)\right)
\end{equation}
with the hybrid importance function being
\begin{equation}
    I\left(\mathbf{x}_i, \overline{\mathbf{x}}_i, \tau, \Delta \tau\right)=S_i(\tau, \Delta \tau) e^{\mathbf{x}_i \cdot \overline{\mathbf{x}}_i-\overline{\mathbf{x}}_i \cdot \overline{\mathbf{x}}_i / 2}.
\end{equation}
Here, the overlap ratio of the $i$-th walker is
\begin{equation}
    S_i(\tau, \Delta \tau)=\frac{\langle\Psi_{\textrm{T}}|\hat{B}\left(\Delta \tau, \mathbf{x}_i-\overline{\mathbf{x}}_i\right)| \psi_i(\tau)\rangle}{\left\langle\Psi_{\textrm{T}}|\psi_i(\tau)\right\rangle}\label{eq:ovlp_ratio}
\end{equation}
whose phase is
\begin{equation}
    \theta_i(\tau)=\arg \left(S_i(\tau, \Delta \tau)\right).
\end{equation}
The weight update in \cref{eq:phaseless} ensures the positivity of weights of all walkers, but
it introduces a bias, which can be eliminated if the trial wavefunction is exact.
This cosine projection can be viewed as constraining open-ended random walks with 
a boundary condition set by a trial wavefunction.
If the trial wavefunction is exact, the boundary condition is also exact and thereby ph-AFQMC recovers the exact ground state energy.

Lastly, we calculate the ph-AFQMC energy as a function of $\tau$ by
\begin{equation}
E = \frac{\sum_iw_i E_{L,i}}{\sum_i w_i}
\end{equation}
where the local energy of the $i$-th walker is
\begin{equation}
E_{L,i} = 
\frac{\langle \Psi_\text{T}|\hat{H}|\psi_i(\tau)\rangle}{\langle \Psi_\text{T}|\psi_i(\tau)\rangle}
\end{equation}
Other details, such as the definition of the force bias and the one-body propagator $\hat{B}(\Delta \tau, \mathbf{x})$, are presented in \cref{app:afqmc}.

%%%%%%%%%%%%%
\subsection{Brief review of MPS}\label{sec:MPS}
Here, we briefly review some formalisms for MPS that will be useful in discussing MPS-AFQMC.
A generic \textit{ab initio} AFQMC trial can be written as a linear combination of determinants (i.e., product states)
\begin{equation}
    |\Psi_\textrm{T}\rangle = \sum_{\{n\}} C_{n_1 n_2\cdots n_N}|n_1 n_2\cdots n_N\rangle\label{eq:msd}
\end{equation}
where $|n_1 n_2\cdots n_N\rangle$ ($n_i\in\{0,1\}$) is an occupancy basis state of $N$ spin-orbitals $\{\phi_1\phi_2\cdots\phi_N\}$ and the coefficients $C$ is a vector of dimension of $d^N$ with $d = 2$.
The MPS representation is written as
\begin{equation}
\label{eq:mps-generic}
|\Psi_\textrm{T}\rangle  = \sum_{\{a\},\{n\}}
 A^{n_1}_{a_1}\cdots A^{n_i}_{a_{i-1}a_i} \cdots
       A^{n_N}_{a_{N-1}}  |n_1\cdots n_i\cdots n_N\rangle
\end{equation}
\begin{equation}
% \label{eq:mps-diagram}
   % \ket{\Psi_{\textrm{T}}} 
   = 
   \begin{diagram}
\MpsCircle{0}{0}{$A$}{ur};
   \MpsCircle{2}{0}{$A$}{url};
   \draw (4, 0) node {$\cdots$};
   \MpsCircle{6}{0}{$A$}{url};
   \MpsCircle{8}{0}{$A$}{lu};
   
   \draw (0,1.5) node {$n_1$};
   \draw (2,1.5) node {$n_2$};
   \draw (8,1.5) node {$n_N$};
   \draw (6,1.5) node {$n_{N-1}$};
   
   \draw (1, 0.5) node {$a_1$};
   \draw (3, 0.5) node {$a_2$};
   \draw (5, 0.5) node {$a_{N-2}$};    	    
   \draw (7, 0.5) node {$a_{N-1}$};    	    
   
   \draw (0,0) node (X){$ $};
   \end{diagram}\nonumber
\end{equation}
where the tensor $A^{n_i}_{a_{i-1}a_i}$ is a three-dimensional entity with dimensions ($D_{i-1}, d, D_i$). Without compression, $D_{i-1}$ equals $d^{i-1}$ and $D_i$ equals $d^i$.
In MPS-AFQMC, we variationally optimize an MPS with a preset bond dimension, $D_{\textrm{T}}$, and use it as an AFQMC trial.
The larger $D_\textrm{T}$ value, the more closely the trial wavefunction approximates the ground state wavefunction.
The matrix product operator (MPO) of the Hamiltonian in Eq.~\eqref{eq:qcham},
\begin{align}
\label{eq:mpo}
   \hat{H} =  \sum_{\{w\},\{n\},\{n'\}} &
    W^{n'_1, n_1}_{w_1} W^{n'_2, n_2}_{w_1w_2} \cdots
                   W^{n'_N, n_N}_{w_{N-1}} \\ \nonumber &|n'_1 n'_2\cdots n'_N\rangle
                   \langle n_N n_{N-1} \cdots n_1|
\end{align}
\begin{equation}\label{eq:mpo-generic}
= 
   \begin{diagram}
   \draw (0,1.5) node {$n'_1$};
   \MpsRectGreen{0}{0}{$W$}{udr};
   \draw (1, 0.75) node {$w_1$};
   \draw (0,-1.5) node {$n_1$};
   
   \draw (2,1.5) node {$n'_2$};
   \MpsRectGreen{2}{0}{$W$}{udrl};
   \draw (3, 0.75) node {$w_2$};
   \draw (2,-1.5) node {$n_2$};
   
   \draw (4, 0) node {$\cdots$};
  
   \draw (4.7, 0.75) node {$w_{N-2}$};
   \draw (6,1.5) node {$n'_{N-1}$};
   \MpsRectGreen{6}{0}{$W$}{ldur}; 
   \draw (6,-1.5) node {$n_{N-1}$};
   
   \draw (7, 0.75) node {$w_{N-1}$};
   \draw (8,1.5) node {$n'_N$};
   \MpsRectGreen{8}{0}{$W$}{dlu}; 
   \draw (8,-1.5) node {$n_N$};
   
   \draw (0,0) node (X){$ $};
   \end{diagram} \nonumber
\end{equation}
can be expressed in the spin-orbital basis using the Jordan-Wiger transformation~\cite{jordan_about_1928,li_spin-projected_2017,ren_general_2020}.
While we have an MPO-free formulation of our approach in Sec.~\ref{sec:virt}, we found that using MPOs can greatly accelerate the simulations for minimal basis set examples.
%%%%%%%%%%%%%
\subsection{MPS-AFQMC}\label{sec:mps_afqmc}
In this section, we introduce the MPS-AFQMC algorithm, as illustrated in Fig.~\ref{fig:wflow}(A-C), which integrates the MPS trial into ph-AFQMC. 
Here, we apply the propagator to walkers (as one-body rotation) so that no approximation is made in this step.
To update the weights, one must evaluate the overlap between the MPS trial and walkers and the force bias.
Furthermore, local energy must be calculated to estimate the final ph-AFQMC energy.
Evaluating these quantities is straightforward if walkers were also in the MPS format.
Therefore, the central subroutine in MPS-AFQMC is the conversion of walker wavefunctions (i.e., Slater determinants in an arbitrary basis) into an MPS.
This step can be done only approximately in practice, and therefore, MPS-AFQMC will inherently have errors in overlap, force bias, and local energy evaluation.
Below, we will discuss the heuristics that we examined to perform this step.

\subsubsection{Slater determinants in an arbitrary basis}\label{sec:sd}
In ph-AFQMC, the walker wavefunction is a Slater determinant in an arbitrary basis, 
\begin{equation}
    |\psi\rangle  =\prod_{i=1}^{N_{\textrm{e}}/2} \hat{c}_{u_i \downarrow}^{\dagger}\hat{c}_{u_{i \uparrow}}^{\dagger}|0\rangle
\end{equation}
where we assumed an equal number of spin up and spin down fermions ($N^\uparrow=N^\downarrow=N_{\textrm{e}}/2$) for simplicity in this work, and
the ``walker orbital'' creation operator $\hat{c}_{{i\sigma}}^\dagger=\sum_{p=1}^N U_{pi}^\sigma \hat{a}_{p}^\dagger$ defined through a unitary transformation, $U$, of the initial orbitals, $\{\hat{a}_p^\dagger\}$. 
The matrix $U$ is complex-valued in the context of AFQMC, as indicated by Eqs.(\ref{eq:Bprop},~\ref{eq: propagate}). By ordering the walker orbitals as an alternating spin-orbital series, $\{\uparrow\downarrow\cdots \uparrow\downarrow\cdots\uparrow\downarrow\}$, the parameterized matrix $U$ exhibits the checkerboard pattern, 
\begin{equation}
    U_{p, i}=
    \begin{cases}
        U_{p/2,i/2}^{\downarrow}& \text{if } p, i \text{ are even}\\
        U_{(p+1)/2,(i+1)/2}^{\uparrow}& \text{if } p, i \text{ are odd}\\
        0 & \text{otherwise}
    \end{cases}\label{eq:checkerboard}
\end{equation}
Using this checkerboard ordering, we write
\begin{equation}
    |\psi\rangle =\prod_{i=1}^{N_{\textrm{e}}/2} \sum_{\substack{p \text{ even} \\ q \text{ odd}}}^N U_{pi}U_{qi}\hat{a}_p^{\dagger}\hat{a}_q^{\dagger}|0\rangle
\label{eq:mpo-mps}
\end{equation}
where $\hat{a}_p^{\dagger}$, $\hat{a}_q^{\dagger}$ are consistent with those defined in \cref{eq:qcham}.
Our goal is to form an MPS representation of \cref{eq:mpo-mps} in the original orbital basis $\{\hat{a}_q^\dagger\}$ such that subsequent operations with an MPS trial (also in the original orbital basis) become straightforward.

\subsubsection{Slater determinant to Gaussian MPS}\label{sec:sd2mps_White}  
A direct formulation of creation operators into MPOs~\cite{wu_tensor_2020,jin_efficient_2020} in \cref{eq:mpo-mps} runs into
steep computational cost due to the long-rangedness of the operators, as examined by Fig.~\ref{fig:mpomps}.
The approach presented by Fishman and White~\cite{fishman_compression_2015} circumvents this by decomposing products of creation operators into a sequence of short-ranged operators with only nearest neighbor couplings (i.e., 2-qubit Givens rotation gates).
Here, we summarize the relevant part of the algorithm and modifications for our use in ph-AFQMC.

The correlation matrix \add{(or one-body reduced density matrix)} of a SD $|\psi\rangle$ is $\Lambda=U U^{\dagger}$,
\begin{equation}
    \Lambda_{pq} = \langle \phi_p | \phi_q \rangle = \sum_{i}U_{p,i} (U_{q, i})^* \label{eq:sd}
\end{equation}
The correlation matrix $\Lambda$ is a $N\times N$ matrix that takes the checkerboard pattern with alternating spin ordering in Eq.~\eqref{eq:checkerboard}.
The algorithm starts with compressing a given SD as a Gaussian MPS (GMPS).
We start from site 1 and move to the next site until we reach the last site. At site 1, 
we first \add{diagonalize} a \( B \times B \) ($2\le B\le N$) upper left subblock of \( \Lambda \), denoted as \( \Lambda_{\mathcal{B}_1} \). $B$ sets the correlation length within GMPS, and $B=N$ recovers the full matrix diagonalization. The eigenvalues of this subblock, \( n_b \), lie between 0 and 1. In the case of full matrix diagonalization of \( \Lambda \), there will be \( N^\uparrow + N^\downarrow \) eigenvalues equal to 1, with the remaining \( N - (N^\uparrow + N^\downarrow) \) eigenvalues being 0. As we increase \( B \), at least one eigenvalue will progressively approach either 1 or 0.
An adaptive threshold value  \( \epsilon_{\textrm{ad}} \) is set to control the block size, which applies to determine the proximity of an eigenvalue to 1 or 0. 
Consequently, a smaller \( \epsilon_{\textrm{ad}} \) results in a larger block size \( B \). 
The smaller this threshold, the more precise the compression becomes. Setting \( \epsilon_{\textrm{ad}} \) to 0 allows for an exact conversion of SD to a GMPS.

Upon diagonalization of the correlation matrix, we obtain a set of 2-qubit Givens rotation gates that contribute to the final GMPS. Once the specified threshold is met, the eigenvector \( \vec{v}^T = (v_1, v_2, \cdots, v_B) \) is chosen, based on its corresponding eigenvalue \( n_{b_1} \) being nearest to either 0 or 1. This selection is facilitated through the application of a transformation matrix \( \vec{v}^T V_{\mathcal{B}_1} = (1, 0, \cdots, 0) \), where \( V_{\mathcal{B}_1} \) denotes a series of local transformations that are associated with site 1:
\begin{equation}
    V_{\mathcal{B}_1} = V_{B-1}(\theta_{B-1})V_{B-2}(\theta_{B-2})\cdots V_1(\theta_1)\label{eq:block_Givens}
\end{equation}
\begin{equation}
    V_{i}(\theta_i)=\left(\begin{array}{llllll}
1 & & & & &\\
& 1 & & & &\\
& & (\cos{\theta_i})^* & -\sin{\theta_i} & & \\ 
& & (\sin{\theta_i})^* & \cos{\theta_i} & & \\
& & & & 1 & \\
& & & & & 1
\end{array}\right)\label{eq:Givens}
\end{equation}
which is an identity matrix except from the $i-$th and $(i+1)-$th rows and columns, and $\theta_i=\tan^{-1}{(v_{i+1}/v_{i})}$.
The correlation matrix is then transformed into
\begin{equation}
    V_{\mathcal{B}_1}^{\textrm{T}} \Lambda V_{\mathcal{B}_1}^* =\begin{pmatrix}
  n_{b_1}
  &   \\
   &
    \Lambda'
\end{pmatrix}
\end{equation}
with the top left entry becoming $n_{b_1}$, and the rest of the elements of the first row and the first column should be nearly zero as long as a tight enough $\epsilon_{\textrm{ad}}$ was chosen.
Then we move to site 2 and pick the subblock $\Lambda_{\mathcal{B}_2}$ with the block size of $B$ ($2\le B\le N-1$) which is the top left entry of $(V_{\mathcal{B}_1}^{\textrm{T}} \Lambda V_{\mathcal{B}_1}^*)$ starting from $(V_{\mathcal{B}_1}^{\textrm{T}} \Lambda V_{\mathcal{B}_1}^*)_{2,2}$.
After reaching the threshold, we have an eigenvalue $n_{b_2}$ of subblock $\mathcal{B}_2$ that is close to 0 or 1. Another transformation matrix $V_{\mathcal{B}_2}$ leads to
\begin{equation}
    V_{\mathcal{B}_2}^{\textrm{T}}V_{\mathcal{B}_1}^{\textrm{T}} \Lambda V_{\mathcal{B}_1}^*V_{\mathcal{B}_2}^* =\begin{pmatrix}
  n_{b_1}
  &  &  \\
  & n_{b_2} & \\
   & &
    \Lambda''.
\end{pmatrix}
\end{equation}
By repeating this procedure for the rest of the sites, the correlation matrix is finally transformed into a diagonal matrix
\begin{equation}
    V_{\mathcal{B}_{N-1}}^{\textrm{T}}\cdots V_{\mathcal{B}_1}^{\textrm{T}} \Lambda V_{\mathcal{B}_1}^*\cdots V_{\mathcal{B}_{N-1}}^* =\begin{pmatrix}
  n_{b_1} &  &  &\\
  & n_{b_2} & & \\
   & & & &\\
   & & & & n_{b_N}
\end{pmatrix}=U_0 U_0^\dagger.
\end{equation}
This diagonal matrix corresponds to a correlation matrix of a 
product state $|\psi_0\rangle = \ket{n_{b_1}n_{b_2}\cdots n_{b_N}}$ 
occupying the \add{natural orbitals $\{\tilde{\phi}_1\tilde{\phi}_2\cdots\tilde{\phi}_N\}$} 
where $n_{b}\in [0, 1]$ with 1 representing occupied and 0 representing 
unoccupied, and $U_0$ is the parameterized matrix for the SD of $|\psi_0\rangle$,
To go back to the original single-particle basis, we can use 
\begin{equation}
    U=V_{\mathcal{B}_1}\cdots V_{\mathcal{B}_{N-1}}U_0
\end{equation}
This sequence of 2-by-2 rotations defines the GMPS state corresponding to the original SD state.

Using this, we can represent walker states as
\begin{equation}
    |\psi\rangle=\hat{V}_{\mathcal{B}_1}\cdots \hat{V}_{\mathcal{B}_{N-1}}|\psi_0\rangle\label{eq:full_sequence}
\end{equation}
where $|\psi_0\rangle$ can be represented as an MPS whose bond dimension is 1,
\begin{equation}
  |\psi_0\rangle = 
   \begin{diagram_08}
\MpsCircle{0}{0}{$n_{b_1}$}{ur};
   \MpsCircle{2}{0}{$n_{b_2}$}{url};
   \draw (4, 0) node {$\cdots$};
   \MpsCircle{6}{0}{$n_{b_{N}}$}{lu};
      
   \draw (0,0) node (X){$ $};
   \end{diagram_08}
\end{equation}
The MPS representation of $|\psi\rangle$ can be obtained by applying a sequence of Givens rotations, 
as shown in Fig.~\ref{fig:wflow}(D).
This gate application to the initial state converts the GMPS into an MPS. 
Each $V$ in Eq.~\eqref{eq:Givens} becomes a two-qubit gate that is used to represent the Givens rotation in \cref{eq:block_Givens,eq:Givens}) 
\begin{equation}
    \hat{G}_i(\theta_i)=\left(\begin{array}{cccc}
1 & 0 & 0 & 0 \\
0 & (\cos{\theta_i})^* & -\sin{\theta_i} & 0 \\
0 & (\sin{\theta_i})^* & \cos{\theta_i} & 0 \\
0 & 0 & 0 & 1 
\end{array}\right)
\end{equation}
The two-qubit gate is written under the site occupation basis of $\{\ket{00},\ket{01},\ket{10},\ket{11}\}$ for the $i$-th and $(i+1)$-th spin orbitals where 0 refers to unoccupied spin and 1 refers to occupied spin. 

Successive applications of Givens rotations result in a growing bond dimension of the MPS. Consequently, it becomes necessary to approximate the corresponding quantum circuits by employing singular value decomposition (SVD) compression with a fixed bond dimension.
\begin{equation}
\begin{diagram_05}
\MpsWalkerCircle{0}{0}{$ $}{ulr};
\MpsWalkerCircle{1.5}{0}{$ $}{ulr};
\draw[->, thick] (0, 2.8) -- (0, 2.0);
\draw[->, thick] (1.5, 2.8) -- (1.5, 2);
\draw[<-, thick] (-0.5, 0.) -- (-1, 0.);
\draw[<-, thick] (2, 0.) -- (2.5, 0.);

\GivensGate{-0.5}{1}{$ $}{}
\draw (-1.5, 0) node {$n_L$};
\draw (3, 0) node {$n_R$};
\draw (0, 3) node {$n_1$};
\draw (1.5, 3) node {$n_2$};

\draw (4, 1.2) node {$\xrightarrow[]{\text{compress}}$};
\MpsWalkerCirclePattern{6.5}{0.5}{$ $}{ulr};
\draw (7.2, 1) node {$D_\textrm{w}$};
\MpsWalkerCirclePattern{8}{0.5}{$ $}{ulr};

\end{diagram_05}\label{eq:svd}
\end{equation}
For approximate contractions, we introduce a parameter, called walker bond dimension ($D_\textrm{w}$), that caps the bond dimension
during the Givens rotation circuit evaluation as shown in Eq.~\eqref{eq:svd}. 
The tradeoff between accuracy and cutoff in this strategy is controlled by $\epsilon_\text{ad}$ and $D_\text{w}$.
\add{Additional strategies and details are provided in  Appendix~\ref{app:sd2mps_others}.} 
\subsubsection{Overlap, force bias and local energy}\label{sec:loce}
Once the walker wavefunction is converted to an MPS format, the overlap can be computed by contracting it with the trial, as shown in Eq.~\eqref{eq:calc_ovlp}.
\begin{equation}\label{eq:calc_ovlp}
\braket{\Psi_{\textrm{T}}|\psi} = 
    \begin{diagram_05}

    \MpsCircle{0}{2}{}{dr};
    \MpsWalkerCircle{0}{0}{}{ur};
    \MpsCircle{2}{2}{}{ldr};
    \MpsWalkerCircle{2}{0}{}{url};
    
	\MpsCircle{6}{2}{}{dlr};
    \MpsWalkerCircle{6}{0}{}{url};
    
    \draw (4, 0) node {$\cdots$};
    \draw (4, 2) node {$\cdots$};
    
    \MpsCircle{8}{2}{}{ld}; 
    \MpsWalkerCircle{8}{0}{}{lu}; 
    
    \draw (0,1) node (X){$ $};
    \end{diagram_05}
\end{equation}

The computation of force bias boils down to performing the following evaluation:
\begin{equation}
\bar{x}_\gamma = i\sqrt{\Delta t}\sum_{pq}L_{pq}^\gamma\frac{\left\langle\Psi_{\textrm{T}}\left| \hat{a}_p^\dagger \hat{a}_q\right| \psi\right\rangle}{\left\langle\Psi_{\textrm{T}} |\psi\right\rangle}\label{eq:vbias}
\end{equation}
There are two ways to compute it.
Most AFQMC codes compute the one-body Green's function first~\cite{motta_ab_2018}. Then the force bias can be computed by simply contracting $G_{pq}$ and $L_{pq}^\gamma$.
However, this approach is expensive in our context because it requires contraction of the MPS trial with MPO $\hat{a}_p^\dagger\hat{a}_q$ and walker MPS for each $p$ and $q$.
Instead, we loop over all auxiliary fields since the number of fields is usually small, and for each field $\alpha$, we express $\hat{v}_\alpha$ as the form of MPO,
\begin{equation}
\label{eq:mps-mpo-expecatation}
    \braket{\Psi_{\textrm{T}} | \hat v_\alpha | \psi} = 
    \begin{diagram_05}
    % MPS part
    \MpsCircle{0}{2}{}{dr};
    
    \MpsCircle{2}{2}{}{drl};
    
    \draw (4, 2) node {$\cdots$};
   
    \MpsCircle{6}{2}{}{ldr}; 
    
    \MpsCircle{8}{2}{}{ld}; 
    
    % MPO part
    \MpsRect{0}{0}{}{udr};
    
    \MpsRect{2}{0}{}{udrl};
    
    \draw (4, 0) node {$\cdots$};
   
    \MpsRect{6}{0}{}{ldur}; 
    
    \MpsRect{8}{0}{}{dlu}; 
    
    \draw (0,0) node (X){$ $};
    
    % MPS part
    \MpsWalkerCircle{0}{-2}{}{ur};
    
    \MpsWalkerCircle{2}{-2}{}{url};
    
    \draw (4, -2) node {$\cdots$};
   
    \MpsWalkerCircle{6}{-2}{}{lur}; 
    
    \MpsWalkerCircle{8}{-2}{}{lu}; 
    \end{diagram_05}  
\end{equation}
The computation of local energy is completely analogous to force bias. In practice, we precompute the MPS of $\langle \Psi_{\textrm{T}} | \hat{v}_\alpha$ and $\langle \Psi_{\textrm{T}} |\hat{H}$ and compress them to maximize the computational efficiency.
This amounts to the ``half-rotation'' used in the AFQMC literature ~\cite{Lee2020Jul}.
Shared memory allocation is utilized across all AFQMC child processes if these half-rotated MPSs require substantial memory.
The bond dimension of the Hamiltonian MPO can also be made smaller by screening the integrals while preserving the accuracy.

We denote the bond dimension of the half-rotated MPS for $\langle \Psi_{\textrm{T}} | \hat{v}_\alpha$ and $\langle \Psi_{\textrm{T}} |\hat{H}$ as $D_\textrm{chol}$ and $D_\textrm{hr}$ separately.
Both $D_\textrm{chol}$ and $D_\textrm{hr}$ can be further compressed.
While the Cholesky operator $\hat{v}_\alpha=\sum_{pq}L_{pq}^\alpha a_p^\dagger a_q$ only contains one-body operators, the compression is very efficient.
As for $D_\textrm{hr}$, the exact $D_\textrm{hr}$ is equal to $D_\textrm{T}D_H$ where $D_H$ is the MPO bond dimension of the Hamiltonian. 
When the trial state becomes more accurate, the compression of $D_\textrm{hr}$ becomes more efficient since in the limit of approaching ground state, $D_\textrm{hr}$ can be compressed to be $D_\textrm{T}$ without losing accuracy. The numerical example will be shown in Sec.~\ref{sec:4by4}.
In Table.~\ref{tab:computational_costs}, we summarize the leading order computational scaling for the subroutines of MPS-AFQMC.
\begin{table}[ht]
\centering
\begin{tabular}{lc}
\hline
Method & Leading order scaling \\
\hline
SD-to-MPS & $O(NBd^3D_\textrm{w}^3)$ \\
Overlap & $O(NdD_\textrm{T}D_\textrm{w}^2)$ \\
Force bias & $O(NdD_\textrm{chol}D_\textrm{w}^2)$ \\
Local energy & $O(NdD_\textrm{hr}D_\textrm{w}^2)$ \\
\hline
\end{tabular}
\caption{Computational costs for the subroutines of MPS-AFQMC. $d=2$ is the physical degrees of freedom of MPS (Eq.~\eqref{eq:mps-generic}) and $B$ is the block size of the corresponding GMPS.}
\label{tab:computational_costs}
\end{table}
\subsection{Virtual correlation energy}\label{sec:virt}
In the case of using MPS as the trial wavefunction for the ``active space", we provide an algorithm for calculating the correlation energy outside the active space. This step is essential for achieving convergence in simulation results towards the basis set limit (or the continuum limit). The correlation energy outside the active space will be called the ``virtual correlation energy"~\cite{huggins_unbiasing_2022}. 
Compared to existing dynamic correlation methods, the uniqueness of our approach is that 
we can calculate the overlap, force bias, and local energy {\it of the entire single-particle space} using overlap evaluations between the trial and walkers only {\it within the active space}.

The original proposal in ref. \citenum{huggins_unbiasing_2022} had a significant overhead in evaluating the local energy using overlap, resulting in the cost being $\mathcal O(N^4)$ more expensive than the overlap evaluation.
Here, we show how to remove such an overhead using ideas from differentiable programming.
Furthermore, the original proposal did not provide implementation details, even for the overlap evaluation. We provide the complete details below.

We begin by writing the trial wavefunction as,
\begin{equation}
    \left|\Psi_{\textrm{T}}\right\rangle= \left|\Xi_{\mathrm{c}}\right\rangle\otimes\left|\Psi_{\textrm{T},a}\right\rangle\otimes |0_v\rangle
\end{equation}
where $\left|\Psi_{\textrm{T},a}\right\rangle$ is the MPS trial wavefunction within the active space with $N_a$ electrons, $\left|\Xi_{c}\right\rangle$ is a Slater determinant composed of occupied orbitals with $N_c$ electrons outside the active space (i.e., frozen core orbitals), and $|0_v\rangle$ is the vacuum state in the space of unoccupied orbitals (i.e., frozen virtual orbitals).
The trial wavefunction with the above form can always be expanded by,
\begin{equation}
    |\Psi_{\textrm{T}}\rangle = |\Xi_{\mathrm{c}}\rangle\otimes \sum_i c_i|\chi_{i}\rangle\otimes |0_v\rangle
\end{equation}
where $|\chi_i\rangle$ is the $i$-th Slater determinant within the Hilbert space of the active space.
We use this uncontracted form only for demonstration purposes and emphasize that our algorithm directly uses the MPS trial without converting it into a set of determinants. 

\paragraph{Overlap.}
With this form, the overlap between trial and walker reads 
\begin{equation}
\sum_i c_i \left\langle\Xi_c\chi_i 0_v|\psi\right\rangle
=
\sum_i c_i 
\operatorname{det}\left(\left(\begin{array}{l|l}
\Xi_c & \mathbf{0} \\
\hline \mathbf{0} & \chi_i\\
\hline \mathbf{0} & \mathbf{0} \\
\end{array}\right)^\dagger\left(\begin{array}{l}
\psi_c \\
\psi_a \\
\psi_v 
\end{array}\right)\right)\label{eq:ovlp_vir}
\end{equation}
which equal to
\begin{equation}
\left\langle\Xi_c \chi_i|\psi\right\rangle=\operatorname{det}\left(\begin{array}{c}
\Xi_c^\dagger \psi_c \\
\chi_i^\dagger \psi_a
\end{array}\right)\label{eq:ovlp_vir2}
\end{equation}
where $\psi_c$ and $\psi_a$ are $N_a+N_c$ column MO coefficient matrices over the molecular orbital basis.
$\Xi_c$ is diagonal with ones up to the number of core electrons and zeros elsewhere.
While virtual orbital degrees of freedom no longer appear, this form still contains core degrees of freedom.

To further remove the core degrees of freedom, we perform SVD on
\begin{equation}
    \Xi_{\mathrm{c}}^\dagger\psi_c = U_{\mathrm{c}}\Sigma_{\mathrm{c}}V^\dagger_{\mathrm{c}},
\end{equation}
where $U_{\mathrm{c}} \in \mathbb{C}^{N_c \times N_c}$ and $V_{\mathrm{c}} \in \mathbb{C}^{(N_a+N_c) \times N c}$.
Then, we define new unitary matrices $U \in \mathbb{C}^{(N_a+N_c) \times (N_a+N_c)}$ and $V \in \mathbb{C}^{(N_a+N_c) \times (N_a+N_c)}$ by padding orthonormal vectors to $U_c$ and $V_c$,
\begin{equation}
U = 
\begin{pmatrix}
  U_c
  & \rvline & \mathbf{0} \\
\hline
  \mathbf{0} & \rvline &
    \mathbf{I}
\end{pmatrix},
V = \begin{pmatrix}
  V_c
  & \rvline & V'
\end{pmatrix},
\end{equation}
With these new unitary matrices, we can rewrite the overlap in  Eq.~\eqref{eq:ovlp_vir2} as
\begin{align}
\left\langle\Xi_c \chi_i|\phi\right\rangle &=\operatorname{det}\left( U^\dagger \left(\begin{array}{l|l}
\Xi_c & \\
\hline & \chi_i
\end{array}\right)^\dagger \left(\begin{array}{l}
\psi_c \\
\psi_a
\end{array}\right) V\right) \\ &= 
\operatorname{det}(\Sigma_c)\operatorname{det}(\chi_i^\dagger \tilde{\psi}_a)\operatorname{det}(R)
\end{align}
where $\tilde{\psi}_a$ is the normalized
Slater determinant within the active space,
and $\operatorname{det}(R)$ is the normalization matrix obtained by performing QR decomposition of the matrix $\psi_a V'$.
Therefore, computing the overlap between the trial and Slater determinant in the total space only requires the evaluation of the overlap between the MPS trial and SD in the active space,
\begin{equation}
    \langle\Psi_{\textrm{T}}|\psi\rangle = \operatorname{det}(\Sigma_c R)\langle \Psi_{\textrm{T},a}|\tilde{\psi}_a\rangle\label{eq:ovlp_virt_simple}
\end{equation}
which we already know how to compute, as described in Sec.~\ref{sec:mps_afqmc}.
We will show that the computation of force bias and local energy can be performed by differentiating overlap. 

\paragraph{Force bias.}
\begin{equation}\label{eq:fb_fd}
\frac{\left\langle\Psi_{\textrm{T}}\left|\hat{v}_\gamma\right| \psi\right\rangle}{\left\langle\Psi_{\textrm{T}} |\psi\right\rangle}=\frac{\left .\frac{\partial\langle\Psi_{\textrm{T}}|e^{\lambda_\gamma\hat{v}_\gamma}| \psi\rangle}{\partial \lambda_\gamma}\right|_{\lambda_\gamma=0}}{\langle\Psi_{\textrm{T}} |\psi\rangle}=\left .\frac{\partial
\ln{\langle\Psi_{\textrm{T}}|e^{\lambda_\gamma\hat{v}_\gamma}| \psi\rangle}}{\partial \lambda_\gamma}\right|_{\lambda_\gamma=0}
\end{equation}
where $e^{\lambda_\gamma\hat{v}_\gamma}| \psi\rangle$ is a SD by rotating the walker's SD $|\psi\rangle$.
Therefore, the function $\ln{\langle\Psi_{\textrm{T}}|e^{\lambda_\gamma\hat{v}_\gamma}| \psi\rangle}$ in Eq.~\eqref{eq:fb_fd} corresponds to the overlap value between trial and the rotated SD (see Eq.~\eqref{eq:ovlp_virt_simple} and the derivative at $\lambda_\gamma=0$ can be computed with finite difference approximations or even algorithmic differentiation~\cite{mahajan2023response}.
When using algorithmic differentiation, one could obtain a $\mathcal O(N_\gamma)$ speed-up compared to the na{\"i}ve derivative implementation. 
\paragraph{Local energy.}
In Sec.~\ref{sec:loce}, we introduced using MPO to compute the local energy, which is not possible within the context of considering both active space, frozen cores, and frozen virtuals.
The local energy can be computed similarly to that of force bias without using MPO.
\add{
    Previously, the derivative method was also proposed to 
compute the local energy in the context of real space QMC~\cite{filippi2016simple},
and shifted-contour AFQMC~\cite{baer2001ab,baer2002shifted}. 
}
Especially for the one-body term, we  replace $\hat{v}_{\gamma}$ with $\sum_{pq}h_{pq}a_p^\dagger a_q$,
\begin{equation}
    E_1 = \frac{\partial\ln\langle\Psi_{\textrm{T}}|e^{\lambda\sum_{pq}h_{pq} a_p^\dagger a_q}| \psi\rangle}{\partial\lambda}\big|_{\lambda=0}
\end{equation}
which requires
a derivative of an overlap evaluation between $|\Psi_{\textrm{T}}\rangle$ and a SD.
For the two-body term,
\begin{equation}
    E_2 = \frac{1}{2}\frac{\left\langle\Psi_{\textrm{T}}\left|\sum_{\gamma}\sum_{pqrs}L_{pq}^\gamma L_{rs}^{\gamma} a_p^\dagger a_q a_r^\dagger a_s\right| \psi\right\rangle}{\left\langle\Psi_{\textrm{T}} |\psi\right\rangle}\\
\end{equation}
where 
\begin{align}
    &\langle\Psi_{\textrm{T}}|\sum_{pqrs}L_{pq}^\gamma L_{rs}^{\gamma} a_p^\dagger a_q a_r^\dagger a_s| \psi\rangle \\\nonumber &= \left .\frac{\partial^2 \langle\Psi_{\textrm{T}}|e^{\lambda \sum_{pq}L_{pq}^\gamma a_p^\dagger a_q} | \psi\rangle}{\partial\lambda^2}\right|_{\lambda=0}\\\nonumber.
\end{align}
Then, the mixed partial derivative can be computed using a finite difference approximation.
\begin{table}[ht]
\centering
\begin{tabular}{lcc}
\hline
Method & Cost Relative to Overlap \\
\hline
Overlap & $\mathcal O(1)$ \\
Force bias$^\dagger$ & $\mathcal O(1)$ \\
One-body energy & $\mathcal O(1)$ \\
Two-body energy & $\mathcal O(N_\gamma)$ \\
\hline
\end{tabular}
\caption{Computational costs of the subroutines relative to the cost of overlap. $^\dagger$ We assume an implementation based on algorithmic differentiation.}
\label{tab:relative_computational_costs}
\end{table}

We note that these apply to arbitrary trial wavefunctions. For instance, they will be critically important in accelerating quantum-classical AFQMC algorithms with matchgate shadows where classical post-processing has been claimed to scale \add{$\mathcal O(N^{8})$~\cite{kiser_classical_2023,Huang2024Apr} per walker per shadows sample, which can be reduced to $\mathcal O(N^{4-5})$ with our approach.}
\section{Complexity of Gaussian-MPS overlaps}\label{sec:BQP}

In general, the classical simulation of quantum computations (and, more generally, the evaluation of tensor networks) is hard.
However, there are certain classes of computations (and tensor networks) that are classically tractable.
The most familiar such class consists of stabilizer states and circuits \cite{gottesman_heisenberg_1998} (and their generalization, affine tensor networks~\cite{cai_clifford_2018}), also known as Clifford circuits.
Another class consists of Gaussian states and circuits \cite{terhal_classical_2002} (and their generalization, matchgate tensor networks~\cite{cai_holographic_2007}).
The last known class consists of states and circuits expressible by bounded-treewidth tensor networks~\cite{markov_simulating_2008}. 
An MPS with a bounded bond dimension is one example of a bounded-treewidth tensor network.
% TODO: Note completeness of list?

All three classes can be extended to universality by adding certain gates or states outside the class (e.g., ``Clifford+T''). The classical simulation time then typically scales exponentially with the number of such additional resources.
Importantly, the ``magic'' resource extending each class can be taken from either of the other two classes.
Here we show one such instantiation of this idea, that, under standard complexity-theoretic assumptions, it is hard to compute the overlap of a (pure) Gaussian state $\ket{\psi_{\mathrm{G}}}$ with a state $\ket{\psi_{\mathrm{NG}}}$ that is both a stabilizer state and an MPS with constant bond dimension. 
This is stated formally in \cref{thm:overlap-hardness}.

\begin{theorem}\label{thm:overlap-hardness}
Let $\ket{\psi_{\mathrm{G}}}$ be a normalized Gaussian state and $\ket{\psi_{\mathrm{NG}}}$ be either a normalized stabilizer state or a normalized MPS with bond dimension $O(1)$, both of $n$ qubits.
Then estimating $\left|\braket{\psi_{\mathrm{G}} | \psi_{\mathrm{NG}}}\right|$ to within multiplicative error $<1/\sqrt{2}$ is $\#\P$-hard under \P-reduction.
\end{theorem}

After some introductory material in~\cref{sec:background}, we will prove \cref{thm:overlap-hardness} in~\cref{sec:hardness-proof}.

\subsection{Background}\label{sec:background}

\begin{definition}[Affine]
An \emph{affine tensor network} is a tensor network; all of its tensors are \emph{affine} tensors.
An \emph{affine tensor} $f: {\{0, 1\}}^r \mapsto \mathbb C$ is one defined by a constant $\lambda \in \mathbb C$, a matrix $A \in \mathbb Z_2^{l \times (r + 1)}$, and a symmetric matrix $Q \in \mathbb Z^{(r + 1)\times(r + 1)}$, such that
    \begin{align}
        f(\mathbf x) &= 
        \lambda \cdot 
        \chi_{A \cdot (\tilde{\mathbf x}, 1)} 
        i^{\tilde{\mathbf x}^T Q \tilde{\mathbf x}},
    \end{align}
    where $\tilde{\mathbf x} = (x_1, \ldots, x_r, 1)$ is just $\mathbf x$ with 1 appended and
    \begin{align}
        \chi_{A\mathbf x} &= \begin{cases}
        1, & A\mathbf x = \mathbf 0, \\
        0, & \text{otherwise}
        \end{cases}
\end{align} is the indicator function of the affine relation defined by $A$. 
A (pure) stabilizer state (also known as a Clifford state) is simply an affine tensor network that has unit Euclidean norm (i.e., is a normalized quantum state).
\end{definition}

\begin{definition}[Matchgate]
A \emph{matchgate} tensor $f: {\{0, 1\}}^r \mapsto \mathbb C$ is one defined by a weighted planar graph $G = (V, E, w)$, with $r$ vertices identified as ``external'' such that
\begin{align}
f(\mathbf x)
&=
\sum_{M \in \mathrm{PerfectMatchings}(G^{\mathbf x})} \prod_{e \in M} w(e),
\end{align}
where $G^{\mathbf x}$ is the subgraph of $G$ induced by the vertices $\left\{v: x_v = 0\right\}$ and $\mathrm{PerfectMatchings}(H)$ is the set of perfect matchings of a graph $H$.
A (pure) Gaussian state is simply a matchgate tensor that has unit Euclidean norm.
A Slater determinant is a Gaussian state with the further restriction that its support is on a single Hamming level, i.e. $\norm{\mathbf x}_1 = k$ for some $k$.
\end{definition}
Importantly, because of the planarity requirement, and unlike with affine tensors, the ordering of the legs of a matchgate tensor matters.
Every 2-qubit unitary matchgate can be specified by two 1-qubit unitaries $A$ and $B$ with $\det A = \det B$.
Specifically, $G(A, B)$ acts as $A$ on the even subspace $\{\ket{0, 0}, \ket{1, 1}\}$ and as $B$ on the odd subspace $\{\ket{0, 1}, \ket{1, 0}\}$.
We will not define the treewidth of the tensor network; see \cite{ogorman_parameterization_2019} for more details.
It will suffice to note that an MPS of bond-dimension $d$ has treewidth $O(d)$.

\subsection{Proof of overlap hardness}\label{sec:hardness-proof}
The main idea for showing the hardness of computing the overlap of interest is to show that it is essentially equivalent (up to scaling and with some blowup in the number of qubits) to computing the output amplitude of an arbitrary quantum circuit.
This is captured in \cref{lem:overlap-amplitude}.

\begin{lemma}\label{lem:overlap-amplitude}
    Given a quantum circuit $C$ on $n$ qubits and containing $m=\poly(n)$ 1- and 2-qubit gates, there is
    a normalized Gaussian state $\ket{\psi_{\mathrm{G}}}$, a normalized state $\ket{\psi_{\mathrm{NG}}}$, and an $s = \poly(n)$ such that 
    \begin{align}
        \braket{0^n | C | 0^n} &= 4^{s} \braket{\psi_{\mathrm{G}} | \psi_{\mathrm{NG}}},
    \end{align}
    $\ket{\psi_{\mathrm{G}}}$ and $\ket{\psi_{\mathrm{NG}}}$ are states of $n'=\poly(n)$ qubits,
    and $\ket{\psi_{\mathrm{NG}}}$ is both a stabilizer state and an MPS with bond dimension $O(1)$.
\end{lemma}

\begin{proof}
First, we define a new circuit $C_{\mathrm{NN}}$ that is equivalent to $C$ except that all of its 2-qubit gates only act on neighboring qubits, for some linear ordering of the qubits.
The simplest way of doing so is, for every gate acting on qubits $i$ and $j$, to insert a series of swap gates bringing the state on qubit $j$ to qubit $i + 1$. This introduces at most $n-1$ swap gates, so that $C_{\mathrm{NN}}$ has at most $m_{\mathrm{NN}} \leq nm$ gates. Note that the final ordering of the qubits does not matter because we are only calculating the overlap of the output state with $\ket{0^n}$.

Next, we define a circuit $C_{\mathrm{swp}}$ that compiles each 2-qubit gate into a sequence of matchgates and swaps (on the same two qubits).
There are many ways of doing this.
    One way is to start with the compilation of each gate of $C_{\mathrm{NN}}$ into single-qubit gates and (at most 3) CNOTs.
Then turn each CNOT into a CZ by inserting a Hadamard on the target qubit on each side of the CNOT.
Lastly, replace each CZ by a swap and a fermionic swap.
    The resulting circuit $C_{\mathrm{swp}}$ still acts on the original $n$ qubits and has $m_{\mathrm{swp}} = O(m_{\mathrm{NN}}) = O(nm)$ gates.
Note that all single-qubit gates and the two-qubit fermionic swap are matchgates, so the only non-matchgates contained in $C_{\mathrm{swp}}$ are the swaps.

In the last transformation, we introduce the magic state 
\begin{align}
\ket{\mathrm{M}} &= \frac{1}{4} \left(\ket{0,0,0,0} + \ket{0,1,0,1} + \ket{1,0,1,0} + \ket{1, 1, 1, 1}\right).
\end{align}
Importantly, $\ket{\mathrm{M}}$ is both a stabilizer state and an MPS of constant bond dimension.
Let $s$ be the number of swap gates in $C_{\mathrm{swp}}$.
We construct a new circuit $C_{\mathrm{mgc}}$ that acts on $n_{\mathrm{mgc}} = n_{\mathrm{swp}} + 4s$ qubits in the following manner.
Start with $C_{\mathrm{swp}}$.
Identify each swap gate in $C_{\mathrm{swp}}$ with a 4-qubit ancilla register.
Then replace the swap gate in $C_{\mathrm{swp}}$ with the gadget described in \cref{fig:swap-gadget}, consisting of a series of fermionic swaps and two $G(H, H)$ gates.
\begin{figure}
\centering
% \includesvg[width=1.2\linewidth,pretex=\relscale{1.0}]{swap_gadget.svg}
\includegraphics[width=0.4\textwidth]{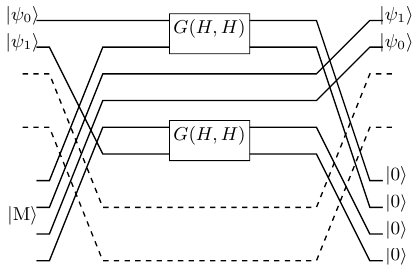}
    \caption{Gadget for implementing a swap gate using only matchgates and the magic state $\ket{\mathrm{M}}$.
    Every intersection of two-qubit wires indicates a \emph{fermionic} swap.
    The dashed lines indicate the arbitrary number of qubits between the two to be swapped and the ancilla register containing the magic state.
    }
\label{fig:swap-gadget}
\end{figure}
If we start with the magic state $\ket{\mathrm{M}}$ on the ancilla register and project at the end onto $\ket{0^4}$, the overall effect is to implement the swap, with an additional scale factor.

To see how the gadget works, let us zoom in on the part in between the swaps.
    It will be sufficient to show that it works for computational basis states $\ket{\psi_0, \psi_1} = \ket{x_0, x_1}$.
    \begin{align}
    &\left(\bra{0, 0} \otimes I_2 \otimes \bra{0, 0}\right)    \left(G(H, H) \otimes I_2 \otimes G(H, H) \right) \ket{x_0} \ket{\mathrm{M}} \ket{x_1}
    \\
    &= 
        2^{-1} 
        \hspace{-1em}
        \sum_{y_1, y_4 \in \{0, 1\}}
        \hspace{-1em}
         \left(\bra{y_1, y_1} \otimes I_2 \otimes \bra{y_4, y_4}\right)     
         \ket{x_0} \ket{\mathrm{M}} \ket{x_1}
    \\
    &=
        2^{-2} \hspace{-2em}
        {\sum_{y_1, y_2, y_3 y_4 \in \{0, 1\}}}
        \hspace{-2em}
        \left(\bra{y_1, y_1} \otimes I_2 \otimes \bra{y_4, y_4}\right)     
        \ket{x_0} \ket{y_2, y_3, y_2, y_3} \ket{x_1}
    \\
    &=
   2^{-2} 
       \left(\bra{x_0, x_0} \otimes I_2 \otimes \bra{x_1, x_1}\right)     
       \ket{x_0} \ket{x_0, x_1, x_0, x_1} \ket{x_1}
   \\
   &=
   2^{-2} \ket{x_1, x_0},
    \end{align}
    where we used the fact that
    \begin{align}
    G(H, H) \ket{0, 0} &= 2^{-1/2} \sum_{x \in \{0, 1\}} \ket{x, x}.
    \end{align}
    Overall, this implies
    \begin{align}
    \braket{0^n | C | 0^n}
    &=
    \braket{0^n | C_{\mathrm{NN}} | 0^n} \\
    &=
    \braket{0^n | C_{\mathrm{swp}} | 0^n}\\
    &=
    4^s \braket{0^{n + 4s} | C_{\mathrm{mgc}} | 0^n, \mathrm{M}^s}.
    \end{align}
    Defining
    \begin{align}
    \ket{\psi_{\mathrm{NG}}} &= \ket{0^n, \mathrm{M}^s},\\
    \ket{\psi_[\mathrm{G}} &= C_{\mathrm{mgc}}^{\dagger} \ket{0^{n + 4s}} 
    \end{align}
    completes the proof.
\end{proof}

Finally, we get to the proof of the main result.

\begin{proof}[Proof of \cref{thm:overlap-hardness}]
We reduce the problem of estimating the output probability of a quantum circuit.
    Specifically, we show that polynomial-time algorithm that produces an approximation to the magnitude of the overlap $\left|\braket{\psi_{\mathrm{G}}  | \psi_{\mathrm{NG}}}\right|$ with multiplicative error $<1/\sqrt{2}$ can be used to construct a polynomial-time algorithm to produce an approximation to the output probability ${\left|\braket{0^n | C | 0^n}\right|}^2$ with multiplicative error $<1/2$, which is $\#\P$-hard [Corollary 9~\cite{hangleiter_computational_2023}].

Consider a circuit $C$ on $n$ qubits and containing $m=\poly(n)$ gates.
    By \cref{lem:overlap-amplitude}, we can efficiently construct normalized states $\ket{\psi_{\mathrm{G}}}$ and $\ket{\psi_{\mathrm{NG}}}$ such that $\braket{0^n | C | 0^n} = 4^s \braket{\psi_{\mathrm{G}} | \psi_{\mathrm{NG}}}$, where $s=\poly(n)$ is known from the contsruction.
Suppose that we can produce an estimate $\hat{v} > 0$ to $\left|\braket{\psi_{\mathrm{G}} | \psi_{\mathrm{NG}}}\right|$ with multiplicative error $<1/\sqrt{2}$, i.e.,
    \begin{align}
        \frac{1}{\sqrt{2}} \left|\braket{\psi_{\mathrm{G}} | \psi_{\mathrm{NG}}}\right|
        &<
        \hat{c} < 
        \sqrt{2} \left|\braket{\psi_{\mathrm{G}} | \psi_{\mathrm{NG}}}\right|
        .
\end{align}
    Then $\hat{v}' = 16^s \hat{v}^2$ is an estimate of ${\left|\braket{0^n | C | 0^n}\right|}^2$ with multiplicative error $<1/2$:
    \begin{align}
        & \frac{1}{2} {\left|\braket{0^n | C | 0^n}\right|}^2
       = {\left(
        \frac{1}{\sqrt{2}} 4^s \left|\braket{\psi_{\mathrm{G}} | \psi_{\mathrm{NG}}}\right|
       \right)}^2 \\
       & < 16^s \hat{v}^2 = \hat{v}'
       \\
        & < 
        \sqrt{2} \cdot 4^s {\left|\braket{\psi_{\mathrm{G}} | \psi_{\mathrm{NG}}}\right|}^2
        =
        {2} {\left|\braket{0^n | C | 0^n}\right|}^2.
    \end{align}
In other words, the ability to compute an estimate $\hat{v}$ to the overlap with multiplicative error $<1/\sqrt{2}$ yields the ability to compute an estimate $\hat{v}'$ to the output probability with multiplicative error $<1/2$, which is $\#\P$-hard.
\end{proof}
\section{Results and Discussion}\label{sec:results}
The focus now shifts to investigating practical implementation and performance evaluations of MPS-AFQMC.
We proved the combination of MPS and AFQMC is \#P-hard but have 
several classical heuristics to implement MPS-AFQMC.
How the formal complexity limits average-case chemistry applications is unclear, and this is what we hope to investigate numerically. 
% Specifically, we seek answers to the following questions:
% \begin{enumerate}[label=(\roman*)]
% \item What is the tradeoff between accuracy and computational cost in each of the proposed heuristics? Which one should one use for practical calculations?
% \item How do the approximations of SD-to-MPS and the overlap calculations affect the final AFQMC energies?
% \item In practical calculations, how does the approach perform across systems, particularly in variations of system sizes, two-dimensional layouts, and large basis sets?
% \end{enumerate}
% Here, we present systematic evaluations of these heuristics to answer these questions.

The one- and two-dimensional hydrogen lattice models display rich many-body effects~\cite{simons_collaboration_on_the_many-electron_problem_ground-state_2020} and serve as a testbed for electronic structure methods where strong correlation is controlled by inter-hydrogen distances~\cite{stella2011strong,simons_collaboration_on_the_many-electron_problem_towards_2017}.
For instance, the \textit{ab initio} hydrogen lattice in a minimal basis closely resembles the Hubbard model, augmented with long-range Coulomb interactions. Larger basis sets employ many orbitals per site, as in typical materials simulations.
We consider these systems below for numerical explorations of our heuristics.

We utilize \textsc{PySCF}~\cite{sun_recent_2020} to obtain the electron integrals for the Hamiltonian in Eq.\eqref{eq:qcham} and conduct preliminary electronic structure calculations such as Hartree-Fock. These integrals are then supplied to subsequent DMRG and AFQMC calculations. 
DMRG calculations are performed with the \textsc{Renormalizer}~\cite{renormalizer,ren2022time} package.
The MPS-AFQMC algorithm is developed in the AFQMC package \textsc{ipie}~\cite{ipie,malone2022ipie,jiang2024improved}.

\subsection{Comparisons of the SD-to-MPS strategies}\label{sec:overlapconv}
\begin{figure}
    \centering
    \includegraphics[width=0.5\textwidth]{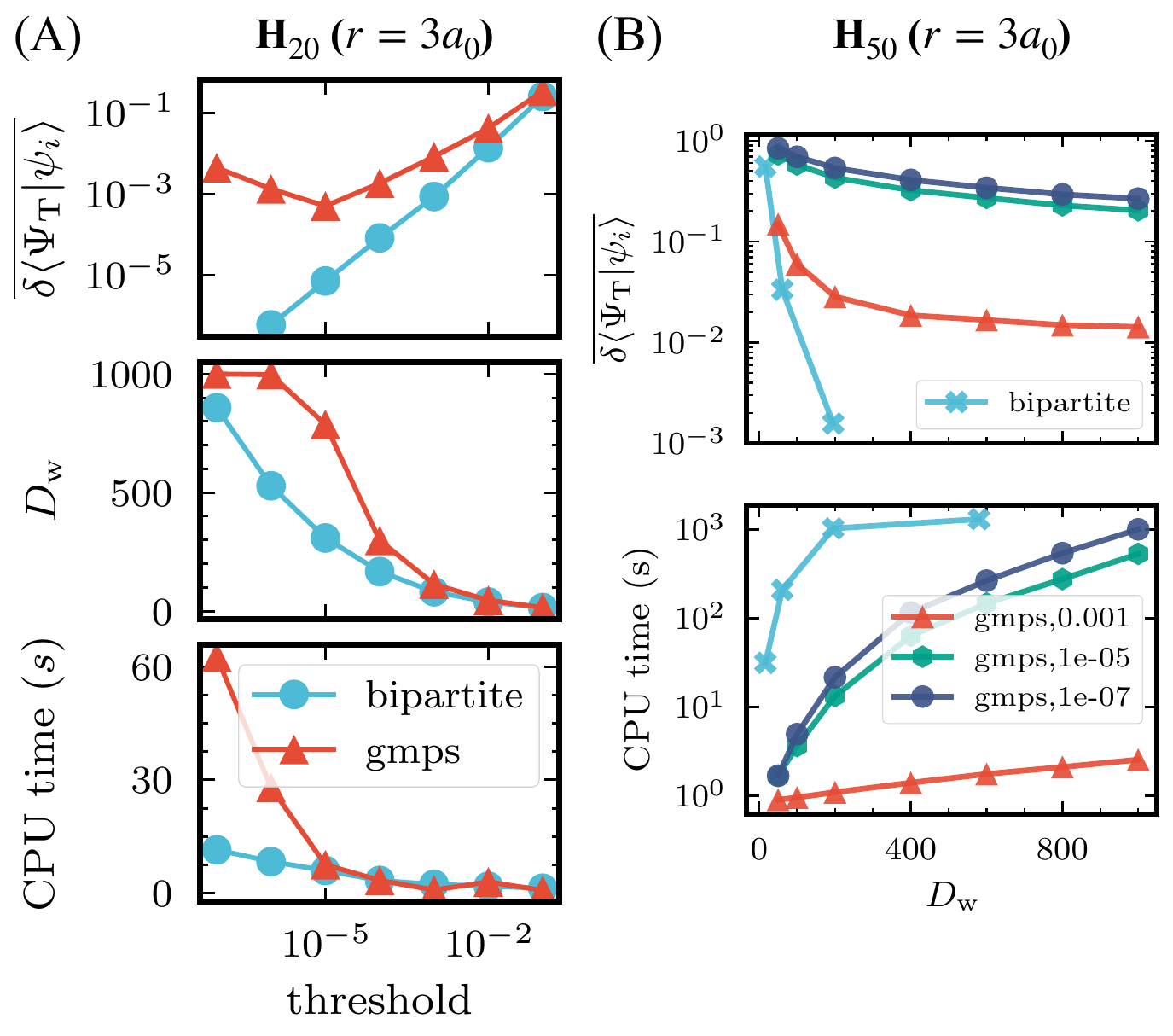}
    \caption{\textbf{Comparing performance of GMPS-to-MPS, bipartite strategies for the overlap calculation across different sybstems}. (A) The relative error, the walker MPS's bond dimension and CPU time of performing SD-to-MPS, averaged over 3200 equilibrated walkers of H$_{20}$, as a function of the threshold for the bipartite approach ($\epsilon_{\textrm{trunc}}$) and the GMPS-to-MPS approach ($\epsilon_{\textrm{ad}}$). The maximal bond dimension of the GMPS-to-MPS approach is set as 1000. The reference overlap values are taken from the bipartite approach with truncation threshold $10^{-7}$.
    (B) The averaged relative error and averaged CPU time as a function of walker MPS's bond dimension for H$_{50}$. The reference values are taken from the bipartite approach with truncation threshold $10^{-4}$. The bipartite results are obtained with $\epsilon_{\textrm{trunc}}\in[10^{-1}, 10^{-2}, 10^{-3}, 10^{-4}]$.}
    \label{fig:H20}
\end{figure}
As mentioned, in ph-AFQMC calculations, the overlap between trial and walker wavefunctions is the central quantity to calculate when performing the phaseless approximation.
In this section, we compare the performance of different strategies 
for this task in terms of accuracy and efficiency. 
The approaches we consider are GMPS-to-MPS, bipartite, and perfect sampling, 
\add{as outlined in Sec.~\ref{sec:sd2mps_White} 
and Appendix.~\ref{app:sd2mps_others}.   
}
To see this, in Fig.~\ref{fig:H20}, we employ one-dimensional hydrogen chains with varying system sizes and inter-atomic distances.

Both bipartite and GMPS-to-MPS approaches involve converting the walker to MPS with some bond dimension $D_\textrm{w}$ before computing the overlap with the MPS trial. We aim to understand the tradeoff between accuracy and cost set by different $D_\textrm{w}$.
First, we compare these two approaches for \ce{H20} in a minimal basis with $r=3a_0$ ($a_0$ represents the Bohr radius). 
In Fig.~\ref{fig:H20}(A), we investigate the impact of the threshold $\epsilon_{\textrm{trunc}}$, which denotes the coefficient truncation threshold in the bipartite method (as defined in Eq.\eqref{eq:bipar_coeff}), and the threshold $\epsilon_{\textrm{ad}}$, which adaptively controls the block size in the GMPS-to-MPS conversion.
Tightening the threshold for the bipartite strategy enhances the precision and increases the bond dimension of the walker MPS, albeit with a slightly increased time cost.
Unlike the bipartite approach, which utilizes a single parameter to control the conversion accuracy, the GMPS-to-MPS approach relies on two quantities: the adaptive threshold, which determines the block size and hence the number of gates, and the maximum allowed bond dimension during compression.
Here, we cap the maximum bond dimension at 1000 and tune only the adaptive threshold $\epsilon_{\textrm{ad}}$. 
Tightening the threshold initially leads to a decrease in error but eventually to an increase in error. 
This arises because a smaller threshold leads to a larger block size and, hence, more gates to be applied, resulting in a higher bond dimension for achieving the same level of accuracy. 
For $\epsilon_{\textrm{ad}} \leq 10^{-5}$, the GMPS-to-MPS approach strikes a good balance between the cost and accuracy. The bipartite method demonstrates greater cost-effectiveness and accuracy for H$_{20}$ overall.

For larger systems (e.g., H$_{50}$), the bipartite approach becomes significantly more expensive, as illustrated in Fig.~\ref{fig:H20}(B). 
The error and computational cost of the bipartite and GMPS-to-MPS approaches with different thresholds are analyzed as a function of the walker bond dimension in H$_{50}$.
We select $\epsilon_{\textrm{trunc}}$ values ranging from $0.1, 0.01, 0.001$ to $0.0001$ for the bipartite approach, where each threshold yields an MPS with a different bond dimension $D_\textrm{w}$. For the GMPS-to-MPS approach, we maintain the maximum allowed $D_\textrm{w}$ at 1000 and adjust the adaptive threshold $\epsilon_{\textrm{ad}}$. The reference values of the overlap are obtained from the bipartite approach with $\epsilon_{\textrm{trunc}}=0.0001$.
It is observed that for similar walker bond dimensions $D_\textrm{w}$, the bipartite approach still achieves higher accuracy but at the cost of substantially longer computational times. For example, when compared to the GMPS-to-MPS approach with $\epsilon_{\textrm{ad}}=0.001$ and $D_\textrm{w}=200$, the accuracy of the bipartite approach with $\epsilon_{\textrm{trunc}}=0.01$ is comparable, yet the computational time is approximately 1000 times slower.
% The complexity associated with maintaining and manipulating a high-accuracy representation of the system's state becomes a significant bottleneck, as compared to smaller systems.

Similarly to \ce{H20}, tightening the threshold for the GMPS-to-MPS approach results in decreased accuracy for a fixed bond dimension, $D_\textrm{w}$. In principle, a tighter threshold should lead to improved accuracy, provided that no approximation is made during the gate applications. However, when we impose a restriction on the walker's bond dimension for compressions during the gate applications, the GMPS-to-MPS method becomes more approximate for that given threshold. One needs a larger bond dimension to maintain a similar accuracy.
% Therefore, a balance need to reach between $D_\textrm{w}$ and the threshold $\epsilon_{\textrm{ad}}$.
% Further analysis of this effect will be presented in the following sections (\textit{eg.}, Figs.~(\ref{fig:chunk_ovlp}(B),\ref{fig:mps_msd_afqmc}).
GMPS-to-MPS approach with $\epsilon_{\textrm{ad}}=10^{-3}$ with $D_{\textrm{w}} \le 200$ seems to strike the balance between cost and accuracy needed for \ce{H50}. 
These observations indicate that the GMPS-to-MPS approach, while potentially less accurate than the bipartite approach, will likely offer a more scalable solution for larger systems.

\begin{figure}
    \centering
    \includegraphics[width=0.45\textwidth]{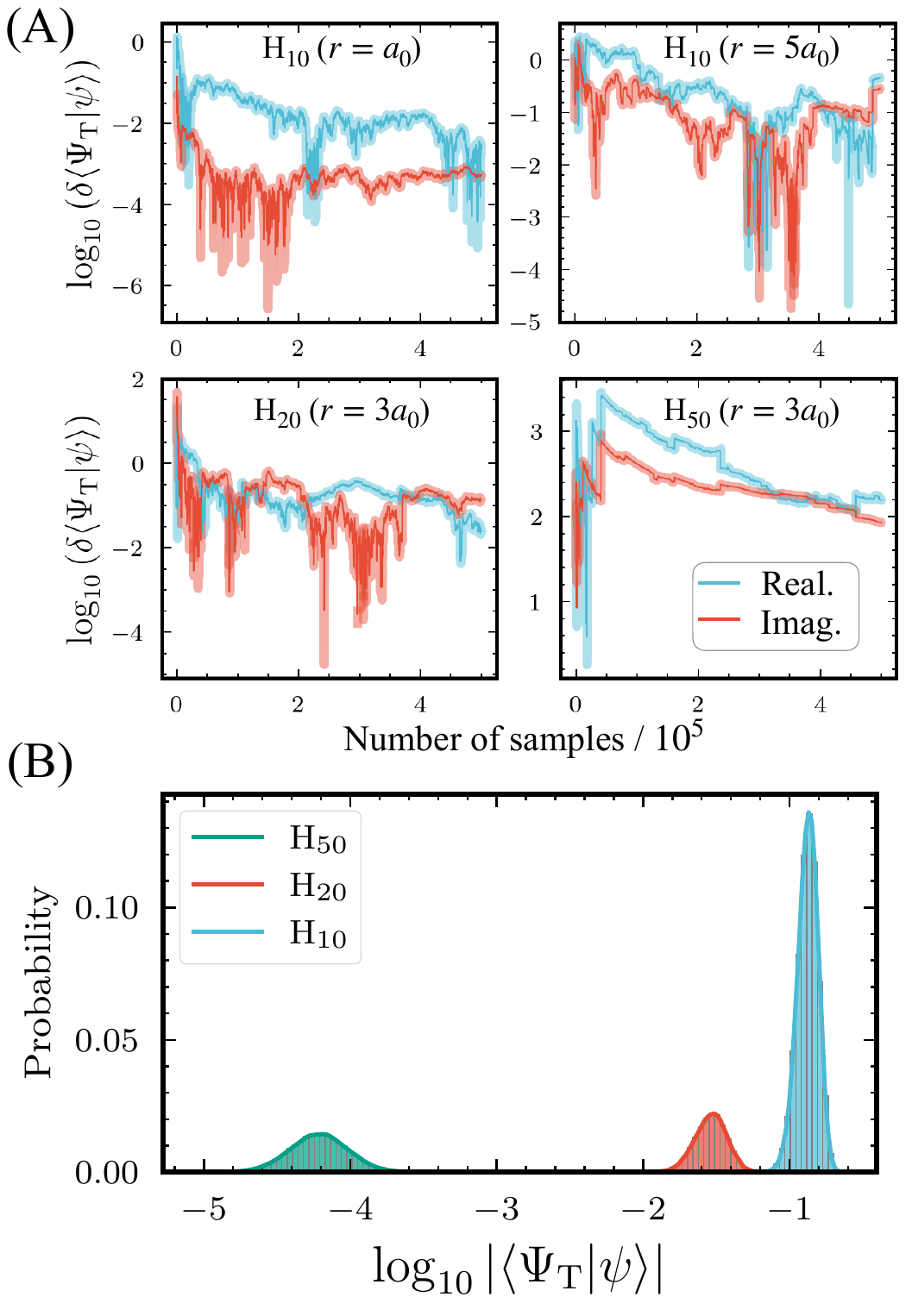}
    \caption{\textbf{(A) Overlap error from perfect sampling with different system sizes and (B) the overlap distribution for different system sizes.} 
    We use the nearest neighbor bond length of $r=3a_0$ in all cases.
    In (A), we show the relative error of the perfect sampling approach for an equilibrated walker of H$_{10}$, H$_{20}$, and H$_{50}$ as a function of number of samples. The walker state in H$_{50}$ is projected with a block size of $5$ (see Appendix.~\ref{sec:lesp} for more information) for demonstration.
    The presented distributions are derived from $6.4 \times 10^5$ equilibrated walkers. 
    For these analyses, the block-averaged distributions are calculated by averaging groups of 50 walkers per block.
}
    \label{fig:sampling}
\end{figure}
The third strategy circumvents the expensive conversion of the walker state to MPS and instead samples the overlap value using the perfect sampler. Following the formalism outlined in Sec.\ref{sec:strategy3}, we investigate the perfect sampling method by selecting an equilibrated walker from systems of varying sizes and assessing the relative error against the number of samples, as depicted in Fig.\ref{fig:sampling}(A).
We sample exactly from the probability distribution given by the walker matchgate state, which corresponds to Eq.~\eqref{eq:sample_from_walker}.
The results are similar to a strategy where the sampling is performed with the probability from the trial, which is not shown here for brevity.
This technique demands a significantly higher number of samples for larger systems due to the exponential decrease of overlap as system size increases. Specifically, the magnitude of the overlap we examined here —H$_{10}(r=a_0)$, H$_{10}(r=5a_0)$, H$_{20}(r=3a_0)$, and H$_{50}(r=3a_0)$—are on the order of $10^{-2}$, $10^{-3}$, $10^{-3}$, and $10^{-6}$, respectively.
We show the overlap distribution for systems of H$_{10}$, H$_{20}$ and H$_{50}$ in Fig.~\ref{fig:sampling}(B), where the overlap values become exponentially smaller as system size increases.

According to Sec.~\ref{sec:sd2mps_White}, the accuracy of SD-to-MPS via the GMPS approach is determined by two factors.
Although a tighter $\epsilon_{\textrm{ad}}$ leads to higher accuracy, it concurrently necessitates a substantially larger $D_{\textrm{w}}$ for convergence. As illustrated in Fig.~\ref{fig:balance}, 
a balanced choice employing a moderately small $\epsilon_\textrm{ad}$ and an appropriate $D_{\textrm{w}}$ appears to be the most effective.
Based on this, we chose the GMPS strategy for the rest of the paper.
Although careful optimization of different methods is beyond the scope of this work, it is worth revisiting this in the future.
Additional numerical details for these strategies are given in the Appendix.~\ref{app:sd2mps_others}.

\begin{figure}
    \centering
    \includegraphics[width=0.5\textwidth]{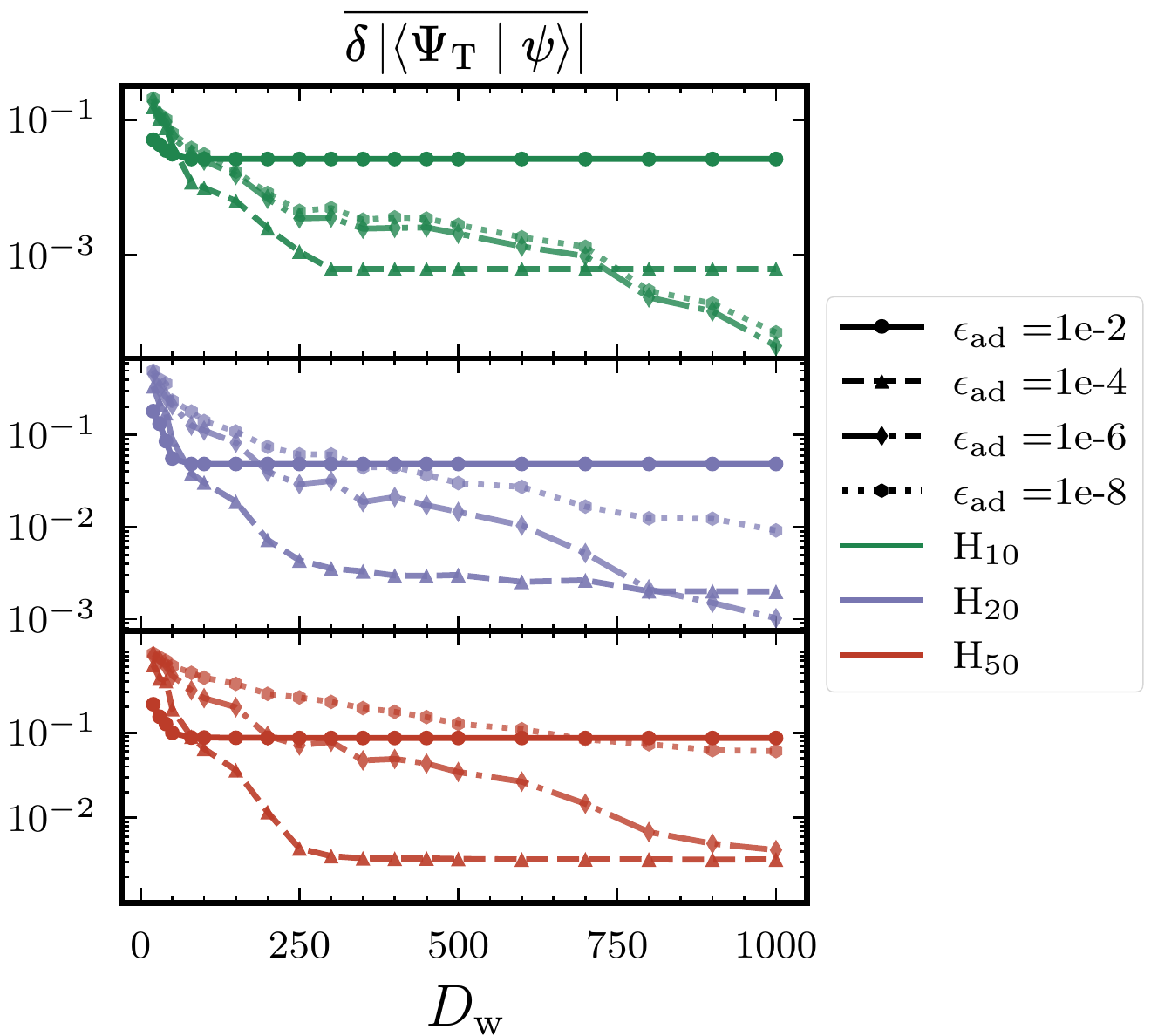}
    \caption{\textbf{Convergence of MPS-AFQMC overlap as a function of system size.} 
    All systems under study are characterized by an interatomic bond distance of $r=3a_0$ and employ the STO-6G basis set.
    The convergence analysis is performed using the mean relative error of $|\langle\Psi_\textrm{T}|\psi\rangle|$ over 100 walkers with respect to $\epsilon_{\textrm{ad}}$ and $D_\textrm{w}$.}
    \label{fig:balance}
\end{figure}

\subsection{Performance of MPS-AFQMC with the GMPS-to-MPS strategy}
Given the comparison between various strategies for implementing the overlap evaluation subroutine for AFQMC, we adopt the GMPS strategy as our primary approach.
In this section, we hope to benchmark how different convergence parameters in the GMPS strategy affect the final AFQMC energies.

\subsubsection{Case study of 1D \ce{H10} chain}
We test our MPS-AFQMC on a one-dimensional stretched H$_{10}$ chain.
We used a minimal basis set, STO-6G, and an interatomic distance of 3$a_0$ to get strong correlation effects. 
With such a small system, it is feasible to transform the MPS trial into a linear combination of multiple Slater determinants (MSD)~\cite{lee_externally_2021}, and the MSD-AFQMC results serve as the reference data for benchmarking MPS-AFQMC,
\begin{equation}
|\Psi_{\textrm{T}}\rangle = \sum_{\substack{s,\ |c_s|>\eta}} c_s|s\rangle,
\end{equation}
where the threshold $\eta$ is used to truncate the MSD expansion~\cite{lee_externally_2021}, which means that SDs with coefficients smaller than $\eta$ are omitted. Here, we set $\eta=10^{-14}$ to obtain a numerically exact representation of the original MPS trial.

We use this numerically exact representation to assess the accuracy of MPS-AFQMC when the overlap could be evaluated exactly.
In other words, in this test, we evaluate the necessary overlap exactly for a given MPS trial.
As illustrated in Fig.~\ref{fig:trial}, with the bond dimension of MPS increasing, the variational energies gradually improve towards the exact energy. 
When using the corresponding MPS as the trial of AFQMC, the deviation from exact full configuration interaction (FCI) energy decreases. The convergence of MPS-AFQMC energy to FCI is much faster than MPS itself, demonstrating the synergy between MPS and AFQMC, even within a small active space. The encouraging results presented in Fig.~\ref{fig:trial} are what is expected when SD-to-MPS conversion can be done exactly or very accurately.
However, in practice, the MPS-AFQMC must make approximations that might cause a large error in the overlap as we saw in Fig~\ref{fig:H20} and Fig.~\ref{fig:balance}.
\begin{figure}
    \centering
    \includegraphics[width=0.4\textwidth]{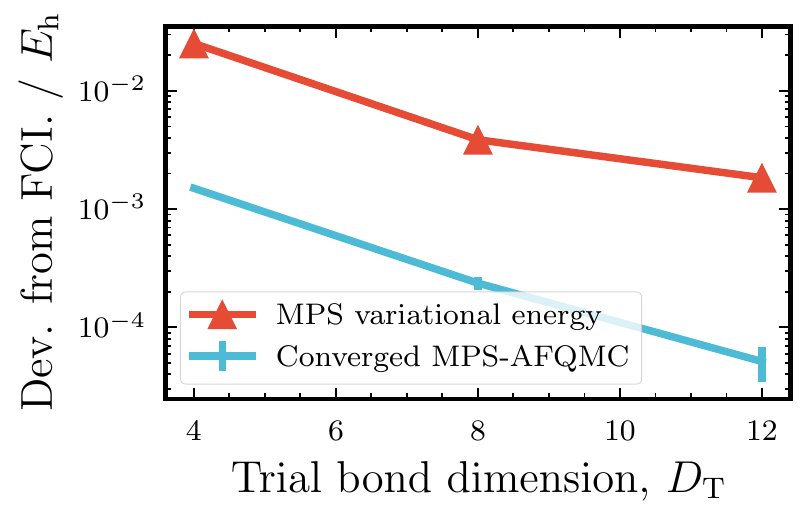}
    \caption{\textbf{Deviation of MPS trial and MPS-AFQMC from the FCI Energy of H$_{10}$.} The MPS-AFQMC calculations are conducted without any approximation in SD-to-MPS. We employ minimal STO-6G atomic orbitals with L\"owdin orthogonalization and the interatomic distance of $r=3a_0$.}
    \label{fig:trial}
\end{figure}

\begin{figure*}
    \centering
    \includegraphics[width=0.75\textwidth]{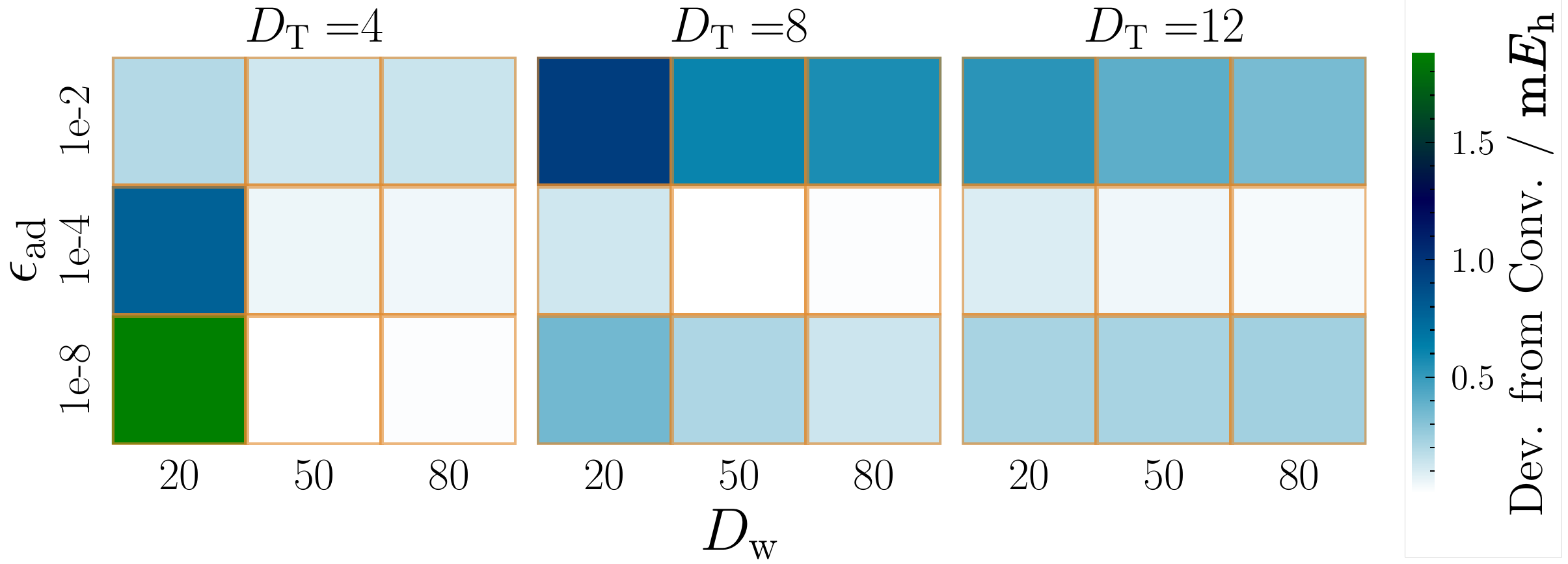}
    \caption{\textbf{Impact of approximating SD-to-MPS on AFQMC energy for H$_{10}$.} 
    The heatmaps detail the deviation of the MPS-AFQMC energy from the converged MPS-AFQMC value shown in Fig.~\ref{fig:trial}.
    }
    \label{fig:mps_msd_afqmc-1}
\end{figure*}

\begin{figure}
    \centering
    \includegraphics[width=0.45\textwidth]{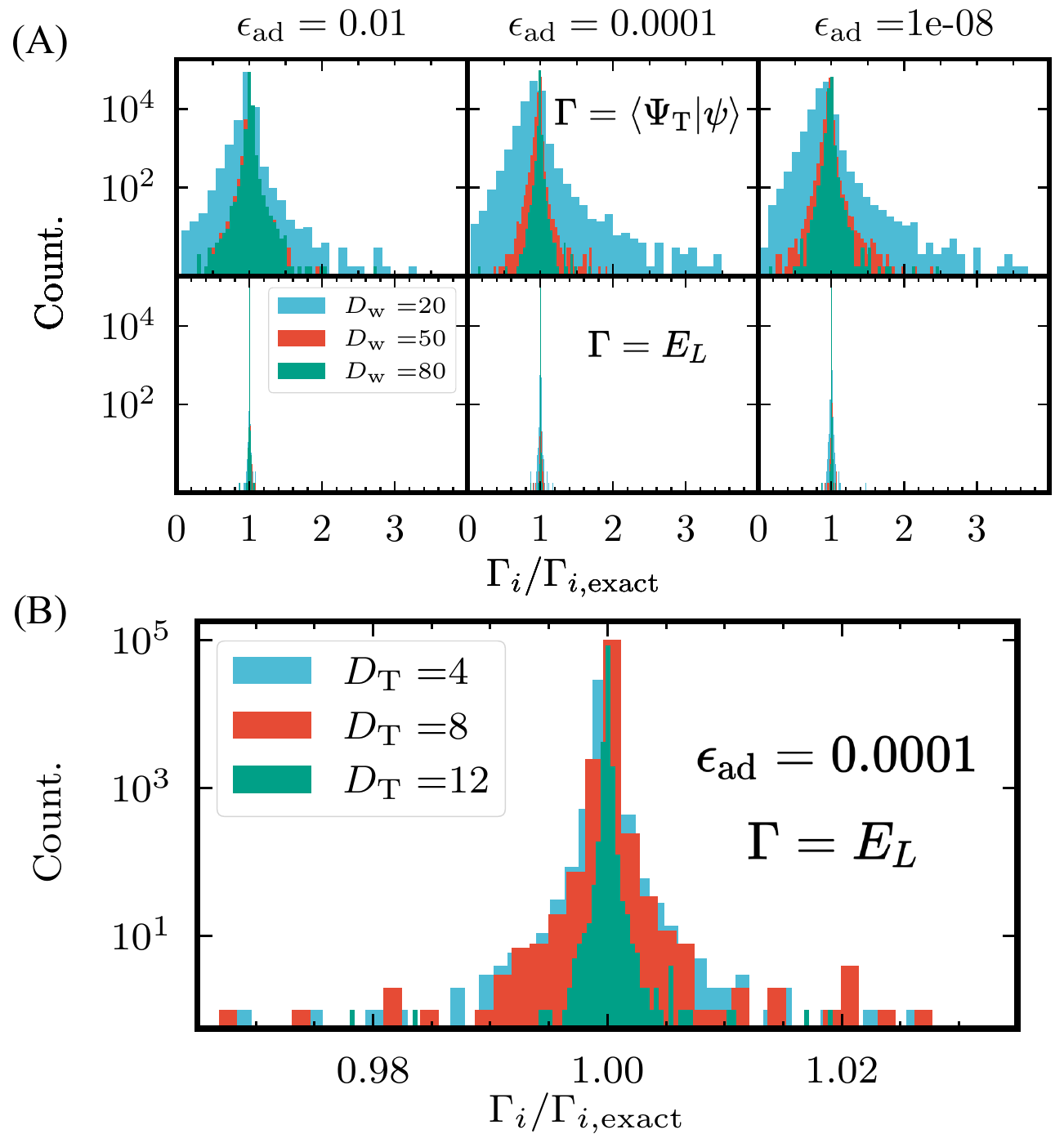}
    \caption{\textbf{Impact of approximating SD-to-MPS on overlap and local energy distribution of H$_{\textrm{10}}$.} 
    The trial are the same as in Fig.~\ref{fig:mps_msd_afqmc-1}, 
    where (A) shows the distribution of the overlap ratio using the MPS trial with $D_\textrm{T}=4$ and 
    (B) displays the distributions of the local energy ratio with trials at different $D_\textrm{T}$.
    The overlap ratio is defined as the ratio of the overlap calculated with a specific $D_{\textrm{w}}$ and $\epsilon_{\textrm{ad}}$ to the converged overlap value. 
    Similarly, the local energy ratio compares the local energy with its converged value. 
    These distributions are aggregated across the SD of 70,000 equilibrated walkers where spin-projection is used for walkers~\cite{purwanto2008eliminating}.
    }
    \label{fig:mps_msd_afqmc-2}
\end{figure}

Now, we investigate the impact of these approximations on the final ph-AFQMC energies using trials with varying quality.
% As demonstrated in Fig.~\ref{fig:chunk_ovlp}(B), for a given walker, the approximation of MPS-AFQMC with the GMPS strategy is determined by two parameters within the SD-to-MPS process: the adaptive threshold $\epsilon_{\textrm{ad}}$ and the maximal bond dimension $D_{\textrm{w}}$ allowed during the Givens gate application.
Within the GMPS-to-MPS method, we vary $\epsilon_{\textrm{ad}}$ and $D_{\textrm{w}}$ and compare the results using three different trials with different bond dimensions ($D_\textrm{T}$).
Here, the error is measured with respect to the corresponding MPS-AFQMC results without any approximations made to the overlap evaluation.
As $\epsilon_{\textrm{ad}}$ approaches 0 and $D_\text{w}$ approaches $2^{10}$ (maximum possible bond dimension for this problem), we expect to recover the exact MPS-AFQMC answer.
As shown in Fig.~\ref{fig:mps_msd_afqmc-1}, for a fixed $\epsilon_{\textrm{ad}}$, increasing $D_{\textrm{w}}$ leads to a smaller deviation from the exact MPS-AFQMC energies. 
A more stringent $\epsilon_{\textrm{ad}}$ gives a smaller deviation at its converged $D_\textrm{w}$, but it also requires a higher $D_{\textrm{w}}$ to achieve energy convergence.
This is consistent with our observation of the overlap convergence study in \cref{sec:overlapconv}.
% These results are shown in Fig.~\ref{fig:mps_msd_afqmc}(A), which also aligns with the findings previously discussed in Fig.\ref{fig:H20}(A)(B).
% It is crucial to note that, for a fixed $\epsilon_{\textrm{ad}}$, achieving convergence in $D_{\textrm{w}}$ is quicker when the trial quality is lower.
Similar to the overlap evaluation, practical applications of GMPS-to-MPS for MPS-AFQMC must balance $\epsilon_\text{ad}$ and $D_\text{w}$ to obtain good accuracy without increasing the cost too much.

In Fig.~\ref{fig:mps_msd_afqmc-2}, we examine the distributions of the overlap and local energy ratios over 70,000 walkers. Here, we use the ratio between the values from specific $D_{\textrm{w}}$ and $\epsilon_{\textrm{ad}}$ and their converged exact values. A narrower distribution centered around $1$ indicates smaller relative errors.
For a fixed adaptive threshold $\epsilon_{\textrm{ad}}$, an increase in $D_{\textrm{w}}$ leads to a narrower distribution of both overlap and local energy ratios, as expected. 
In contrast to the overlap, the local energies are typically much more robust against variations in both $D_{\textrm{w}}$ and $\epsilon_{\textrm{ad}}$.
This robustness largely explains why, despite sometimes significant deviations in overlap from the exact values, the final MPS-AFQMC energy deviation from the exact MPS-AFQMC remains small, with most values falling within the threshold of chemical accuracy (1 m$E_\mathrm{h}$). 
The deviation is even smaller for trials with better quality (higher $D_{\textrm{T}}$), as shown in Fig.~\ref{fig:mps_msd_afqmc-2}.
Similar observations are obtained for H$_{50}$, as shown in Fig.~\ref{fig:uhfmps}, where the trial is the UHF wavefunction.

The robustness of local energy calculations can be attributed to the zero-variance principle~\cite{assaraf_zero-variance_1999}.
More concretely, let us
consider the following error model:
\begin{align}
|\psi_\text{true}\rangle &= |\psi\rangle + |\epsilon\rangle\\
\langle \Psi_{\textrm{T}} | \psi_\text{true}\rangle &= \langle \Psi_{\textrm{T}}|\psi\rangle + 
\langle \Psi_{\textrm{T}}|\epsilon\rangle
\end{align}
where $|\epsilon\rangle$ quantifies the error in SD-to-MPS and the deviation from the true walker wavefunction. We see that the relative overlap error, $\langle \Psi_{\textrm{T}}|\epsilon\rangle/\langle \Psi_{\textrm{T}}|\psi_\text{true}\rangle$, can easily become large because the true value in the denominator can be small.
%Here, $|\epsilon\rangle$ appears due to the approximate conversion of Slater determinants into MPS.
Now, we consider another error model for the local energy evaluation:
\begin{align}
|\Psi_{\textrm{T}}\rangle &= |\Psi_0\rangle + |\epsilon_H\rangle\\ \label{eq:errormodel}
\langle \Psi_{\textrm{T}} |\hat{H} | \psi_\text{true}\rangle &= E_0\langle\Psi_{\textrm{T}}|\psi_\text{true}\rangle 
+\langle\epsilon_H | (\hat{H}-E_0) | \psi_\text{true}\rangle
\end{align}
where $|\epsilon_H\rangle$ quantifies the distance between $|\Psi_\text{T}\rangle$ and the exact ground state, $|\Psi_0\rangle$.
It is possible that the second term in \cref{eq:errormodel} is much smaller than the first term, especially if $|\Psi_{\textrm{T}}\rangle$ is close to $|\Psi_0\rangle$.
In that case, dividing \cref{eq:errormodel} by $\langle \Psi_{\textrm{T}}|\psi_\text{true}\rangle$ will give us an accurate estimation of the local energy even when $\langle \Psi_{\textrm{T}}|\psi_\text{true}\rangle$ is poorly approximated. 
This analysis suggests that MPS-AFQMC may achieve high accuracy even with a relatively loose adaptive threshold $\epsilon_{\textrm{ad}}$ and a smaller $D_{\textrm{w}}$.
These findings also provide insights for quantum-classical hybrid quantum Monte Carlo methods~\cite{huggins_unbiasing_2022,lee_response_2022,kiser_classical_2023}, where the evaluation of overlap is noisy and correct only up to some additive error and yet the final AFQMC energies were found to be accurate.

\begin{figure}
    \centering
    \includegraphics[width=0.4\textwidth]{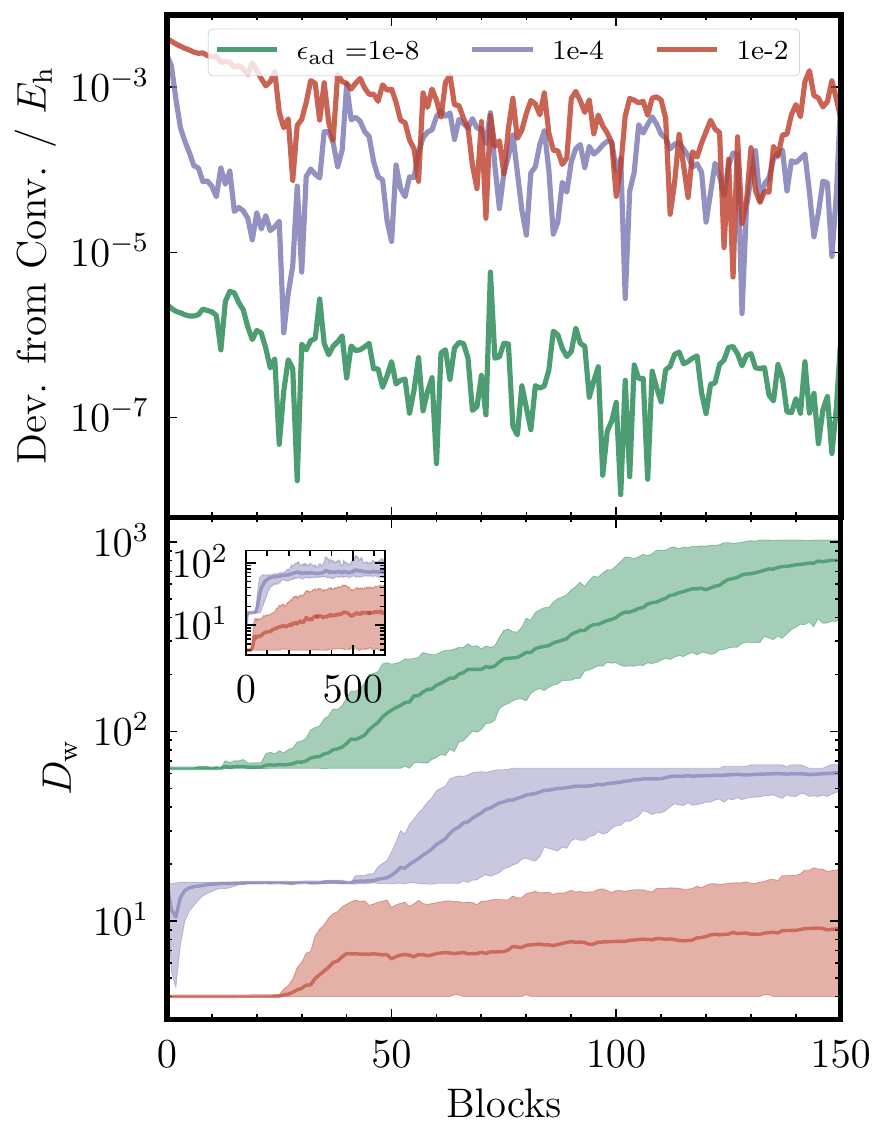}
    \caption{\textbf{MPS-AFQMC energy error and walker bond dimension as a function of imaginary time with various thresholds.} 
    The trial used here is the same as in Fig.~\ref{fig:mps_msd_afqmc-1} with $D_\textrm{T}=4$, with the maximum allowed bond dimension set at 2000. Spin-projection was used in the initial walker set-up. 
}
    \label{fig:evolution1}
\end{figure}

\begin{figure}
    \centering
    \includegraphics[width=0.4\textwidth]{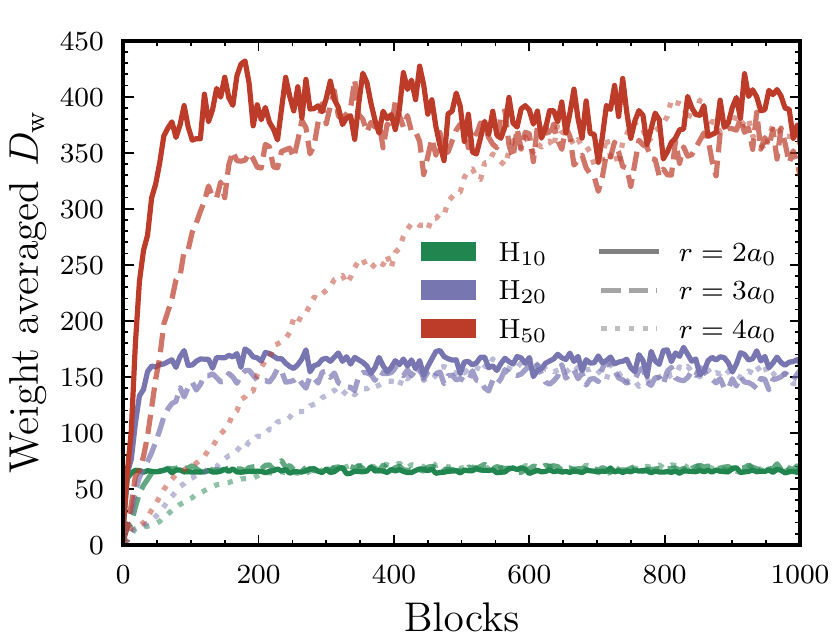}
    \caption{\textbf{Walker bond dimension as a function of imaginary time at various system sizes.} 
    The weight-averaged $D_{\textrm{w}}$ for AFQMC walkers employing a unrestricted Hartree-Fock wavefunction as the trial and initial walker for H$_{10}$, H$_{20}$, and H$_{50}$ systems, evaluated at various interatomic distances $r\in\{2a_0, 3a_0, 4a_0\}$. The adaptive threshold is fixed at $\epsilon_{\textrm{ad}}=10^{-3}$. 
    Each analysis block encompasses 12 steps, with the time step set at $0.01E_{\mathrm{h}}^{-1}$.
}
    \label{fig:evolution2}
\end{figure}
In Fig.~\ref{fig:evolution1}, we fix the trial at $D_{\textrm{T}}=4$, corresponding to the trial used in Fig.~\ref{fig:mps_msd_afqmc-1}, and present the evolution of MPS-AFQMC energy errors from the exact MPS-AFQMC at each step and walker's bond dimension as a function of imaginary time. 
We observe that, without a cap on the walker's bond dimension, a tighter threshold correlates with reduced errors and increased stationary bond dimension. The walker bond dimension initially increases and can grow very fast before reaching a plateau. 
However, as we showed in Fig.~\ref{fig:mps_msd_afqmc-1} and analyzed in Fig.~\ref{fig:mps_msd_afqmc-2}, one appears to need a much smaller walker bond dimension to obtain good MPS-AFQMC energies.

\subsubsection{System size and bond length dependence}
We have previously analyzed the increasing hardness of calculating overlap values for systems as their size increases, as illustrated in Figs.~\ref{fig:H20} and~\ref{fig:sampling}. 
Here, we examine how the walker bond dimension changes during imaginary time propagation as a function of system size. This result is shown in Fig.~\ref{fig:evolution2}.
The trend in these imaginary-time trajectories is similar to H$_{10}$, as we discussed in Fig.~\ref{fig:evolution1}.
Interestingly, while the interatomic distance does not affect the final saturated bond dimension, it does influence the rate of increase of the bond dimension; systems with longer interatomic distances exhibit a slower growth in bond dimension.
Contrary to typical DMRG studies for one-dimensional systems, where bond dimensions needed for a specified accuracy remain fixed across molecular chains of varying lengths due to the area law, in AFQMC calculations, the walker's bond dimension grows with system size. 
In DMRG studies, longer interatomic distances necessitate MPS with smaller bond dimensions, which also contrasts with AFQMC simulations. 

\subsection{Applications of MPS-AFQMC}
In this section, we apply MPS-AFQMC to a variety of systems, including one-dimensional long chains with varying interatomic distances, two-dimensional lattices, and scenarios involving large basis sets that incorporate our MPO-free algorithm.
\subsubsection{\ce{H50} in minimal basis}
\begin{figure*}
    \centering
    \includegraphics[width=0.9\textwidth]{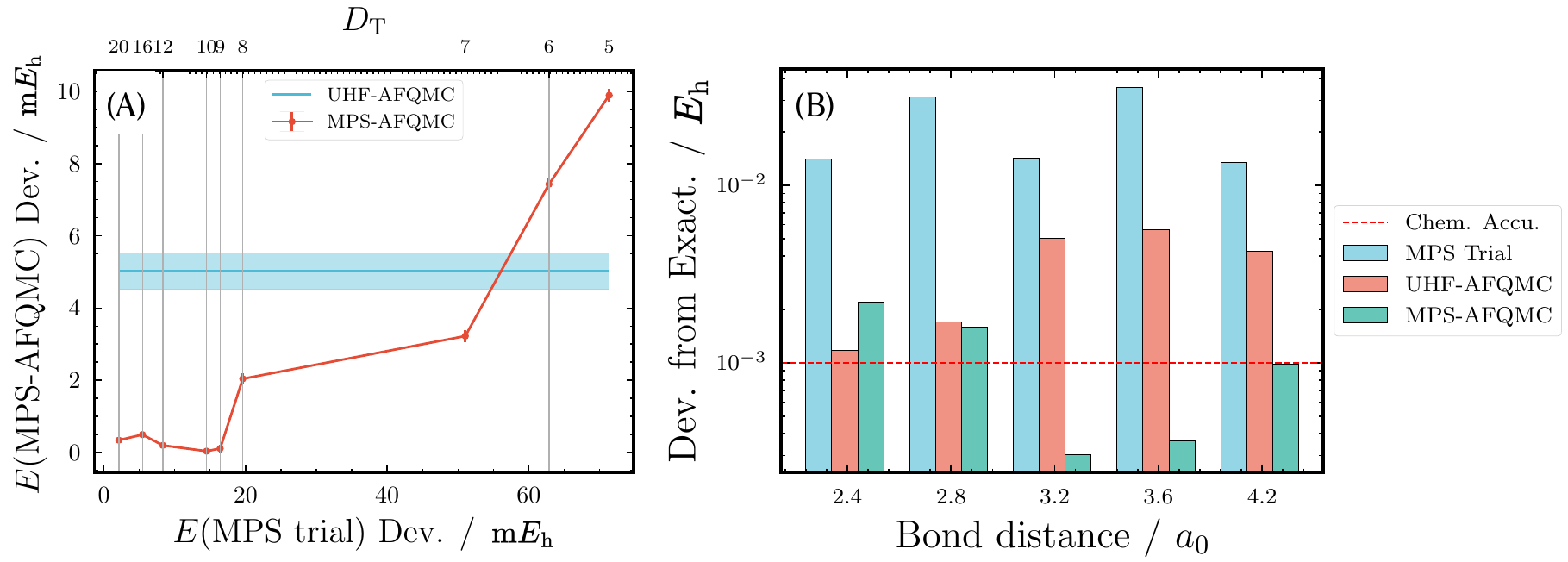}
    \caption{\textbf{The MPS-AFQMC energy deviation for H$_{50}$ at different bond distances and comparison with UHF-AFQMC.} (A) The dependence of the MPS-AFQMC energy deviation on the trial energy deviation. The UHF-AFQMC is also plotted, and the corresponding UHF trial energy deviation is 0.459$E_\textrm{h}$. We present the corresponding overlap/local energy distribution in Fig.~\ref{fig:uhfmps}. 
    (B) The MPS-AFQMC, MPS, and UHF-AFQMC energy for H$_{50}$ across different bond distances. The error bars of MPS-AFQMC calculations are all below $0.1$mEh. The total number of walkers is set to 640, and the time step used is $0.01 E_\textrm{h}^{-1}$ with population control performed every five steps.
    }
    \label{fig:bar}
\end{figure*}
\add{
    We test our approach on H$_{50}$, 
    where the ground state can be computed nearly exactly with DMRG~\cite{hachmann_multireference_2006}.
    Many single-reference methods, such as low-order 
    coupled cluster method (eg., CCSD(T))~\cite{hachmann_multireference_2006} 
    and UHF-AFQMC~\cite{simons_collaboration_on_the_many-electron_problem_towards_2017}, can struggle for some bond distances, making this a good testbed for new methods.
}
Using the GMPS approach with $\epsilon_{\textrm{ad}}=10^{-3}$ and $D_\text{w} = 50$, we apply MPS-AFQMC to calculate the electronic ground state energy of H$_{50}$ at various interatomic distances with STO-6G basis.
In Fig.~\ref{fig:bar}, we calculate the deviation of the ground state energy of H$_{50}$ from the exact reference provided in~\cite{hachmann_multireference_2006}.

In Fig.~\ref{fig:bar}(A), we present the MPS-AFQMC energy and the corresponding MPS variational energy using different MPS trials with varying bond dimensions $D_\textrm{T}$ for H$_{50}$ at $r=3.2a_0$.
MPS-AFQMC significantly improves accuracy compared to the trial's variational energy. Furthermore, as the trial quality improves, the accuracy of MPS-AFQMC also improves, which aligns with our initial results observed on smaller systems, as illustrated in Fig.~\ref{fig:trial}.

Next, we apply MPS-AFQMC to various interatomic distances and compare the results with the variational MPS energy and UHF-AFQMC taken from Ref.~\citenum{simons_collaboration_on_the_many-electron_problem_towards_2017}. These results are shown in Fig.~\ref{fig:bar}(B).
Instead of fixing the bond dimension of the MPS trial for all interatomic distances, we pick trials with a similar energy deviation from the true ground state at each distance. 
% Achieving the same accuracy for a shorter distance becomes harder for DMRG calculations, as it requires a larger bond dimension MPS to capture the delocalization of electrons.
\add{
    At shorter distances (metallic regime), an MPS with localized orbitals requires a larger bond dimension to capture long-range correlations effectively. As a result, the MPS trial does not outperform UHF at small bond lengths.
}
Here, the MPS trials have bond dimension $D_\textrm{T}=20, 10, 10, 8, 5$ for different interatomic distances $r=2.4a_0,2.8a_0,3.2a_0,3.6a_0,4.2a_0$, respectively.
For calculating the singlet ground state, MPS-AFQMC's accuracy could improve by using spin-adapted MPS~\cite{sharma_spin-adapted_2012,*keller2016spin} or spin-projected MPS~\cite{li_spin-projected_2017}, or by employing identical spin-up and spin-down walker wavefunctions. For simplicity, we opted for the latter approach, which is referred to as ``spin-projection'' in the AFQMC community~\cite{purwanto2008eliminating}. Following Eq.~\eqref{eq: propagate}, where both spin matrices use the same propagator, they remain identical throughout the imaginary time evolution.
The errors in MPS-AFQMC energy at relatively long distances (specifically at $r\in\{3.2a_0, 3.6a_0, 4.2a_0\}$) fall below the threshold of chemical accuracy (i.e., 1 m$E_{\mathrm{h}}$).
\add{In comparison, we observed AFQMC with a selected 
   configuration interaction trial (using up to $10^5$ determinants) performs less accurate than UHF-AFQMC, even for a smaller H$_{20}$ system
   at a stretched geometry (results not shown).
   }

We also explored entanglement projection to potentially speed up the conversion from SD to MPS and improve accuracy by controlling the entanglement growth in walker wavefunctions. The efficacy of this technique on AFQMC energy is preliminarily explored in Fig.~\ref{fig:lesp} in Appendix~\ref{sec:lesp}, warranting further investigation for potential enhancements.
MPS-AFQMC is as accurate as UHF-AFQMC at relatively shorter distances and produces much more accurate results for larger bond lengths when the underlying entanglement is smaller in a localized basis.

\add{Lastly, we compared the computational time of MPS-AFQMC and UHF-AFQMC for various system sizes, 
as shown in Fig.~\ref{fig:scaling}. 
Our results indicate that MPS-AFQMC exhibits similar scaling 
behavior to UHF-AFQMC, though with a larger prefactor.
We anticipate future improvements in the computational efficiency 
of MPS-AFQMC, particularly by optimizing 
the SD-to-MPS algorithm or leveraging GPU acceleration of tensor contraction
and GPU parallelization over walkers. 
In the past, AFQMC with SD and MSD trials was significantly accelerated by GPUs~\cite{malone2022ipie, Huang2024Jun,jiang2024improved}.   
}

\subsubsection{Two dimensional lattice: $4\times 4$ hydrogen lattice}\label{sec:4by4}
Expanding our exploration beyond one-dimensional systems, we investigate the effectiveness of our MPS-AFQMC in handling more complex two-dimensional systems. DMRG requires much larger bond dimensions for two-dimensional systems than the one-dimensional case, primarily due to entanglement growth beyond the one-dimensional area law.

Our test system involves a $4\times 4$ square lattice of hydrogen atoms with equivalent nearest-neighbor distances.
The orbitals for both DMRG and MPS-AFQMC calculations are ordered with the genetic ordering algorithm~\cite{olivares2015ab} (as shown in the inset).
The bond dimension of the MPO in the test case is 562, leading to an enlarged half-rotated trial MPS, $\hat{H}|\Psi_\textrm{T}\rangle$, with a bond dimension of $562 \times D_\textrm{T}$ which is impractical for local energy computations. We utilize variational optimization, starting with approximately compressed MPO and MPS as initial guesses. This strategy successfully reduces the bond dimension of the half-rotated trial MPS to 100 or less, enabling significantly more efficient calculations.
For more complex active space models or full space calculations where constructing the MPO is impractical or compressing the half-rotated MPS trial proves challenging, we suggest using the MPO-free approach in Section~\ref{sec:virt}, for which we later present numerical results in Sec.~\ref{sec:h10_dz}.
With these details, the ground state energies were computed with different methods and are shown in Fig.~\ref{fig:h4x4}.

As shown in Fig.~\ref{fig:h4x4}, the performance of DMRG in this two-dimensional system demonstrates a remarkably slow convergence with respect to the bond dimension toward the exact ground state energy.
In comparison, using the MPS trials with $D_\textrm{T}=10$ and $D_\textrm{T}=50$, MPS-AFQMC is highly effective for this system and yields energies within chemical accuracy from the exact answer in both cases. On the other hand, DMRG reaches similar accuracy at substantially higher bond dimensions of about $200$. 
However, employing the UHF trial leads to slightly worse energy, despite the UHF wavefunction having slightly better energy than $D_\textrm{T}=10$ MPS trial.
These calculations used 128 walkers and a time step of $0.1$E$_\textrm{h}^{-1}$. As demonstrated in Fig.~\ref{fig:bias}, Population bias and timestep errors were negligible.
This is a promising result because MPS-AFQMC could be used for systems where MPS alone struggles to achieve high accuracy.
\add{
    We note that the current example serves as a proof-of-concept 
    for MPS-AFQMC.
    Although DMRG is more computationally efficient than MPS-AFQMC for the current system, 
    as DMRG's computational cost increases rapidly with system size, 
    we anticipate MPS-AFQMC's advantage to become more pronounced
    for larger systems after further optimization.
} 

\begin{figure}
    \centering
    \includegraphics[width=0.45\textwidth]{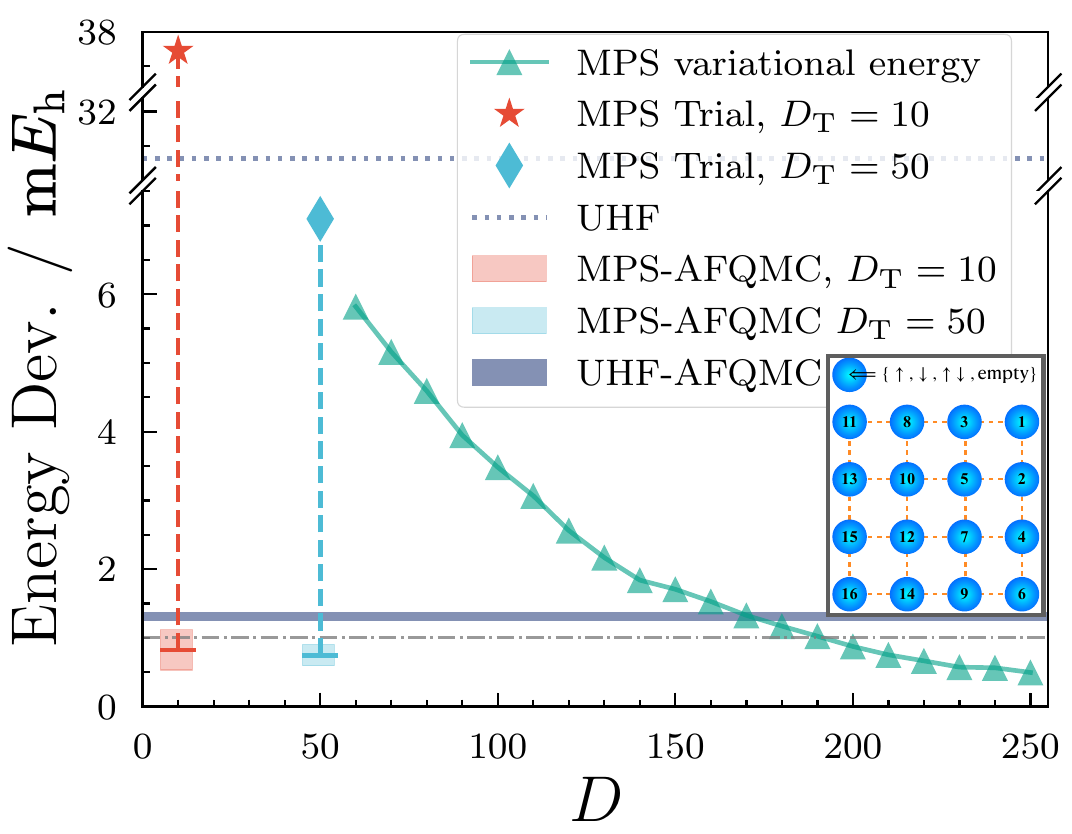}
    \caption{\textbf{Application of MPS-AFQMC to two-dimensional hydrogen lattice with $r=4.2a_0$.} Energy deviation from exact results using MPS-AFQMC, UHF-AFQMC, and DMRG with different bond dimensions.
    \add{The dashed dotted line coresponds to the 1 m$E_{h}$ error.}
    }
    \label{fig:h4x4}
\end{figure}
\subsubsection{Towards larger basis set simulations: \ce{H10} with cc-pVDZ}\label{sec:h10_dz}
\begin{figure}
    \centering
    \includegraphics[width=0.4\textwidth]{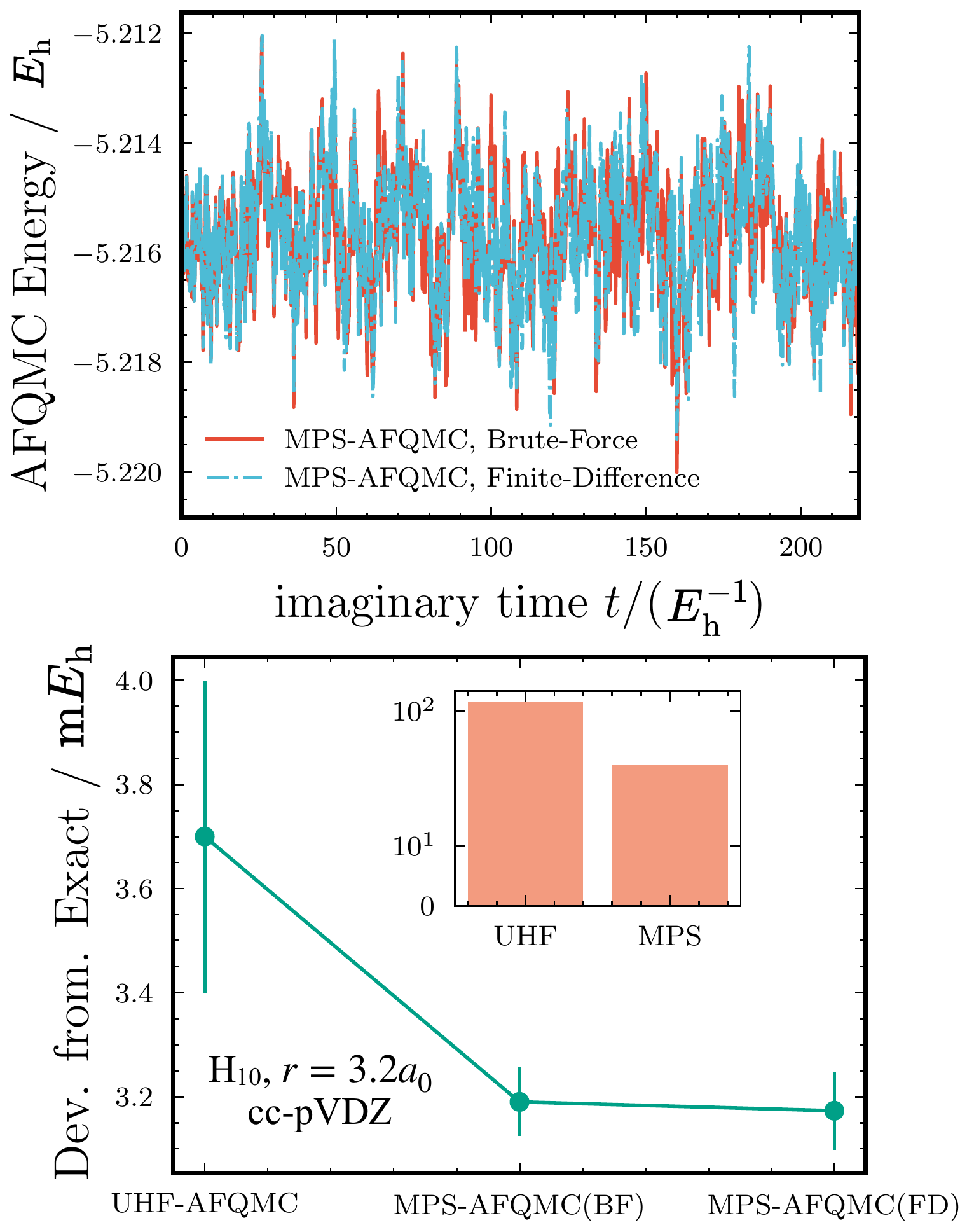}
    \caption{
    \textbf{Application of MPS-AFQMC beyond minimal basis.}
    The interatomic distance for H$_{10}$ is set at $3.2a_0$. 
    (A) The change in AFQMC energy over imaginary time evolution, comparing the results from brute-force MPS-AFQMC to those obtained using finite difference MPS-AFQMC. 
    (B) The energy deviation from the exact reference. The finite difference step is set to $10^{-5}$ for the one-body energy term and $10^{-4}$ for the two-body energy term.
}
    \label{fig:virtual}
\end{figure}
\add{
    We have demonstrated the effectiveness of MPS-AFQMC for systems 
    dominated by static correlation, such as stretched hydrogen chains. 
    Beyond these simple benchmarks, we argue that the most compelling feature of our framework is its ability to extend the capabilities of both DMRG and AFQMC. 
    %MPS-AFQMC is not intended to compete with or replace DMRG, where DMRG is already highly efficient due to the low-entanglement structure. 
    For example, MPS-AFQMC could enable large basis set calculations for systems with strong correlation within active space and non-negligible dynamic correlation outside active space. In these cases, neither method can handle alone. 
    % Additionally, the MPS trial in the active space can be paired with a less accurate but computationally efficient trial wavefunction for the frozen orbitals.
}

% The DMRG algorithm has been mainly optimized for and applied as the FCI solver within a so-called active space. Outside the active space, the remaining dynamic correlation from the larger basis set is missing. However, these dynamic correlations are also important in many quantum systems when one compares the computational results with experiments.
% Our MPS-AFQMC can use MPS within the active space as a trial for AFQMC but seamlessly incorporate dynamic correlation outside the active space. 
\add{As shown in Fig.~\ref{fig:wflow}(a), our approach leverages the strengths of both methods by limiting DMRG to the active space while using AFQMC to capture the remaining dynamic correlation.}
As a proof-of-concept benchmark, we employ MPS-AFQMC to calculate the correlation energy for H$_{10}$ with the cc-pVDZ basis set. We generate our MPS trial within an active space of 10 electrons in 10 orbitals (10e, 10o). 
% The active space was chosen by localizing the unrestricted natural orbitals (UNO)~\cite{bofill_unrestricted_1989} and subsequently employing the Pipek-Mezey (PM) localization~\cite{pipek_fast_1989} as implemented in \textsc{PySCF}. 
% The orbitals are then reordered in real space based on the magnitudes of their coefficients to maximize the efficacy of the variational MPS calculations.
% We utilize the resulting converged MPS wavefunction as the trial in MPS-AFQMC calculations. 
For demonstration, we transform this MPS into a linear combination of Slater determinants, which is possible because this is done in a small active space. 
The corresponding results are illustrated in Fig.~\ref{fig:virtual}, where the MPS-AFQMC employing finite difference refers to the algorithm introduced in Section~\ref{sec:virt}.
We found an excellent agreement between the finite difference algorithm and the brute force approach, demonstrating its correctness and future utility.

\add{ While the improvement of MPS-AFQMC over UHF-AFQMC appears modest, 
this serves as a proof-of-concept demonstration for applying MPS-AFQMC to 
larger systems with MPS restricted to the active space trial. 
For larger systems, UHF-AFQMC is expected to become less reliable, 
often requiring problem-specific trials, such as physics-informed 
Slater determinant expansions~\cite{simons_collaboration_on_the_many-electron_problem_towards_2017}. 
In these cases, our framework enables capturing dynamic correlation beyond 
the active space by using the MPS trial for the active space alongside a 
complementary trial wavefunction (e.g., UHF) for dynamic correlation, 
which will be explored in future work.}

\section{Conclusions and outlooks}\label{sec:conclusion}
In conclusion, we have, for the first time, realized using MPS trial wavefunctions within fermionic AFQMC calculations. In MPS-AFQMC, evaluating the overlap between MPS and a walker state (Slater determinant) is formally and practically the computational bottleneck.
We also proved that the estimation of this overlap up to a multiplicative error is \textit{\#P-hard}. Despite the hardness of this problem, we proposed several heuristics to approximate this overlap evaluation. We found that inaccurate overlap evaluation could still lead to accurate energy estimations due to the zero variance principle in ph-AFQMC.
The strategy we recommend is the one that applies Givens rotation gates to a product state to obtain an MPS representation of walker Slater determinants. This choice was made after examining other approaches, including brute-force conversion, bipartite decomposition, and perfect bitstring sampling. We also discussed balancing accuracy and computational efficiency for practical applications with suitable thresholds and maximum walker bond dimensions.
 
We then showcased accurate MPS-AFQMC results over prototypical strongly correlated systems.
We tested MPS-AFQMC on minimal and double-zeta basis models of one-dimensional hydrogen chains and a minimal basis model of a two-dimensional hydrogen lattice.
For these systems, our method significantly improves over ph-AFQMC with a most commonly used trial wavefunction (i.e., single Slater determinant).
Furthermore, we also proposed an algorithm to evaluate force bias and local energy via differentiation of certain overlap evaluations, which removes the significant overhead introduced by prior works~\cite{huggins_unbiasing_2022,kiser_classical_2023,Huang2024Apr}.
This allows for using an MPS trial obtained within the active space and computing the correlation energy of the entire problem only based on the overlap within the active space. 
Importantly, this low-scaling approach is generalizable to other types of trial wavefunctions.
\add{We also note that while energy is a useful benchmark, correlation functions are crucial for studying 
   phases and phase transitions. Evaluating correlation functions in AFQMC is  
   challenging for operators that do not commute with $\hat{H}$, though recent advancements like 
   algorithmic differentiation offer potential solutions for simple trials~\cite{mahajan2023response,jiang2024improved}. 
}

Our findings open new avenues for fermionic simulations where two state-of-the-art techniques, QMC and tensor network methods, can create synergistic impacts.
\add{
Our method combines the strengths 
of DMRG and AFQMC to achieve, in principle, higher accuracy 
in systems with mixed strong and weak correlations, going beyond the capability of UHF-AFQMC or DMRG alone.  
While our implementation is not yet fully optimized, 
the current bottleneck is the SD-to-MPS conversion 
algorithm. Future work will focus on improving computational efficiency 
through GPU acceleration, parallelization, and optimized SD-to-MPS algorithms
~\cite{malone2022ipie,Huang2024Jun,jiang2024improved}. 
SD-to-MPS conversion is an area of broad interest, and more efficient algorithms may emerge~\cite{Liu2024Aug}. 
}
% Furthermore, integrating GPU acceleration into both the SD-to-MPS subroutine and the AFQMC algorithm~\cite{ipie} can substantially enhance the efficiency of MPS-AFQMC. Such an implementation will extend its applicability to more complex systems.
We also noted that our insight partially explains the success of
quantum-classical hybrid QMC methods~\cite{huggins_unbiasing_2022}. 
We hope that the insight and findings presented in our work will lead to further developments in state-of-the-art fermionic simulation methods.
\section*{Acknowledgments}
T.J. and J.L. were supported by Harvard University’s startup funds, the DOE Office of Fusion Energy Sciences ``Foundations for quantum simulation of warm dense matter'' project, and Wellcome Leap as part of the Quantum for Bio Program. T.J. thanks Jiajun Ren for stimulating discussions and Zhigang Shuai for encouragement and support.
Computations were carried out partly on the FASRC cluster supported by the FAS Division of Science Research Computing Group at Harvard University. This work also used the
Delta system at the National Center for Supercomputing Applications through allocation CHE230032 and CHE230088 from the
Advanced Cyberinfrastructure Coordination Ecosystem:
Services \& Support (ACCESS) program, which is supported by National Science Foundation grants \#2138259,
\#2138286, \#2138307, \#2137603, and \#2138296.
% \nocite{*}

% \clearpage
% \onecolumngrid

\begin{appendix}
\renewcommand{\thefigure}{\thesection\arabic{figure}}
\setcounter{figure}{0} % Reset figure counter 
\section{Additional details on ph-AFQMC}\label{app:afqmc}
This appendix complements the theory of ph-AFQMC in Section~\ref{sec:afqmc}.
With the Cholesky decomposition of the two-electron integrals, the Hamiltonian in Eq.~\eqref{eq:qcham} can be expressed as
\begin{equation}
    \hat{H}=\hat{H}_1+\hat{H}_2 = \hat{v}_0-\frac{1}{2} \sum_{\gamma=1}^{N_\gamma} \hat{v}_\gamma^2
\end{equation}
where 
\begin{equation}
    \hat{v}_0 = \sum_{pq} (h_{pq} -\frac{1}{2}\sum_r (pr|rq))a_{p}^\dagger a_{q}
\end{equation}
\begin{equation}
\hat{v}_\gamma=i \sum_{p q} L_{p q}^\gamma \hat{a}_{p }^{\dagger} \hat{a}_{q}   \label{eq:chol}
\end{equation}
The short-time propagator is approximated using a Trotter decomposition:
\begin{equation}
    e^{-\Delta \tau \hat{H}}=e^{-\frac{\Delta\tau}{2} \hat{v}_0} e^{\frac{\Delta\tau}{2} \sum\hat{v}_{\gamma}^2} e^{- \frac{\Delta\tau}{2} \hat{v}_0}+\mathcal{O}\left(\Delta \tau^2\right)\label{eq:trotter}
\end{equation}
The Hubbard-Stratonovich transformation enables the rewriting of the two-body propagator as an integration over the auxiliary fields $\textbf{x}=(x_1, x_2, \cdots, x_{N_{\gamma}})$ as follows,
\begin{equation}
    e^{\Delta \frac{\tau}{2} \sum\hat{v}_{\gamma}^2} = \int \mathrm{d} \mathbf{x} p(\mathbf{x}) e^{-\sqrt{\Delta \tau} \mathbf{x}\cdot \mathbf{\hat{v}}}
\end{equation}
where $p(\mathbf{x})$ is the standard normal distribution of the auxiliary fields. 
\begin{equation}
    \hat{B}(\mathbf{x}) = e^{-\frac{\Delta\tau}{2} \hat{v}_0} e^{-\sqrt{\Delta \tau} \mathbf{x}\cdot \mathbf{\hat{v}}}e^{-\frac{\Delta\tau}{2} \hat{v}_0}\label{eq:Bprop}
\end{equation}
In practice, a mean-field shift $\left\langle\mathbf{\hat{v}}\right\rangle_\textrm{T}\equiv\left\langle\Psi_\textrm{T}\left|\mathbf{\hat{v}}\right| \Psi_\textrm{T}\right\rangle$ subtraction trick is often used to reduce the statistical fluctuations~\cite{motta_ab_2018, rom_shifted-contour_1997},
\begin{equation}
    \hat{H}_1^{\prime}=\hat{H}_1-\sum_{\gamma=1}^{N_\gamma} \hat{v}_\gamma\left\langle\hat{v}_\gamma\right\rangle_\textrm{T}+\frac{1}{2} \sum_{\gamma=1}^{N_\gamma}\left\langle\hat{v}_\gamma\right\rangle_\textrm{T}^2
\end{equation}
\begin{equation}
    \hat{H}_2^{\prime}=-\frac{1}{2} \sum_{\gamma=1}^{N_\gamma}\left(\hat{v}_\gamma-\left\langle\hat{v}_\gamma\right\rangle_\textrm{T}\right)^2
\end{equation}
which changes the expression of the propagator Eq.~\eqref{eq:Bprop} accordingly.

The force bias in Eq.~\eqref{eq: propagate} is defined as
\begin{equation}
    \overline{\mathbf{x}}_i(\Delta \tau, \tau)=
    -\sqrt{\Delta \tau} \frac{\left\langle\Psi_{\textrm{T}}\left|\hat{\mathbf{v}}-\langle\hat{\mathbf{v}}\rangle_\textrm{T}\right| \psi_i(\tau)\right\rangle}{\left\langle\Psi_{\textrm{T}}|\psi_i(\tau)\right\rangle}.\label{eq:fb}  
\end{equation}

In traditional AFQMC, the periodic orthogonalization of walkers is employed to eliminate round-off numerical errors, thus maintaining numerical stability, in our MPS-AFQMC algorithm, the orthogonalization is done at every step as required by the conversion of SD to MPS (SD-to-MPS). Population control is periodically utilized to remove walkers with small weights and replicate walkers with large weights.
\section{Additional methods and results for SD-to-MPS}\label{app:sd2mps_others}
\begin{figure}
    \centering
    \includegraphics[width=0.45\textwidth]{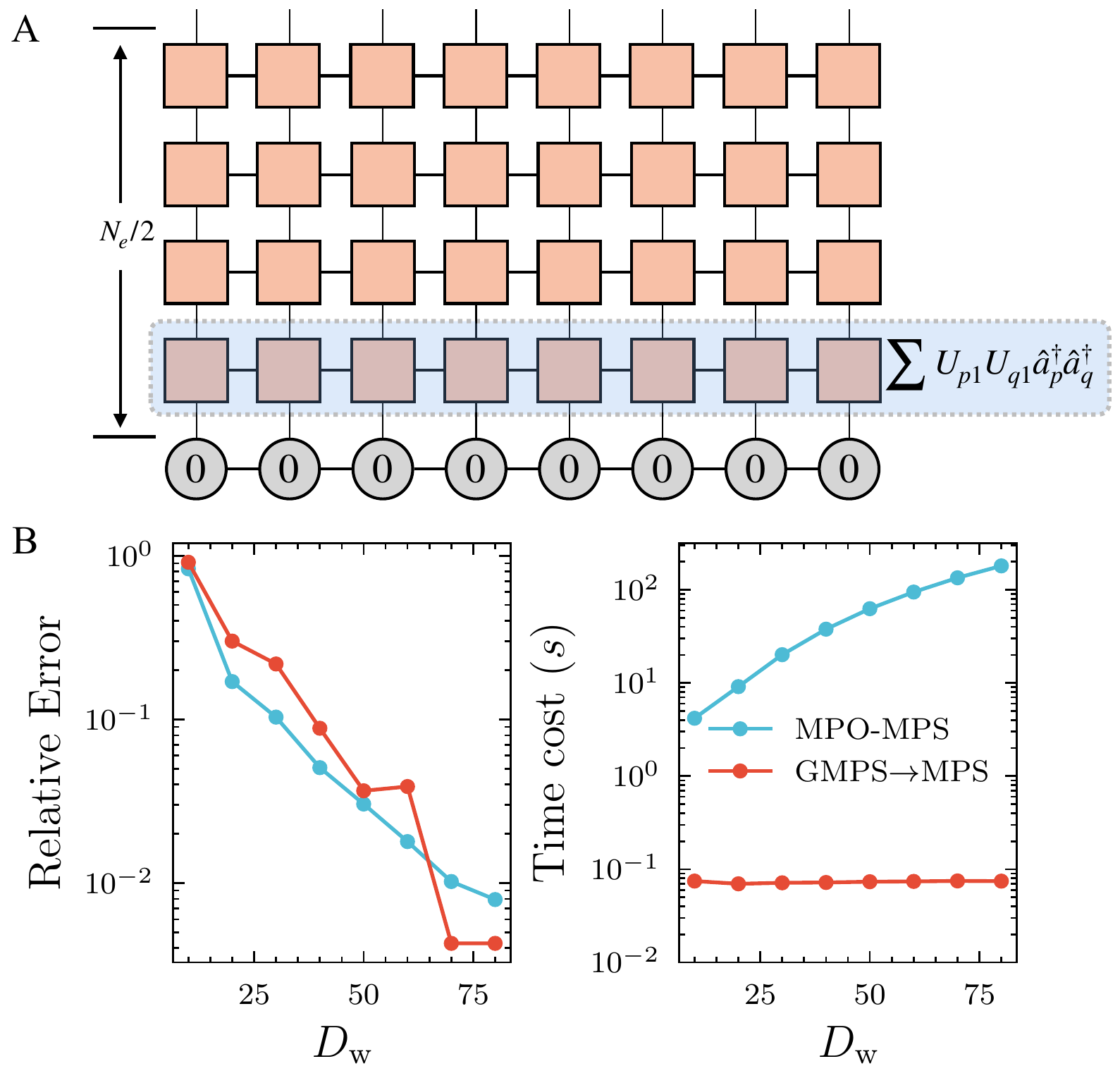}
    \caption{(a) Conversion of SD-to-MPS by using brute force long ranged MPO-MPS application (Eq.~\eqref{eq:mpo-mps}). (b) Comparing the brute force MPO-MPS approach and the Gaussian MPS to MPS approach in terms of relative error of overlap and the time cost as a function of walker's bond dimension.}
    \label{fig:mpomps}
\end{figure}

\paragraph{Brute-force MPO-MPS application.} The brute-force MPO-MPS application is a straightforward, though not computationally efficient method~\cite{wu_tensor_2020}. The conversion of SD-to-MPS is accomplished by sequentially applying a series of creation operators to the initial vacuum state. These creation operators inherently involve long-range MPO, which contributes to a rapid increase in the entanglement levels of the MPS and encounters challenges in efficiently compressing the resultant MPS. The high entanglement levels make it difficult to represent the state with a low-dimensional MPS without losing significant information. 
In Fig.~\ref{fig:mpomps}, we show its demanding computational time by comparing with the GMPS-to-MPS approach. 

\paragraph{Bipartite decomposition of SD}\label{para:bipartite}
Another approach we considered is the formation of MPS from a given SD using a bipartite decomposition~\cite{peschel_special_2012,klich_lower_2006,petrica_finite_2021}.
\add{A naive decomposition can be written as 
\begin{equation}
    |\psi\rangle = \prod_{k=1}^N (\sum_{p=1}^l U_{pi} a_p^\dagger + \sum_{p=l+1}^{N} U_{pi} a_p^\dagger)|0\rangle
\end{equation}
which makes the bipartite decomposition of $|\psi\rangle = \sum_\nu \lambda_\nu |L_\nu^l\rangle |R_\nu^l\rangle$ a  
maximally entangled state. 
}
Another bipartite decomposition of SD can be written as~\cite{peschel_special_2012,klich_lower_2006,petrica_finite_2021}
\begin{equation}
    |\psi\rangle = \Pi_{k=1}^N(\sqrt{\varepsilon_k}|l_k\rangle+\sqrt{1-\varepsilon_k}|r_k\rangle),\label{eq:factorize}
\end{equation}
where $\{|l_k\rangle\}$ and $\{|r_k\rangle\}$ represent the projected basis which spans the subset of spin orbitals $\mathcal{L}=\{\phi_1,\phi_2,\phi_3\cdots,\phi_l\}$ and $\mathcal{R}=\{\phi_{l+1},\phi_{l+2},\cdots,\phi_N\}$, respectively.
\add{This is similar to density matrix embedding theory~\cite{PhysRevLett.109.186404,wouters2016practical}.} 
$|l_k\rangle$ and $|r_k\rangle$ can be obtained by performing SVD on the subblock of the coefficient matrix,
\begin{equation}
    C^{\mathcal{L}}_{pq} \equiv U_{pq} = \sum_s Q_{ps}S_s V^\dagger_{sq}, \; p\le l \label{eq:svd_bipar}
\end{equation}
where $S$ contains $l$ non-zero singular values and $N-l$ zero singular values and these singular values correpsond to the $\sqrt{\varepsilon_k}$ in Eq.~\eqref{eq:factorize}.
A phase emerges for Eq.~\eqref{eq:svd_bipar}, $e^{i\theta}V^\dagger$, which should be removed by dividing the first row of $V^\dagger$ with the scaling factor det($V^\dagger$). 
Using this, $|l_k\rangle$ and $|r_k\rangle$ read
\begin{equation}
    \begin{aligned}
    & \add{a_{l_k}^\dagger} = \sum_{p=1}^l \sum_q U_{pq}V_{qk} a_p^\dagger, \\\nonumber
    & \add{a_{r_k}^\dagger} = \sum_{p=l+1}^N\sum_q U_{pq}V_{qk} a_p^\dagger    
\end{aligned}\label{eq:lr}
\end{equation}
Using the Schmidt decomposition at the $l$th site,
we have
\begin{equation}\label{eq:sch_decomp}
    |\psi\rangle = \sum_{\nu}\lambda_\nu|L^l_\nu\rangle|R^l_\nu\rangle
\end{equation}
where the $\lambda_\nu$, $|L^l_\nu\rangle$ and $|R^l_\nu\rangle$ can be computed from Eq.~\eqref{eq:factorize}~\eqref{eq:lr}.
\begin{equation}
    \lambda_{\nu} = \sqrt{\epsilon_{L_1}\epsilon_{L_2}\cdots\epsilon_{L_\mu}(1-\epsilon_{R_1})\cdots(1-\epsilon_{R_{N-\mu}})}\label{eq:bipar_coeff}
\end{equation}
\begin{equation}\label{eq:Lnu}
    |L_\nu^l\rangle = |l_{L_1}l_{L_2}\cdots l_{L_\mu}\rangle
\end{equation}
\begin{equation}\label{eq:Rnu}
    |R_\nu^l\rangle = |r_{R_1}r_{R_2}\cdots r_{R_{N-\mu}}\rangle
\end{equation}
A complete linear expansion of $|\psi\rangle$ in spin orbital basis is obtained with Eq.~\eqref{eq:sch_decomp} to \eqref{eq:Rnu}, with exponentially many terms.

By introducing a threshold \(\epsilon_{\text{trunc}}\), relatively smaller components can be effectively screened out. 
If the coefficients can be sorted by amplitudes and truncated based on \(\epsilon_{\text{trunc}}\), the SD is \textit{optimally} compressed.
To manage the $2^N$ components, we sort the components by an \(N\)-length binary bitstring for \(|L_\nu\rangle|R_\nu\rangle\), where a 0 in the bitstring signifies selecting \(|l_k\rangle\) and a 1 signifies \(|r_k\rangle\). The initial bitstring, chosen for having the largest coefficients, selects \(|l_k\rangle\) if \(\sqrt{\epsilon_k} > \sqrt{1-\epsilon_k}\) and \(|r_k\rangle\) otherwise.
Starting with the bitstring representing the largest coefficients, we first flip a single bit. Then, in a structured manner, we systematically explore further combinations by flipping pairs of bits ($\mathcal{C}_N^2$), followed by triples ($\mathcal{C}_N^3$), and so forth, increasing the number of flipped bits. 
A bitstring is retained only if its corresponding coefficients exceed the threshold $\epsilon_{\textrm{trunc}}$. 
The process halts at $K$ flips when no additional bitstrings are generated above the threshold after exploring all combinations within \(K\) flips.
To compute the MPS site tensor, we iterate through all bitstrings above the threshold and use the following:
\begin{equation}\label{eq:clark_A}
    \begin{aligned}
        & A_{a_{1}}^{n_1} = \lambda_{a_1}\langle n_1|L_{a_1}^1\rangle, A_{a_{N-1}}^{n_N}=\langle n_N|R_{a_{N-1}}^{N-1}\rangle \\
        & A_{a_{i},a_{i+1}}^{n_{i+1}}=\langle n_{i+1} R_{a_{i+1}}^{i+1}|R_{a_i}^{i}\rangle,\; 1\le i< N-1,
    \end{aligned}
\end{equation}
where $A$ is our MPS site tensor.
It is also noteworthy that this method can be parallelized across each site, where the bipartite decomposition (Eq.~\eqref{eq:sch_decomp}) is independent for each site, and the computation of local site elements (Eq.~\eqref{eq:clark_A}) relies only on the decomposition information of two neighboring sites.
\paragraph{Perfect MPS sampler of bit strings.}\label{sec:strategy3}  
Apart from performing SD-to-MPS with Strategies 1 and 2, the overlap can be computed by sampling bitstrings from the walker state~\cite{ferris_perfect_2012, sandvik_variational_2007, schuch_simulation_2008}. 
This is called ``perfect'' because unlike Metropolis sampling it produces uncorrelated samples~\cite{ferris_perfect_2012}.
\begin{equation}
    \langle\Psi_{\textrm{T}}|\phi\rangle = \sum_s \langle\phi|s\rangle\langle s|\phi\rangle\frac{\langle \Psi_\textrm{T}|s\rangle}{\langle \phi|s\rangle} \simeq \frac{1}{N_{\textrm{samp}}}\sum_s \frac{\langle \Psi_\textrm{T}|s\rangle}{\langle \phi|s\rangle}\label{eq:sample_from_walker}
\end{equation}
with the sampling probability coming from walker, $P(|s\rangle) = |\langle\phi|s\rangle|^2$~\cite{PhysRevLett.128.220503,wan_matchgate_2023}.
The state $|s\rangle=|s_1 s_2\cdots s_N\rangle$ is a product state that is represented as a bitstring with binary values ($s_i\in\{0, 1\}$) indicating the occupation or absence of spin orbitals.
As $|\phi\rangle$ is a matchgate, one can efficiently sample bit strings from this state without relying on Metropolis sampling.
$|s\rangle$ is a computational state, hence an MPS state of bond dimension 1.

The lower part in Fig.~\ref{fig:circuit}(B) shows using perfect sampler from the MPS to obtain the probability of occupation for the $i-$th bit of the occupancy bitstring, given the partial occupation bitstring $\tilde{s}_1\tilde{s}_2\cdots \tilde{s}_{i-1}$, the probability of occupation ($s_i=1$) and unoccupation ($s_i=0$) is obtained by
\begin{equation}
    P(s_i|\tilde{s}_1\tilde{s}_2\cdots \tilde{s}_{i-1}) = \left(\sum_{a_i} (\tilde{A}_{a_i}^{s_i})^*\tilde{A}_{a_i}^{s_i}\right) / P(\tilde{s}_1\tilde{s}_2\cdots \tilde{s}_{i-1})
\end{equation}
\begin{align*}
\label{eq:perfect_sampler}
    \tilde{A}_{a_i}^{s_i}=\sum_{a_j,j<i}A_{a_1}^{\tilde{s}_1}A_{a_1 a_2}^{\tilde{s}_2}\cdots A_{a_{i-2}a_{i-1}}^{\tilde{s}_{i-1}}A_{a_{i-1}a_i}^{s_i}
\end{align*}
as digramically expressed as follows,

\begin{equation}
P(s_i|\tilde{s}_1\tilde{s}_2\cdots \tilde{s}_{i-1})\Leftarrow
\begin{diagram_05}
\MpsWalkerCirclePattern{0}{1}{}{ur};
\MpsWalkerCirclePattern{0}{-1}{}{dr};
\draw (1.0,1) edge[out=0,in=0] (1.0, -1);
\draw (0, 2.5) node {$s_i$};
\draw (0, -2.5) node {$s_i$};
\end{diagram_05}
\end{equation}

\begin{equation}
\begin{diagram_05}
\MpsWalkerCircleOrange{0}{0}{$ $}{ur};
\MpsWalkerCircleOrange{1.5}{0}{$ $}{ulr};
\MpsWalkerCircleOrange{4.5}{0}{$ $}{ulr};
\MpsWalkerCircle{6}{0}{$ $}{ulr};
\draw (3, 0) node {$\cdots$};

\draw (0, 1.5) node {$\tilde{s}_1$};
\draw (1.5, 1.5) node {$\tilde{s}_2$};
\draw (4.5, 1.5) node {$\tilde{s}_{i-1}$};
\draw (6, 1.5) node {$s_i$};
\draw (-2.5, 1.5) node {$s_i$};
\draw (-1., 0) node {$=$};
\MpsWalkerCirclePattern{-2.5}{0}{$ $}{ur};
\end{diagram_05}\label{eq:pr1r2}
\end{equation}

% \begin{widetext}
% \begin{equation}
% \sum_{\substack{\{a_j\}, \\ \{a_j'\}, \\ j<i}} 
% \end{equation}
% \end{widetext}
A random number is proposed to determine the value of $s_i\in[0, 1]$.
The following algorithm describes the process of using the perfect sampler to compute the overlap between MPS and SD, which uses the sampling probability from the trial. The algorithm that uses the probability in Eq.~\eqref{eq:sample_from_walker} can be adjusted accordingly.
\begin{algorithm}[H]\label{mc_algo}
\SetAlgoLined
\caption{Perfect sampling of the overlap}
Set initial overlap $f=0$\;
\For{$sample = 1$ \KwTo $N_{\textrm{samp}}$}{
\For{$i=1$ \KwTo $N$}{Compute the probability matrix of site $i$, get $P_0$ and $P_1$\;
Draw $u \sim \text{Uniform}(0, 1)$\;
    \eIf{$u \leq P_0$}{
        Assign $s_i=0$;
    }{
        Assign $s_i=1$;
    }
    Next site sampling\;
    }
    Obtain $|s\rangle=|s_1s_2\cdots s_N\rangle$ and $f_0=\langle \Psi_{\textrm{T}}|s\rangle$\;
    Compute $f_1=\langle s|\psi\rangle$\;
    $f = f + f_1/f_0$\;
}
Get the sampled overlap, $f/N_{\textrm{samp}}$.
\end{algorithm}

\begin{figure}
\centering       
\includegraphics[width=0.5\textwidth]{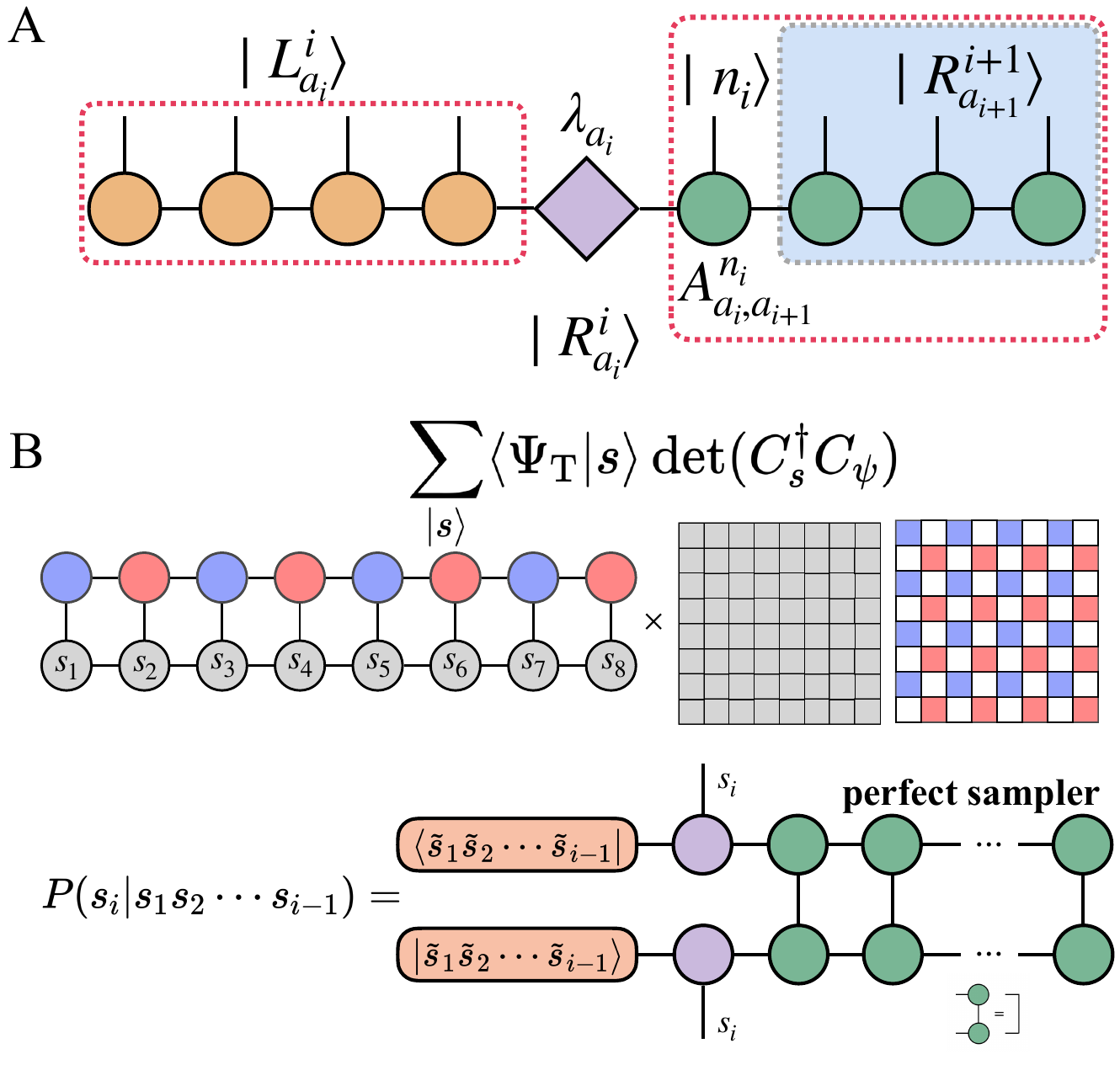}     
\caption{\textbf{Schematics of other overlap computation strategies for MPS-AFQMC}. 
(A) Strategy 2: computation of each local matrix of MPS by the bipartite 
decomposition of Slater determinants; 
(C) Strategy 3: sampling of the overlap between the MPS trial and SD with the perfect sampler.
}
\label{fig:circuit}
\end{figure}  
 
\paragraph{The GMPS-to-MPS approach.} We add some details on the GMPS-to-MPS approach we mentioned in Sec.~\ref{sec:sd2mps_White}.
The Abelian symmetry of MPS~\cite{zgid_spin_2008,sharma_spin-adapted_2012} is conserved by employing quantum number conserving SVD to preserve the numerical stability in Eq.~\eqref{eq:svd}. When converting the SD to an MPS, we observed that a random phase factor, denoted as $e^{i\theta}$, emerges due to the U(1) gauge freedom. In some applications, this random phase factor is inconsequential, as it does not affect the resulting state remaining an eigenstate of the one-body operator~\cite{fishman_compression_2015,petrica_finite_2021}. 
However, in our current context, it affects the phaseless constraint, which relies on the phase change in the overlap. To eliminate this gauge freedom and fix the gauge, we multiply the MPS by the following scaling factor $e^{i\theta}$, where
\begin{equation}
\theta = \arg{(\langle \psi_\textrm{SD}|\psi_{\textrm{max}}\rangle}) - \arg{(\langle \psi_\textrm{MPS}|\psi_{\textrm{max}}\rangle)}\label{eq:remove_phase}
\end{equation}
% \begin{equation}
%     \frac{\langle \psi_\textrm{MPS}|\psi_{\textrm{max}}\rangle}{\langle \psi_\textrm{SD}|\psi_{\textrm{max}}\rangle}
% \end{equation}
where $|\psi_{\textrm{max}}\rangle$ is the product state with the largest overlap with the MPS, which can be 
deterministically obtained easily by the perfect sampler algorithm. 

Another problem is that the compression applied during the contraction of the gates makes the resulting MPS unnormalized.
Therefore, we apply the normalization after all the gates are contracted, which mitigates the error. We analyze the impact of normalization on the error mitigation of GMPS-to-MPS.
We assume the GMPS-to-MPS approach produces an error of $|\epsilon\rangle$.
Without doing normalization of the resulting MPS, the relative error of the overlap value is equal to 
\begin{equation}
  \left|\frac{\langle\Psi_\textrm{T}|\epsilon\rangle}{\langle\Psi_\textrm{T}|\psi_\text{true}\rangle}\right|
  \label{eq:unnormalizederror}
\end{equation}
On the other hand, by doing normalization of the resulting MPS, the overlap value is expressed as
\begin{equation}
\frac{\langle\Psi_\textrm{T}|(|\psi_\text{true}\rangle+|\epsilon\rangle)}{\sqrt{1
+\langle\psi_\text{true}|\epsilon\rangle
+\langle\epsilon|\psi_\text{true}\rangle+\langle\epsilon|\epsilon\rangle}}\label{eq:b1}    
\end{equation}
and by performing first-order Taylor expansion of the denominator,
the overlap is approximated to be
\begin{align}\nonumber
     &\langle\Psi_\textrm{T}|\psi_\text{true}\rangle + \langle\Psi_\textrm{T}|\epsilon\rangle - \frac{1}{2}\left(\langle\psi_\text{true}|\epsilon\rangle +
     \langle\epsilon|\psi_\text{true}\rangle + \langle\epsilon|\epsilon\rangle\right)\\
     &\times(\langle\Psi_\textrm{T}|\psi_\text{true}\rangle + \langle\Psi_\textrm{T}|\epsilon\rangle)
\end{align}
Assuming a small error vector, which implies negligible second-order contributions, the relative error can be approximated (up to the first order) as
\begin{equation}
    \left|
    \frac{\langle\Psi_\textrm{T}|\epsilon\rangle}{\langle\Psi_\textrm{T}|\psi_\text{true}\rangle}
    - \frac{1}{2}\left(\langle\psi_\text{true}|\epsilon\rangle +
     \langle\epsilon|\psi_\text{true}\rangle\right)
    \right|
    \label{eq:normalizederror}
\end{equation}
While \cref{eq:normalizederror} appears to have a larger error than \cref{eq:unnormalizederror}, 
the extra terms in \cref{eq:normalizederror} 
are expected to be orders of magnitude smaller since they do not have a prefactor that is inverse of a small quantity.
Given that there is only a small additional error suggested by our analysis, we chose to normalize the resulting MPS representation of walkers.
Based on several numerical tests, we found that this choice does not affect the final MPS-AFQMC energies in a statistically significant way.
 
The conversion to MPS can be performed with various orderings of the site basis. However, the final ordering must match the trial MPS ordering to evaluate their overlap.
Fig.~\eqref{fig:wflow}(D) and Fig.~\ref{fig:abab} illustrate two distinct schemes for SD-to-MPS, both finally maintaining the same ordering. 
The second conversion is achieved using \textit{SWAP} gates and associated parity gates to recover the alternating ordering~\cite{li_fly_2022}. In the specific example depicted in Fig.~\ref{fig:abab}(B), the parity gate is denoted as:
\begin{equation}
   \hat{\mathcal{P}} = (-1)^{\sum_{k=1}^{N_{\downarrow}}\hat{\sigma}_{k}^{z\downarrow}\Pi_{l\le k}^{N_{\uparrow}}\hat{\sigma}_{l}^{z\uparrow}}(-1)^{\sum_{i=1}^{N_{\uparrow}}\hat{\sigma}_{i}^{z\uparrow}\Pi_{j<i}^{N_{\downarrow}}\hat{\sigma}_{j}^{z\downarrow}}
\end{equation}
% NEED to be double checked
where $\sigma_{i}^{z\uparrow(\downarrow)}$ is the Pauli Z operator for the $i$th spin-up(down) orbitals.
As shown in Figure~\ref{fig:abab}(B), the SD-to-MPS conversion using \textit{SWAP} gates is both less accurate and more time-consuming compared to the alternating ordering method used in the main text of the manuscript. 
\begin{figure}
    \centering
    \includegraphics[width=0.45\textwidth]{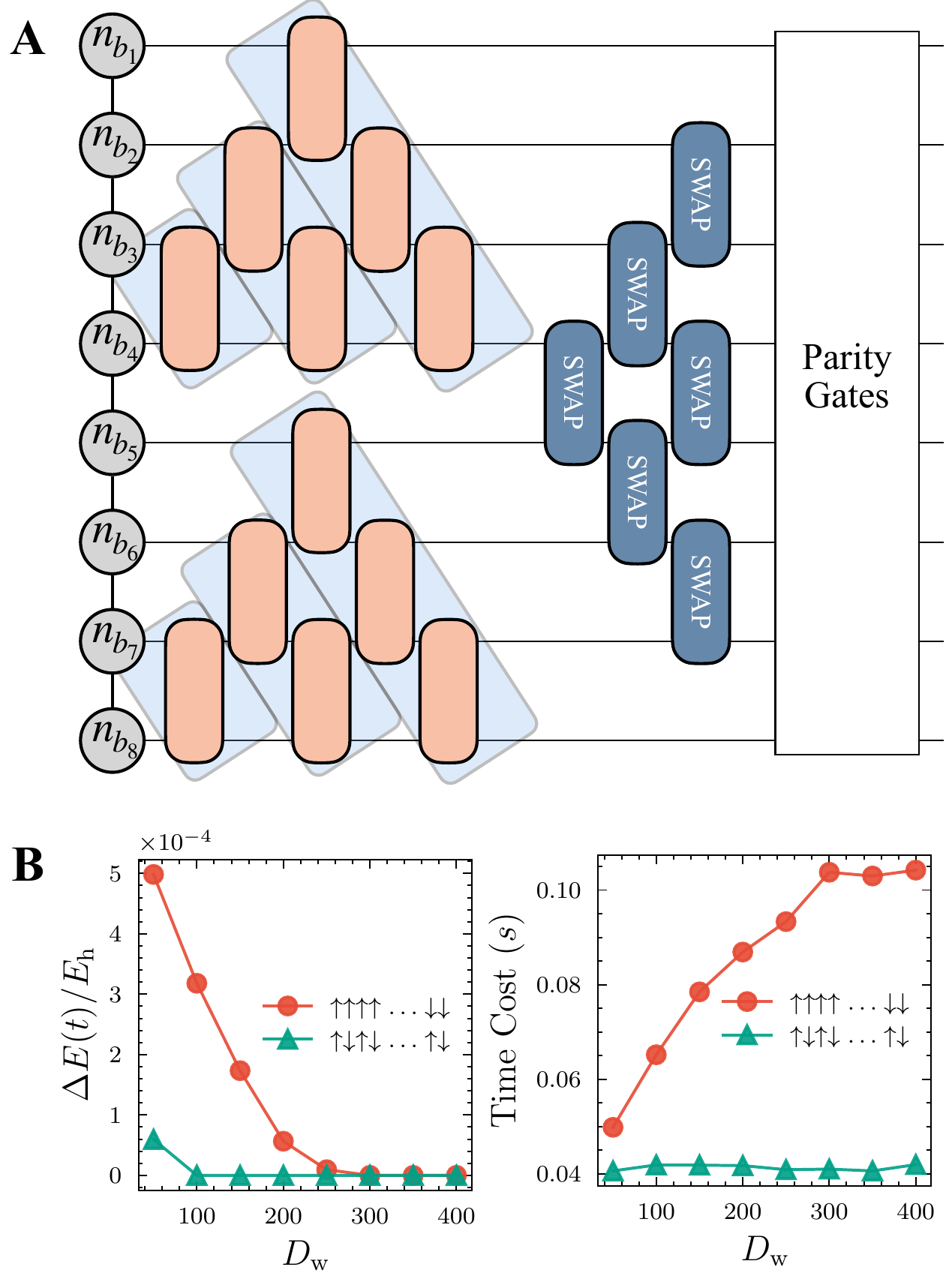}
    \caption{(A) Performing the GMPS-to-MPS for spin up and spin down block separately then reordering; (B) Comparing two ordering schemes with H$_{10}$ with interatomic distance $r=5.0a_0$; The left panel is the mean absolute errors compared to for the energy of selected 5 equilibrium AFQMC blocks and the right panel shows the time cost.}
    \label{fig:abab}
\end{figure}

\begin{figure} 
    \centering
    \includegraphics[width=0.4\textwidth]{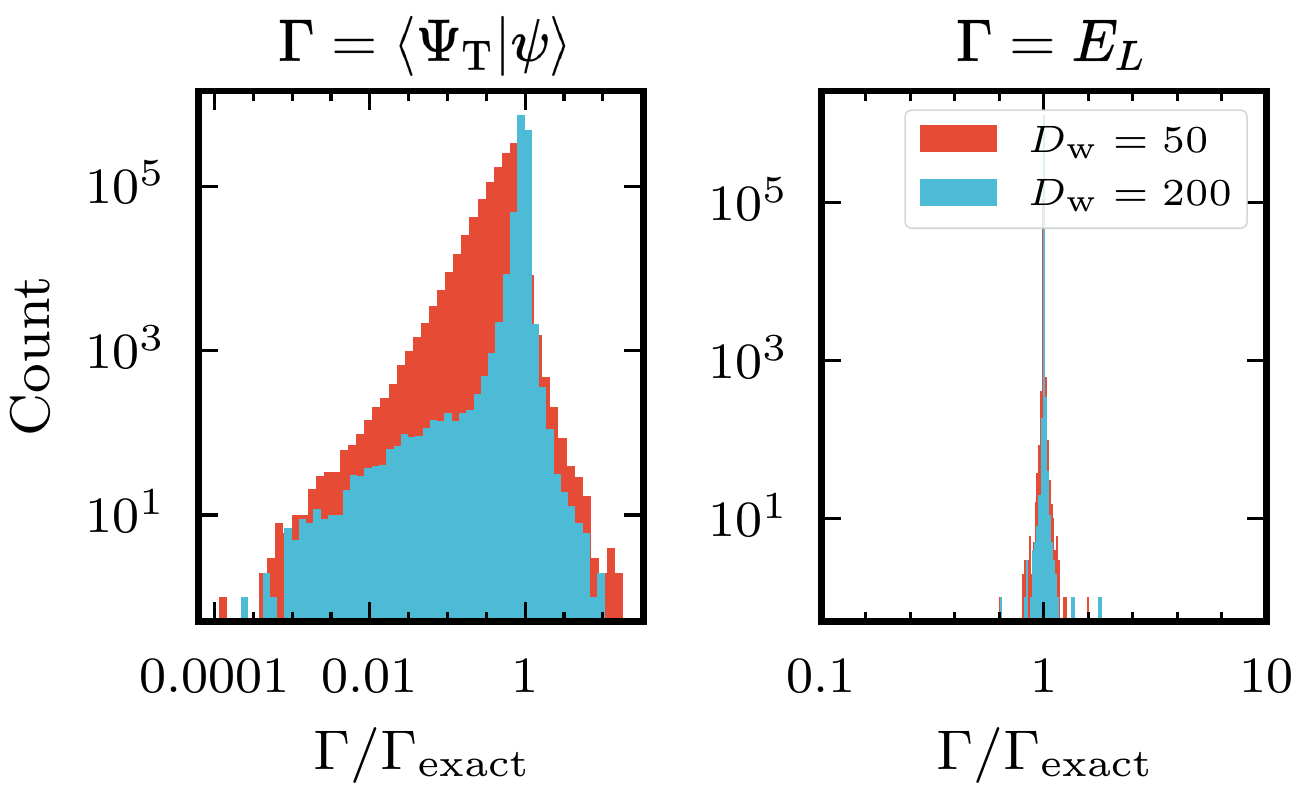}
    \caption{Distributions of the overlap ratio and the local energy ratio for H$_{50}$ with the MPS form of UHF trial.
    The overlap ratio is defined as the ratio of the overlap calculated with a specific $D_{\textrm{w}}$ and $\epsilon_{\textrm{ad}}$ to the exact overlap value, which can be obtained by AFQMC with UHF trial in the Slater determinant form. 
    Similarly, the local energy ratio compares the local energy with its converged value. 
    These distributions are aggregated across the SD of 1,280,000 equilibrated walkers.}
    \label{fig:uhfmps}
\end{figure}

\section{The AFQMC convergence test with time step and number of walkers}
In Fig.~\ref{fig:bias} we test the number of walkers and time step used in AFQMC calculation for the ${4\times 4}$ hydrogen lattice system.
When using different setups of the number of walkers and the time step, the AFQMC energies remain the same up to statistical error bars.

\begin{figure}
    \centering
    \includegraphics[width=0.4\textwidth]{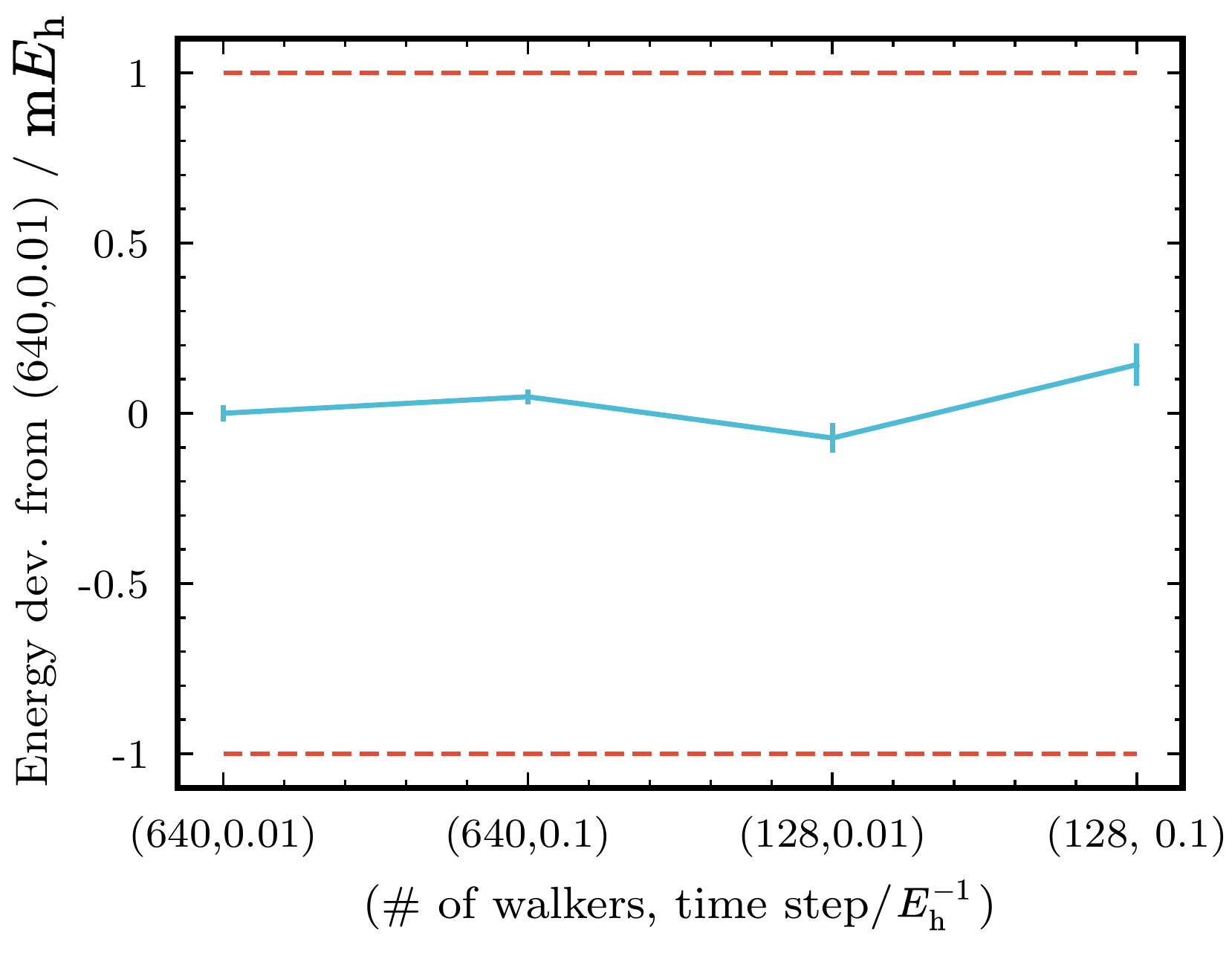}
    \caption{\textbf{The UHF-AFQMC Energy for the ${4\times 4}$ hydrogen lattice.} The energy differences associated with varying the number of walkers and the time step, compared to using 640 walkers and $dt = 0.01$ E$_\textrm{h}^{-1}$, are plotted.}
    \label{fig:bias}
\end{figure}

\section{Low bandwidth projection of correlation matrix}\label{sec:lesp}
\begin{figure}
    \centering
    \includegraphics[width=0.5\textwidth]{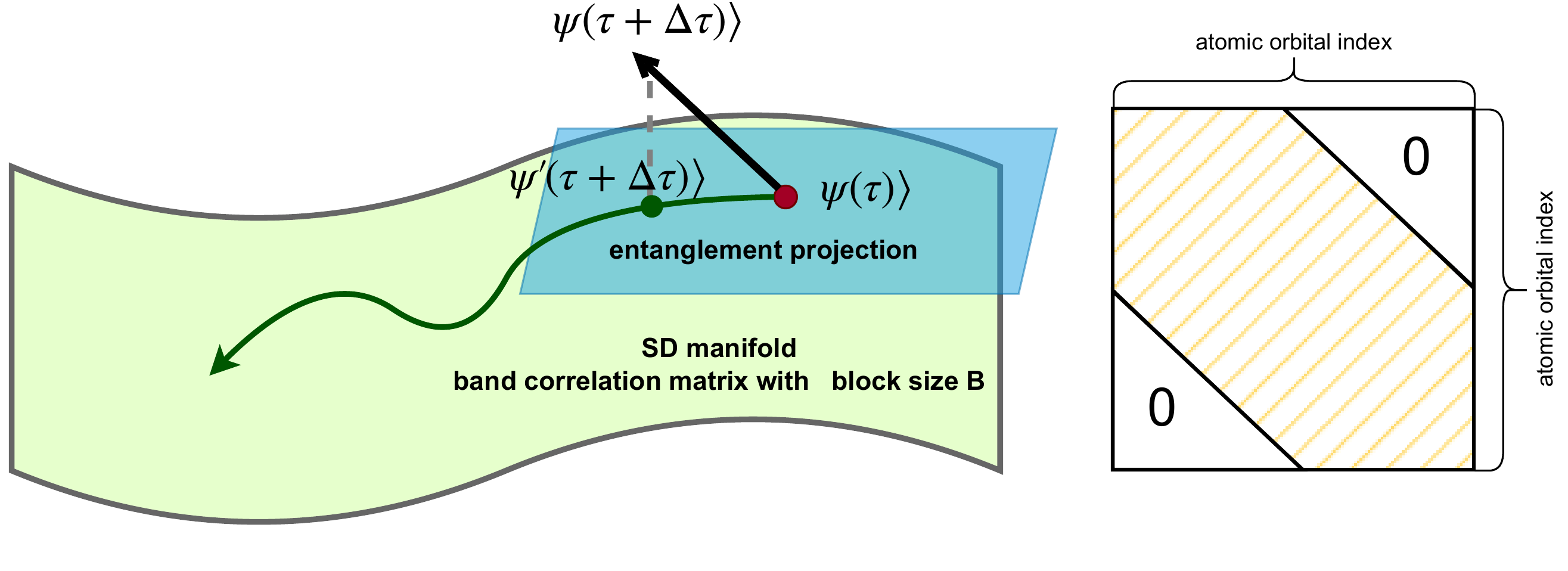}
    \caption{\textbf{Low entanglement projection.} Projecting the SD after imaginary time evolution to a subspace with fixed block size for the correlation matrix.}
    \label{fig:proj_scheme}
\end{figure}
\begin{figure*}
    \centering
    \includegraphics[width=0.8\textwidth]{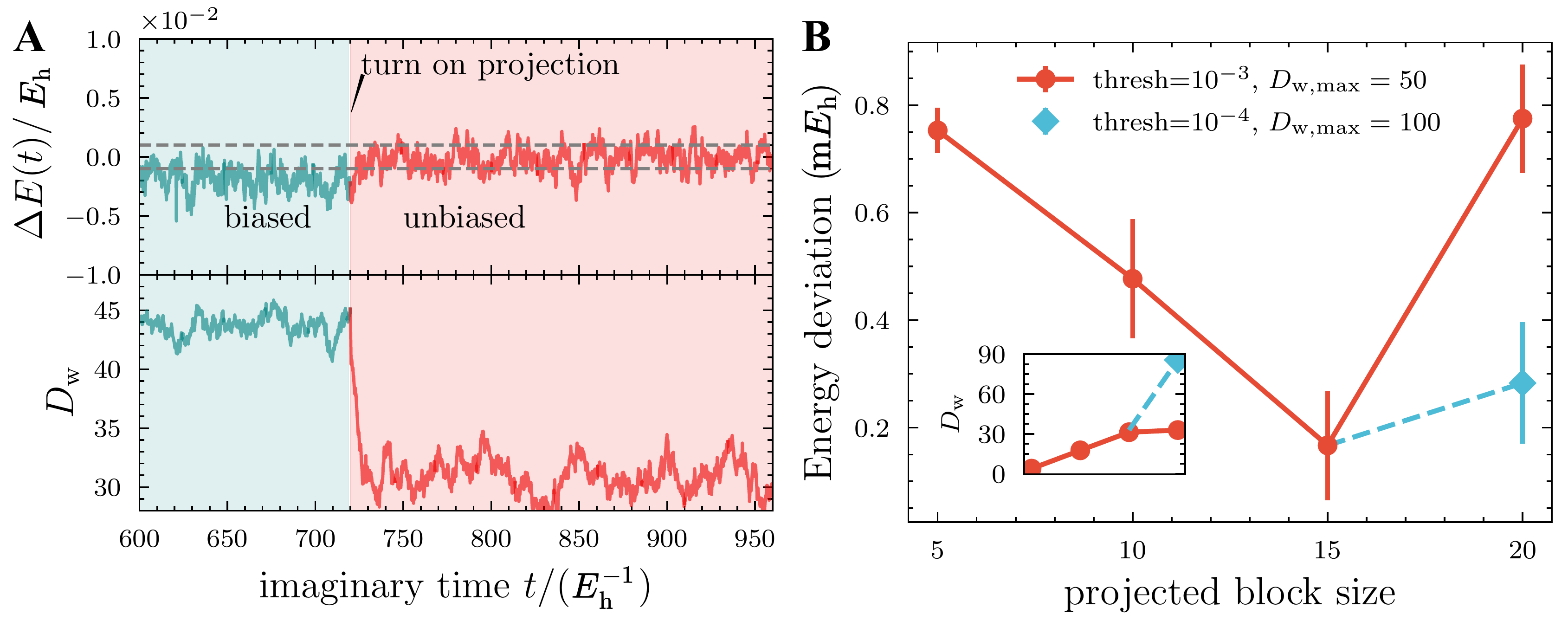}
    \caption{\textbf{Low entanglement projection.} 
    The system is H$_{\textrm{50}}$ with STO-6G basis and inter-atomic distance of $r=3a_0$. The upper panel of (A) shows the energy deviation from the exact answer during the imaginary time evolution and the lower panel of (A) gives the weight averaged bond dimension of walkers. The first stage (green shaded) corresponds to simulation with the adaptive threshold $\epsilon_{\textrm{ad}}=10^{-3}$ and maximally allowed walker's bond dimension $D_{\textrm{w}}=50$, and the second stage turns on the low entanglement projection with projected block size of 5. (B) The statistical energy deviation for choosing different projected block sizes with different adpative threshold. The inset show the weight-averaged bond dimension after performing low entanglement projection.}
    \label{fig:lesp}
\end{figure*}
A product state can be represented as an MPS with a bond dimension of 1, and the correlation matrix is characterized by having exactly $N_e$ diagonal elements set to one. These correspond to the occupied atomic orbitals in the product state. All other elements of the correlation matrix are zero.
As the randomness (Eq.~\eqref{eq:Bprop}\eqref{eq: 
propagate}) or off-diagonal correlation in a system increases, deviating from a purely diagonal form, and the entanglement of the system also increases. Consequently, a larger bond dimension is required to accurately represent the state as a MPS, as shown in
Fig.~\ref{fig:evolution2}. 
This increase in bond dimension allows the MPS to encapsulate the more intricate correlation patterns that arise in the system's state.
Consequently, compressing the walker wavefunction becomes more challenging, necessitating larger block sizes and a greater number of Givens gates. This, in turn, makes the compression of the MPS less effective. To address this challenge, we propose the subspace low entanglement projection during the time evolution, which can also speed up the conversion of SD-to-MPS.

Drawing inspiration from the method to approximate the Gaussian MPS of SD, the correlation matrix can be approximated as a band matrix. In this approximation, the matrix retains significant elements along the diagonal and in a finite band around it, effectively capturing the dominant correlations. This band matrix representation simplifies the structure while preserving the essential features of the correlation matrix, facilitating a more efficient and tractable representation in the context of MPS. This approach could be particularly useful when dealing with systems where the relevant correlations are localized or decay rapidly with distance.:
\begin{equation}
\Lambda'_{pq}=
\begin{cases}
\Lambda_{pq} & \text{if } p\leq q < p+\mathcal{B} \\
0 & \text{otherwise}
\end{cases}\label{eq:compress_Lambda}
\end{equation}
and the coefficient matrix of the projected walker $\psi'$ is obtained by the singular value decomposition of its correlation matrix
\begin{equation}
    \Lambda' = VSV^\dagger
\end{equation}
with $\mathcal{B}$ in Eq.~\eqref{eq:compress_Lambda} restrict the overlap $
\langle\psi|\psi'\rangle\ge1-\epsilon_{\textrm{proj}}$.
With a sufficiently small $\epsilon_{\textrm{proj}}$, $\Lambda'$ will exhibit $N_e$ singular values that are nearly equal to 1, and $N-N_e$ eigenvalues that are nearly equal to 0. 
We can express the approximated walker wavefunction, denoted as $U'$, with columns corresponding to those in $V$ with corresponding singular values close to 1. 
A phase introduced from the projection should be addressed using $\langle \psi|\psi'\rangle$.
The resulting $|\psi'\rangle$ should multiply the factor $e^{-i\theta}$ where 
\begin{equation}
\theta=\arg{(\langle\psi|\psi'\rangle)}    
\end{equation}
which is similar to what we did in Eq.~\eqref{eq:remove_phase}.

This procedure is referred to as low-entanglement subspace projection or low bandwidth projection of correlation matrix,
\begin{equation}
    |\psi'\rangle = \hat{P}|\psi\rangle
\end{equation}
which leads to a modification of Eq.~\eqref{eq: propagate},
\begin{equation}
\left|\psi(\tau+\Delta \tau)\right\rangle  =\hat{P}\hat{B}\left(\Delta \tau, \mathbf{x}-\overline{\mathbf{x}}\right)\hat{P}\left|\psi(\tau)\right\rangle\label{eq: proj_propagate}
\end{equation}
We present the schematics in Fig.~\ref{fig:proj_scheme}.
We apply the projection for every step and monitor the overlap of $\langle \psi|\psi'\rangle\ge 1-\epsilon_{\textrm{proj}}$, and it turns out the $\mathcal{B}$ will gradually decrease and saturate as a small block size.
The low entanglement subspace projection, as implied by its name, serves to reduce the entanglement of walker's SDs, leading to faster convergence with respect to $D_{\textrm{w}}$ and the adaptive threshold $\epsilon_{\textrm{ad}}$ in the GMPS approach.
The projection makes the SD-to-MPS easier.
Subsequently, we can obtain a more accurate MPS representation of the approximated walker wavefunction, effectively eliminating the errors associated with the conversion of the wavefunction to the GMPS format. 
As an example,  in Fig.~\ref{fig:H20}(B), we demonstrate the convergence of accuracy and the time cost concerning the bond dimension $D_{\textrm{w}}$ of equilibrated walkers for both the bipartite approach and the GMPS approach with different projected block size. 
Smaller block sizes greatly reduce computational time and improve accuracy.

Fig.~\ref{fig:lesp}(a) displays the energy deviation from the exact answer during the imaginary time evolution for H$_{50}$ with an interatomic distance of $r=3a_0$. The upper panel demonstrates the energy deviation, and the lower panel presents the corresponding bond dimension at each time. The first stage, marked with a green shade, corresponds to the simulation with the adaptive threshold $\epsilon_{\textrm{ad}}=10^{-3}$ and a fixed walker's bond dimension $D_{\textrm{w}}=50$. The second stage, marked with a red shade, corresponds to a projected block size equal to 5.
Without projection, there is a bias below the exact answer. Turning on the low entanglement projected technique with a block size of 5, as shown in Fig.\ref{fig:lesp}(a), the energy bias can be removed, and the final statistical error is below 1 $m\textrm{H}$, as demonstrated in Fig.\ref{fig:lesp}(b).
As the statistical evolution progresses, the walker wavefunction becomes more entangled, which makes converting to MPS more challenging. 
Increasing the projected block size for a given adaptive threshold leads to more accurate results. However, a larger block size also requires a tighter threshold and a larger bond dimension. For a given block size, tightening the adaptive threshold will also improve the accuracy of the results.

% \section{The timing comparison.}
\begin{figure}
    \centering
    \includegraphics[width=0.35\textwidth]{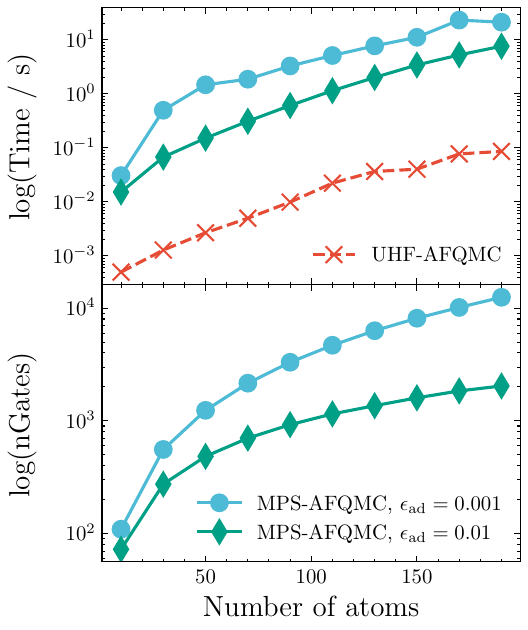} 
    \caption{\add{(a) The computational time 
    of MPS-AFQMC (1 walker's SD-to-MPS operation) and UHF-AFQMC (1 propagation step of 1 walker) for different system sizes.
    (b) The corresponding number of gates in the quantum circuit 
    for the SD-to-MPS operation using different adaptive threshold.
    The results are averaged over 3200 equilibrated walkers.}   
    }  
    \label{fig:scaling} 
\end{figure}    
\end{appendix}

\bibliography{references}
\end{document}